\documentclass[11pt,notitlepage]{article}
\usepackage{fullpage}
\usepackage{graphicx}
\usepackage{bbm}
\usepackage{authblk}
\usepackage[ruled]{algorithm2e}
\usepackage{epstopdf}
\usepackage{amsmath,amsbsy} 
\usepackage{amssymb}
\usepackage{float}
\usepackage{amsthm}
\usepackage{enumerate}
\usepackage{qcircuit}
\usepackage{complexity}
\usepackage[usenames,dvipsnames]{color}
\usepackage[normalem]{ulem}
\usepackage{tikz}
\usetikzlibrary{arrows}
\usepackage[all]{nowidow}
\usepackage[left=2.6cm, right=2.6cm, top=2.6cm,bottom=2.6cm]{geometry}
\usepackage{stmaryrd}
\usepackage[export]{adjustbox}
\usepackage[colorlinks=true,urlcolor=blue, citecolor=blue,linkcolor=blue]{hyperref}
\hypersetup{final}

\usepackage[utf8]{inputenc}
\usepackage[T1]{fontenc}
\usepackage{sidecap}

\usepackage[numbers,sort&compress]{natbib}
\usepackage[sectionbib]{bibunits}
\defaultbibliographystyle{plainnat}
\usepackage{enumerate}
\usepackage[shortlabels]{enumitem}
\usepackage{booktabs}
\usepackage{ragged2e}

\newcommand*{\boldplus}{\pmb{+}}
\newcommand*{\boldminus}{\pmb{-}}
\newcommand{\bra}[1]{\left\langle #1 \right|}
\newcommand{\ket}[1]{\left|#1\right\rangle}
\newcommand{\braket}[2]{\left\langle#1 |  #2\right\rangle}
\def\ketbra#1{ |{#1}\rangle\!\langle{#1}| }

\def\Erfc{\text{Erfc}}
\def\mcEbar{\lvert \mathcal{E} \rvert}
\DeclareMathOperator*{\EV}{\mathbb{E}}

\usepackage{ifdraft}
\ifdraft{\newcommand{\authnote}[3]{{\color{#3}({\bf  #1:} #2)}}}{\newcommand{\authnote}[3]{}}

\newtheorem{theorem}{Theorem}
\newtheorem{definition}{Definition}

\newtheorem{lemma}[theorem]{Lemma}

\newtheorem{proposition}{Proposition}

\newtheorem{condition}{Condition}

\title{Mind the gap: Achieving a super-Grover quantum speedup by jumping to the end }
\author[1,2]{Alexander M. Dalzell}
\author[1,2]{Nicola Pancotti}
\author[3,4]{Earl T. Campbell}
\author[1,5]{Fernando G.~S.~L.~Brand\~ao}

\affil[1]{AWS Center for Quantum Computing, Pasadena, CA, USA}
\affil[2]{California Institute of Technology, Pasadena, CA, USA}
\affil[3]{Riverlane, Cambridge, UK}
\affil[4]{Department of Physics and Astronomy, University of Sheffield, Sheffield S3 7RH, UK}
\affil[5]{Institute for Quantum Information and Matter, California Institute of Technology, Pasadena, CA, USA}

\date{}

\begin{document}

\maketitle
\thispagestyle{empty}

\begin{abstract}
    We present a quantum algorithm that has rigorous runtime guarantees for several families of binary optimization problems, including Quadratic Unconstrained Binary Optimization (QUBO), Ising spin glasses ($p$-spin model), and $k$-local constraint satisfaction problems ($k$-CSP). We show that either (a) the algorithm finds the optimal solution in time $O^*(2^{(0.5-c)n})$ for an $n$-independent constant $c$, a $2^{cn}$ advantage over Grover's algorithm; or (b) there are sufficiently many low-cost solutions such that classical random guessing produces a $(1-\eta)$ approximation to the optimal cost value in sub-exponential time for arbitrarily small choice of $\eta$. Additionally, we show that for a large fraction of random instances from the $k$-spin model, and for any sufficiently close-to-regular, fully satisfiable (or slightly frustrated) $k$-CSP formula, statement (a) is the case. 
    The algorithm and its analysis are largely inspired by Hastings' short-path algorithm [\href{https://doi.org/10.22331/q-2018-07-26-78}{\textit{Quantum} \textbf{2} (2018) 78}].
\end{abstract}

\tableofcontents

\thispagestyle{empty}

\newpage 

\section{Overview}

\pagenumbering{arabic}
\setcounter{page}{1}

\subsection{Motivation and the search for super-Grover speedups}

It is hoped that quantum computers will outperform classical computers at solving combinatorial optimization problems, both in theory and, eventually, in practice. A primary motivation for this expectation is the existence of Grover's quantum algorithm \cite{grover1996fast} for unstructured search: since classical algorithms for hard optimization problems often resort to some form of brute-force search, Grover's algorithm potentially offers a quadratic speedup without exploiting any particular structure in the problem. 

However, the hope to leverage Grover's algorithm for a practical advantage over state-of-the-art classical algorithms faces two pitfalls. First, it is rare that an exhaustive search through all possible solutions is the best classical algorithm for a specific combinatorial optimization problem. For example, the most well-known constraint satisfaction problem, 3-SAT\footnote{An instance of the 3-SAT problem is the question of whether a Boolean formula in conjunctive normal form (CNF)---i.e.~where each constraint is the conjunction of at most three of the $n$ binary variables---has a solution that satisfies all constraints.  Elsewhere in the paper, we refer to this problem as 3-CNF-SAT.}, admits a $2^{0.39n}$-time classical algorithm for instances with $n$ binary variables \cite{hertli2014-3SAT,dueholm2019fasterKSAT}, a nearly cubic speedup over a simple exhaustive search. These algorithms are still exponential time, and they typically retain elements of exhaustive search. However, they exploit structure to reduce the search space and to search more efficiently. 
Second, practical implementation of quantum algorithms with asymptotic quadratic speedup on actual quantum devices will suffer constant-factor slowdowns when compared to state-of-the-art classical hardware due to slower clock speeds, error-correction overheads, and general lack of parallelizability. When these factors are considered, realistic assessments of resources needed for quantum advantage using a quadratic speedup are pessimistic, suggesting that the breakeven point where quantum overtakes classical would occur only after many days, or in some cases, many years of runtime \cite{campbell2019applyingquantum,sanders2020compilationHeuristics,babbush2021focus}. 

The outlook for practical advantage dramatically improves as the power of the polynomial speedup becomes greater than quadratic. Ref.~\cite{babbush2021focus} found that an algorithm with quartic speedup (i.e.~the quantum runtime scales as $T^{1/4}$ where $T$ is the classical runtime) would offer a much more viable path to actual quantum advantage. This motivates the question of whether the quadratic Grover speedup can be surpassed for combinatorial optimization. 

In some cases, a super-quadratic speedup over exhaustive search can be realized by combining the quantum techniques of amplitude amplification \cite{brassard2002AmpAndEst} or quantum walk search \cite{szegedy2004quantumSpeedupMarkovChain,magniez2007SearchQuantumWalk,belovs2013quantumWalks}---both of which can be understood as generalizing Grover's algorithm---with classical techniques for exploiting problem structure such as backtracking \cite{montanaro2018backtracking,ambainis2017treeSizeEstimation,jarret2018improvedBacktracking,martiel2020practicalBacktracking}, branch-and-bound \cite{montanaro2020branchAndBound}, nested search \cite{cerf2000nestedSearch}, dynamic programming \cite{ambainis2019dynamicProgramming}, and Markov chain Monte Carlo methods \cite{somma2008annealing,wocjan2008speedupQSampling, montanaro2015MonteCarlo}. Establishing that these classical techniques can be employed while retaining the quadratic quantum speedup is typically non-trivial; however, these ingredients alone can only \textit{restore} the quadratic speedup compared to classical state-of-the-art methods, not surpass it, since the fundamentally quantum part of the algorithm generates only a Grover-like speedup.

Progress toward proving a genuine \textit{super-Grover} speedup for combinatorial optimization is limited. The quantum adiabatic algorithm (QAA) \cite{farhi2000quantum, lidar2018AQCreview} is one example of an algorithm that could conceivably generate super-Grover speedups for some problems. However, its runtime is notoriously hard to study due to dependence on the minimal spectral gap of certain non-commuting Hamiltonians. Additionally, some prior work has argued that the QAA can in some cases exhibit runtime that scales super-exponentially---even worse than exhaustive search---when applied to combinatorial optimization problems \cite{altshuler2010AndersonLocalization,wecker2016trainingQuantumOptimizer}, although for random instances it is expected to scale merely exponentially \cite{knysh2010relevance,young2010firstOrderAdiabatic,hen2011exponentialAdiabatic}. Many other works (see, e.g., Ref.~\cite{zhou2020qaoaPRX}), often with a focus on near-term algorithms, aim only to solve combinatorial optimization problems \emph{approximately} rather than exactly (and typically in only polynomial time, rather than exponential time); these algorithms have the potential to be practically useful, but rigorous guarantees of quantum advantage are difficult to establish.

In this work, we focus on \emph{exact} binary optimization problems where the goal is to find the assignment $z^*$ that minimizes some efficiently computable, classical cost function $H(z)$, with $z \in \{+1,-1\}^n$. We desire not only the best solution $z^*$, but also a high degree of confidence that no better solution exists; thus, we generally expect that any algorithm that solves this problem, quantum or classical, requires exponential runtime. For exponential-time algorithms, we follow the convention of writing $O^*(2^{an})$ to mean that there is an upper bound on the runtime of the form $h(n)2^{an}$ that holds for sufficiently large $n$, where $h$ is a polynomial function. Exhaustive enumeration runs in (classical) time $O^*(2^n)$. Meanwhile, Grover's algorithm runs in (quantum) time $O^*(2^{0.5n})$. We seek a ``super-Grover speedup,'' that is, a quantum algorithm that runs in time $O^*(2^{(0.5-c)n})$ for some $n$-independent constant $c$, where the additional advantage over Grover's algorithm is generated by a fundamentally quantum mechanism.

Obtaining speedups over Grover's algorithm was also the goal of a series of works by Hastings \cite{hastings2018shortPath,hastings2018weaker,Hastings2019shortPathToyModel} on an algorithm called the \emph{short-path algorithm}. Whereas Grover's algorithm begins in the equal superposition state $\ket{\boldplus} \equiv 2^{-n/2}\sum_{j} \ket{j}$ and measures the classical function $H(z)$ whilst performing amplitude amplification to boost the probability of obtaining an optimal-cost measurement outcome, the short-path algorithm instead uses quantum phase estimation to measure $H + V$, where in this context $H$ denotes the diagonal operator corresponding to the cost function and $V$ is an off-diagonal non-commuting perturbation term. Here again, amplitude amplification is used to boost the probability of ending in the ground state of $H+V$. After the ground state is prepared, a computational basis measurement is performed to find the optimal solution. It was shown that for certain choice of $V$ and under certain assumptions on the spectral density of $H$ at low energies (i.e.~low cost values), a super-Grover speedup could be accomplished. However, no concrete scenario was given where the spectral density assumption was proved to hold.  In follow-up work \cite{hastings2018weaker}, Hastings showed that for a specific family of cost functions with 2-local Ising-like terms, the spectral density assumption could be dropped, but in this case the super-Grover speedup constant $c$ was not $n$-independent, decaying like $1/\log(n)$ or faster (depending on the amount of frustration in the cost function).

\subsection{A simple quantum algorithm for exact combinatorial optimization}

Our main contribution is a simple algorithm inspired by Hastings' short-path algorithm \cite{hastings2018shortPath} for which we can prove a super-Grover speedup for a few specific families of cost functions, including Ising spin glasses and $k$-local constraint satisfaction problems. The algorithm can be run on other families of cost functions and may very well have a super-Grover speedup more generally, but we present rigorous guarantees only for these specific cases.

The algorithm can be understood as an implementation of the QAA for combinatorial optimization problems with two crucial modifications, as we now explain. For simplicity, suppose here that $z^*$ is \emph{unique}, i.e.~$H(z^*) < H(z)$ for all $z\neq z^*$, and suppose that the optimal value $E^* = H(z^*)$ is known ahead of time (these assumptions will be dropped later). By convention, we assume that $H$ is offset such that the average cost of a uniformly random input is zero, and thus $E^* < 0$.  Let $X = \sum_{i=1}^n X_i$ be the transverse-field operator, where $X_i$ denotes the Pauli-$X$ operator on qubit $i$. Then, the QAA evolves by the Hamiltonian 
\begin{equation}\label{eq:Hb_adiabatic}
    H_b = -\frac{X}{n} + b\;\frac{H}{|E^*|} \qquad \text{(for the QAA)}
\end{equation}
while the parameter $b$ is slowly tuned from $b=0$, where the ground state of $H_b$ is $\ket{\boldplus}\equiv \ket{+}^{\otimes n}$, to $b=\infty$, where the ground state of $H_b$ is $\ket{z^*}$. The first modification is that $\frac{H}{|E^*|}$ is replaced by $g(\frac{H}{|E^*|})$ for a certain piecewise linear function $g: [-1,\infty) \rightarrow [-1,0]$:
\begin{equation}\label{eq:Hb_intro}
    H_b = -\frac{X}{n} + b\;g\left(\frac{H}{|E^*|}\right) \qquad \text{(for our algorithm)}
\end{equation}
This modification allows us to control the spectral properties of the cost function to enable proof of our claims. The second modification is, rather than evolve continuously through values of $b$ for which the spectral gap of $H_b$ is small and unknown, the algorithm simply \emph{jumps}, first from $b=0$ to an $n$-independent value of $b>0$ where the gap is guaranteed to be large, and then from that value of $b$ all the way to the end of the algorithm ($b=\infty$). These jumps are accomplished with quantum phase estimation along with amplitude amplification to boost the success probability of projecting onto the ground state of $H_b$. The first jump is small in the sense that the success probability is nearly one (and little or no amplitude amplification is required). The second jump is large in the sense that the success probability is exponentially small. We will be able to show that (for some specific families of cost functions) the success probability of the second jump is larger than $2^{-(1-2c)n}$, and hence, after amplitude amplification, the runtime of the algorithm is less than $O^*(2^{(0.5-c)n})$, where $c$ is a constant independent of $n$. 

Both modifications described above, as well as the analysis of the algorithm, are inspired by (and in some cases closely follow) Hastings' short-path algorithm \cite{hastings2018shortPath}. However, the algorithms are not identical, and in some sense they are dual to each other: where our algorithm makes a small jump and then a large jump, the short-path algorithm instead makes a large jump and then a small jump. Moreover, where our algorithm makes the modification $H/|E^*| \rightarrow g\left(H/|E^*| \right)$ for a piecewise linear function $g$ (see Eqs.~\eqref{eq:Hb_adiabatic} and \eqref{eq:Hb_intro}), the short-path algorithm makes the modification $-X/n \rightarrow g\left(-X/n\right)$ for the function $g(x) = x^K$, with $K$ an odd integer. 

By taking our approach, we can prove more concrete results than was possible to show for the original short-path algorithm. This is true for subtle technical reasons. Firstly, as in Ref.~\cite{hastings2018shortPath}, our proof utilizes the log-Sobolev inequality, which relates the expectation value of the $X/n$ operator for a given state $\ket{\phi}$ to the entropy of the distribution over measurement outcomes that arises when $\ket{\phi}$ is measured in the computational basis. This step goes through more cleanly for us because the $X/n$ operator appears directly in our Hamiltonian (see Eq.~\eqref{eq:Hb_intro}), whereas $(X/n)^K$ appeared in Hastings' Hamiltonian, requiring some work to argue that the expectation of $(X/n)^K$ could be approximated by the $K$th power of the expectation of $X/n$. Secondly, by switching from Hastings' strategy of a large-then-small jump to our strategy of a small-then-large jump, (roughly speaking) Hastings' spectral density assumption switches from a statement about the number of states at low energy to a statement about the number of states at high energy, which is much easier to show using straightforward tail bounds. This switch does come at a cost: we lose Hastings' perturbation theory argument that proves the exponential advantage over Grover, assuming the spectral density assumption.  We are able to replace this part of the argument in our case with a proof that utilizes an approximate ground state projector.

\subsection{Overview of provable statements}\label{sec:overview_provable}

\noindent Consider the following families of optimization problems over inputs in the set $\{+1,-1\}^n$.
\begin{itemize}
\item \textbf{MAX-E$k$-LIN2}: $H(z) = p(z_1,\ldots,z_n)$, where $p$ is a polynomial consisting only of monomials of degree $k$. 
\item \textbf{Quadratic Unconstrained Binary Optimization (QUBO)}: $H(z) = p(z_1,\ldots,z_n)$, where $p$ is a polynomial consisting only of monomials of degree-1 or degree-2.
\item \textbf{MAX-$k$-CSP with limited frustration}: $H(z) = \sum_{j=1}^m \mathcal{C}_j$ where each $\mathcal{C}_j$ is a $k$-constraint, defined by the criteria that (i) it is a function of at most $k$ of the $n$ bits of the input $z$, and (ii) it takes the value $-1$ on $s_j \in [1,2^k-1]$ (``satisfying'') assignments to those bits and value $s_j/(2^k-s_j)$ on the other $2^k-s_j$ (``unsatisfying'') assignments, such that the average across all $2^k$ assignments is zero. Let $E^*$ be the optimal value of $H$.  The magnitude of the super-Grover speedup will depend on $\lvert E^*\rvert / m$ and is maximal when $\lvert E^*\rvert / m=1$, i.e.~when the instance is frustration free (all constraints are simultaneously satisfiable). The speedup guarantee also depends on the parameter $D = nk^{-2}m^{-2}\sum_{j=1}^n d_j^2$, where $d_j$ is the number of constraints in which the bit $z_j$ participates.\footnote{\label{footnote:error}In the \href{https://arxiv.org/abs/2212.01513v1}{first version} of this paper and in the version published in the \href{https://doi.org/10.1145/3564246.3585203}{STOC'23 proceedings}, the proof of Theorem 7 attempted to argue that solving MAX-$k$-CSP for arbitrary instances can essentially be reduced to solving instances with bounded value of $D \leq O(1)$. This argument allowed us to state our main results for MAX-$k$-CSP without mentioning $D$. We thank Fran\c{c}ois Le Gall for pointing out an error in this argument. This version fixes the error by reporting the $D$-dependence of the algorithm when it is applied to MAX-$k$-CSP. The results for MAX-E$k$-LIN2 and QUBO are not affected. } The parameter $D$ measures, in a sense, how far from regular the interaction hypergraph of the CSP instance is. That is, $D$ achieves its minimum of 1 when each bit participates in exactly $mk/n$ constraints.
\end{itemize}

We also study a random ensemble of MAX-E$k$-LIN2 instances known as the ``$p$-spin'' model \cite{derrida1980random}, which here we refer to as the ``$k$-spin'' model so that the symbol $k$ consistently represents the locality of the terms of $H$ throughout the paper. An instance of the $k$-spin model is given by
\begin{equation}\label{eq:k-spin}
    H(z_1,\ldots,z_n) = \sqrt{\frac{k!}{n^{k-1}}} \sum_{1\leq i_1<\ldots < i_k \leq n} J_{i_1,\ldots,i_k}z_{i_1}z_{i_2}\ldots z_{i_k}\,,
\end{equation}
where the weights $J_{i_1,\ldots,i_k}$ are each chosen independently at random from a standard Gaussian distribution with mean 0 and variance 1. The $k$-spin model is used in physics to model spin glasses. When $k=2$ it is identical to the Sherrington-Kirkpatrick (SK) model \cite{sherrington1975solvable}.

Before we state our main results, we define our usage of big-$O$, big-$\Omega$ notation, which appears in various places throughout the paper.  Let $h$ be a function of $n$, $k$, and $\frac{m}{|E^*|}$. Then a quantity is said to be $O(h(\cdot))$ (resp.~$\Omega(h(\cdot))$) if it is upper bounded (resp.~lower bounded) by some constant times $h(\cdot)$. We typically think of $k$ as a constant as $n$ grows, but we include the $k$-dependence in our expressions (rather than absorbing it into big-$O$) to communicate how the size of the speedup depends on $k$. 

Our main results are the following:
\begin{itemize}
    \item For each instance of MAX-E$k$-LIN2 and QUBO, either there is a quantum algorithm with super-Grover speedup, or there is a classical algorithm that can achieve an arbitrarily good approximation ratio in sub-exponential time. More precisely, for any $\eta \in [0,1]$ and any $\gamma \in [0,1]$, either (a) the quantum algorithm has runtime $2^{0.5(1-\Omega(\gamma\eta/k))n}$, or (b) the classical algorithm that repeatedly samples assignments uniformly at random produces a bit string $y$ with cost $H(y) \leq (1-\eta)E^*$ within time $O^*(2^{\gamma n})$ (with high probability). Thus, if there is no $n$-independent choice of $\eta,\gamma$ for which (a) holds, then for arbitrarily small $\eta$, the runtime of the classical algorithm is better than $O^*(2^{\gamma n})$ for arbitrarily small $\gamma$ (i.e.~sub-exponential time).
    \item For the $k$-spin model, there is a choice of $\eta = \Omega(1)$ and $\gamma = \Omega(1/k^2)$ such that case (a) is satisfied in the previous bullet for all but at most a $e^{-\Omega(n/k)}$ fraction of instances.
    \item For instances of MAX-$k$-CSP where $\lvert E^*\rvert/m$ is independent of $n$ and irregularity parameter $D = O(1)$, the quantum algorithm has a super-Grover speedup. More precisely, there exists a constant $c = \Omega\left(\frac{\lvert E^*\rvert^3}{k^32^{3k}m^3 D}\right)$ such that MAX-$k$-CSP admits a quantum algorithm with runtime $O^*(2^{(0.5-c)n})$. The algorithm can thus solve the question of whether or not a $k$-CSP instance is fully satisfiable in $O^*(2^{(0.5-\Omega(1/k^32^{3k}D))n})$ time. When $|E^*|/m$ is independent of $n$ but $D$ is growing with $n$, we may instead make a statement similar to the first bullet---either a super-Grover speedup is available (i.e., $c$ is lower bounded by an $n$-independent quantity), or else a classical algorithm can find an arbitrarily good approximate solution in sub-exponential time. 
\end{itemize}

These results are stated in more formal language later in Sec.~\ref{sec:formal_results}. Note that for random MAX-$k$-CNF-SAT, the random ensemble of MAX-$k$-CSP instances where each clause is the conjunction of $k$ distinct randomly chosen variables (or their negations), we have that $D = O(1)$ for typical instances, and meanwhile the amount of frustration increases with the number of clauses. If $m \leq \alpha_c n$ for some critical value $\alpha_c$ (which grows like $2^k$), then most instances are fully satisfiable  \cite{achlioptas2003thresholdkSAT}, which means that $|E^*|/m = 1$. If $m = \alpha n$ for $\alpha > \alpha_c$, then most instances are not fully satisfiable (``frustrated''), but for typical instances the ratio $|E^*|/m$ is larger than some $n$-independent number \cite{achlioptas2007randomMaxSat}. Thus, the super-Grover speedup persists despite the frustration, although the constant $c$ indicating the size of the $2^{cn}$ advantage over Grover is reduced compared to the frustration-free case. If $m/n\rightarrow \infty$ as $n \rightarrow \infty$, then  $|E^*|/m \rightarrow 0$ for typical instances in that limit (in particular, it is expected to decay as $\sqrt{n/m}$) \cite{coppersmith2003random}, and our method fails to prove a super-Grover speedup.

\subsection{Comparison to classical algorithms and significance}\label{sec:comparison}

In Tab.~\ref{tab:comparison_with_classical}, we present several problems for which we can make concrete statements about the runtime of our algorithm, and we compare with the corresponding best known classical runtime. While our algorithm has a super-\textit{Grover} speedup for the 3-CNF-SAT problem (i.e.~the question of whether there is a fully satisfying assignment for a 3-CSP instance where each clause is in conjunctive normal form) on instances where $D = O(1)$, it does not have a super-\textit{quadratic} speedup---in fact, it has no speedup at all---owing to the fact that there are also classical algorithms with a significant speedup over exhaustive enumeration. At large $k$, our algorithm does give a speedup for $k$-CNF-SAT on instances when $D=O(1)$,\footnote{Note that the classical algorithm being compared against works for any $D$.} but the speedup is sub-quadratic. However, it is important to note that these classical algorithms for $k$-CNF-SAT are the product of decades of incremental improvements: for $k=3$, Monien and Speckenmeyer found a $O^*(2^{0.70n})$-time algorithm in 1985 \cite{monien1985solvingSatisfiability} and the coefficient was subsequently reduced to 0.57, 0.45, 0.42, and 0.39 in Refs.~\cite{rodovsek1996new}, \cite{paturi2005PPSZalgorithm}, \cite{schoning2002probabilistic}, and \cite{hertli2014-3SAT}, respectively, with further infinitesimal improvements in Refs.~\cite{scheder2017simplerPPSZ} and \cite{dueholm2019fasterKSAT}. For the SK model, our algorithm also fails to deliver a speedup due to the existence of a classical branch-and-bound algorithm \cite{montanaro2020branchAndBound} which gives a significant provable advantage over exhaustive search. It is worth noting that the classical algorithms for $k$-CNF-SAT in Ref.~\cite{dueholm2019fasterKSAT} and for the SK model in Ref.~\cite{montanaro2020branchAndBound} each admit a quadratic quantum speedup by applying amplitude amplification or quantum walk search techniques. The resulting quantum algorithm is the best-known quantum algorithm in both cases.  

\begin{table}[h]
    \centering
    \begin{tabular}{|p{6em}|p{10em}|p{9em}|}
         \hline
         \textbf{Problem} & \textbf{Our quantum algo} & \textbf{Best classical algo} \\
         \hline
         3-CNF-SAT & $0.5 - (1.7 \times 10^{-6})D^{-1}$ & $0.39$ \hspace{36 pt} \cite{dueholm2019fasterKSAT} \\
         $k$-CNF-SAT & $0.5 - \Omega(2^{-3k}k^{-3}D^{-1})$ & $1-\Omega(k^{-1})$ \hspace{2.5pt} \cite{dueholm2019fasterKSAT} \\
         SK model & $0.5-(2.7 \times 10^{-5})$ & $0.45$ \hspace{30 pt} \cite{montanaro2020branchAndBound} \\
         $k$-spin & $0.5-\Omega(k^{-3})$ & $1$ \\
         \hline
    \end{tabular}
    \caption{Summary of concrete combinatorial optimization problems where we can prove an upper bound on the runtime of our algorithm, in comparison to the best known classical algorithm for the same problem. The number displayed is the coefficient of $n$ in the exponential, i.e.~if the algorithm runs in time $O^*(2^{an})$ then $a$ appears in the table. The problem $k$-CNF-SAT refers to the question of whether or not a Boolean formula in conjunctive normal form with $k$-local constraints is fully satisfiable. The $k$-spin model refers to the random ensemble of MAX-E$k$-LIN2 instances where every $k$-local term appears with a random Gaussian weight, defined in Eq.~\eqref{eq:k-spin}. The Sherrington-Kirkpatrick (SK) model corresponds to the $k$-spin model with $k=2$. The irregularity parameter $D$ takes a value $D=1$ when the interaction hypergraph is regular, and $D \leq 3$ holds for typical randomly chosen $k$-CNF-SAT instances.} 
    \label{tab:comparison_with_classical}
\end{table}

For the $k$-spin model with $k \geq 3$, the branch-and-bound technique of Ref.~\cite{montanaro2020branchAndBound} does not obviously generalize, and we do not know of a classical algorithm that has been proved to run in time $O^*(2^{(1-c)n})$ for an $n$-independent value $c$. In contrast, our algorithm runs in time $O^*(2^{(0.5-\Omega(1/k^3))n})$, a potential super-quadratic speedup.  However, we believe that it is plausible that there does exist a $O^*(2^{(1-c)n})$-time classical algorithm for $k$-spin: a potential candidate is classical Metropolis sampling at high temperature. Evidence that this algorithm would be effective comes from prior work on the \textit{spherical} $k$-spin model, a continuous variable analogue of the ``Ising'' $k$-spin model we study here. In particular, Ref.~\cite{gheissari2019sphericalSpinGlass} showed that for the spherical $k$-spin model, the Langevin dynamics are rapidly mixing when the temperature is above some $n$-independent threshold, allowing efficient classical sampling from the Gibbs distribution. Since the $k$-spin model is normalized such that typical instances have an extensive optimal cost value $|E^*| = \Omega(n)$, the Boltzmann factor $e^{-\beta E^*}$ for the optimal assignment $z^*$ at constant inverse temperature $\beta$ is exponentially large in $n$. This fact suggests that there exists an $n$-independent choice of $c$ for which only $2^{(1-c)n}$ Gibbs samples need to be drawn to find $z^*$ with high probability. It would be interesting to extend these results from the spherical to the Ising $k$-spin model and formally verify that it leads to a speedup over exhaustive enumeration. 

In all cases, it is apparent that the provable advantage of our algorithm over $O^*(2^{0.5n})$ is very small. Our goal has been to prove that there exists \textit{some} constant improvement over Grover's algorithm, and we have not dedicated much effort to optimizing the proofs to maximize the constant. We are certain that the size of the provable speedup could be improved, and furthermore in Sec.~\ref{sec:numerics}, we give numerical evidence that the speedup over Grover is much more substantial than the proofs imply. 

We emphasize that a key reason this speedup is interesting despite its small numerical size is that the speedup mechanism has no immediate classical analogue with comparable runtime guarantees upon which our algorithm simply applies a general technique like amplitude amplification. We connect the speedup mechanism to an observation about 1-norm vs.~2-norm localization, a feature shared by Hastings' short-path algorithm \cite{hastings2018weaker}. Namely, it is possible for a wavefunction to be localized on a single basis state when using the 2-norm, yet de-localized across many basis states when using the 1-norm, a situation with no classical analogue. See Sec.~\ref{sec:speedup_mechanism} for a more detailed discussion. 
With further innovations, future algorithms may be able to leverage this phenomenon for more substantial speedups.

\section{Algorithm}\label{sec:algorithm}

\subsection{Enacting jumps from one ground state to another}
The algorithm we present in this work is conceptually simple; it consists essentially of just two steps, each of which is a jump from the ground state of one Hamiltonian to the ground state of another Hamiltonian. The time it takes to perform each jump is related to the overlap of the two ground states and the spectral gap of the two Hamiltonians. This is the same primitive step that was used in Hastings' short-path algorithm \cite{hastings2018shortPath}, as well as various other quantum algorithms before it (e.g.~\cite{boixo2010fast}). 

\begin{proposition}
[jump from $K_1 \rightarrow K_2$, simplified]
\label{prop:jump_simplified}
    Given two $n$-qubit Hamiltonians $K_1$ and $K_2$, let $\ket{\psi_1}$ be the (unique) ground state of $K_1$ and $\Pi_2$ be the projector onto the (possibly degenerate) ground space of $K_2$. Let $\Delta_1$ and $\Delta_2$ denote the spectral gap above the ground space for $K_1$ and $K_2$. Then there is a unitary $U$ for which $U\ket{\psi_1}\propto \Pi_2\ket{\psi_1}$ that is enacted up to error $\delta$ by a quantum circuit consisting of $[\min(\Delta_1,\Delta_2)]^{-1}\lVert \Pi_2 \ket{\psi_1} \rVert^{-1}\poly(n,\log(\delta^{-1}))$ gates, where $\lVert \cdot \rVert$ denotes the standard Euclidean norm for a vector. 
    
    Additionally, if $K_1$ or $K_2$ is a classical Hamiltonian (i.e.~diagonal in either the computational basis or the Hadamard basis, where diagonal entries can be efficiently classically computed), then the number of gates does not depend on the corresponding gap parameter $\Delta_1$ or $\Delta_2$, respectively. 
\end{proposition}

To actually construct the unitary, we need to have knowledge of a lower bound on $\lVert \Pi_2 \ket{\psi_1}\rVert$, upper bounds on the ground state energy of $K_1$ and $K_2$, and lower bounds on the excited energy of $K_1$ and $K_2$. A more complete version of the proposition that considers these factors appears as Prop.~\ref{prop:jumps} in App.~\ref{app:jumps}, along with its proof. There we describe the unitary $U$, which is constructed through two steps. First, it uses phase estimation to produce unitary operators $R_1$ and $R_2$ that reflect about the state $\ket{\psi_1}$ and about the ground space of $K_2$, respectively. The gate cost of approximating $R_j$ to error $\delta$ is $O(\Delta_j^{-1}\log(\delta^{-1}))$ calls to a so-called ``block-encoding'' of the Hamiltonian $K_j$, which typically requires just $\poly(n)$ gates. The exception to this statement is the case where $K_j$ is a classical Hamiltonian. In this case, the terms of $K_j$  commute and the energy can be measured exactly, allowing the reflection operator to be implemented exactly in $\poly(n)$ gates regardless of how small $\Delta_j$ is. Second, the unitary $U$ performs fixed-point amplitude amplification \cite{yoder2014fixedpoint} to produce the state $\frac{\Pi_2 \ket{\psi_1}}{\lVert \Pi_2 \ket{\psi_1}\rVert}$ using the reflection operators $R_1$ and $R_2$ $O(\lVert \Pi_2 \ket{\psi_1} \rVert^{-1}\log(\delta^{-1}))$ times each.

\subsection{Specification of algorithm}
Now we specify the main algorithm. 
The inputs to the algorithm are as follows:

\vspace{12 pt}

\noindent \textbf{Inputs}:
\begin{enumerate}[(a)]
    \item A classical cost function $H$ on $n$-bit binary assignments $z \in \{+1,-1\}^n$. This may be specified, for example, by giving the coefficients of $H(z)$ when it is expanded as a polynomial in $z_1,\ldots, z_n$, where $z_i \in \{+1,-1\}$ denotes the $i$th bit of $z$. By convention we offset $H$ so that it has no constant term, i.e.~$\sum_z H(z) = 0$. In any case, we assume that for any $z$, $H(z)$ can be evaluated classically in $\poly(n)$ time. 
    \item The value $E^* = \min_z H(z)$. 
    \item A value for $\eta$ satisfying $0 \leq \eta < 1$.
    \item A value for $b$ satisfying $0 \leq b < 1$. 
\end{enumerate}
\noindent Note that here we assume the optimal value $E^*$ of the cost function is known ahead of time, which may not always be the case. In App.~\ref{app:unknown_Estar}, we discuss how this assumption can be dropped at the expense of only polynomial overheads.

We define the piecewise-linear function $g_\eta:[-1,\infty)\rightarrow [-1,0]$, and provide a plot for convenience:
\begin{equation}\label{eq:g_eta}
    g_\eta(x) = \min\left(0,\frac{x+1-\eta}{\eta}\right) \qquad \qquad \includegraphics[width=0.3\textwidth,valign=c,draft=false]{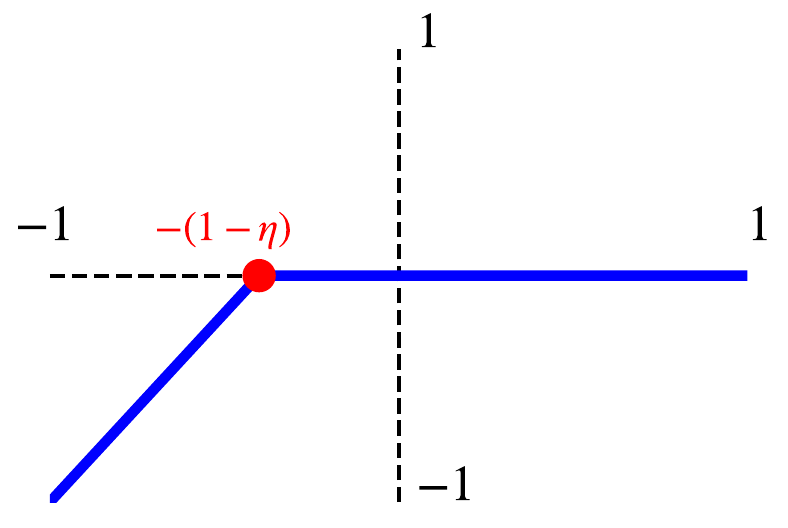}\,.
\end{equation}
We also define the Hamiltonian $H_b$, parameterized by $b>0$, as
\begin{equation}\label{eq:Hb}
\begin{split}
    H_{b} &= -\frac{X}{n} + b\; g_\eta\left(\frac{H}{\lvert E^* \rvert}\right) \\
    \ket{\psi_b} &= \text{ground state of } H_b \\
    E_b &= \text{ground state energy of } H_b\,,
\end{split}
\end{equation}
where here (in a slight abuse of notation) $H$ denotes the diagonal $2^n \times 2^n$ Hermitian operator for which $\bra{z}H\ket{z} = H(z)$ and $X$ denotes the transverse field $\sum_{i=1}^n X_i$, with $X_i$ the Pauli-$X$ operator on qubit $i$. Let $\ket{\psi_b}$ denote the ground state of $H_b$ and $E_b$ the ground-state energy. For illustration purposes, the spectrum of $H_b$ for an example $n=20$, $\eta=0.5$ instance is shown in Fig.~\ref{fig:spectrum_n20}. Note that $H_b$ is a \emph{stoquastic} Hamiltonian, that is, when written in the computational basis, all of the off-diagonal entries of $H_b$ are non-positive. As a consequence of being stoquastic, the ground state $\ket{\psi_b}$ can be taken to have non-negative real entries in the computational basis \cite{bravyi2006complexity}, and we assume this convention throughout. 

Let $\Pi^*$ denote the projector onto the (potentially degenerate) groundspace of $H$ (spanned by computational basis states). The algorithm begins by preparing the initial state $\ket{\boldplus} \equiv \ket{+}^{\otimes n}$, the ground state of the Hamiltonian $-X/n$ (which is diagonal in the Hadamard basis). Next, the algorithm prepares $\ket{\psi_b}$, the (unique) ground state of $H_b$ by performing a jump from Hamiltonian $-X/n$ to Hamiltonian $H_b$, using the unitary described in Prop.~\ref{prop:jump_simplified}.\footnote{One might be concerned about how the appearance of the operator $b\,g_\eta(H/|E^*|)$ might affect the construction of this unitary. Indeed, even if $H$ can be decomposed as a sum of $\poly(n)$ Pauli strings, $g_\eta(H/|E^*|)$ will not have this property. Nevertheless, in App.~\ref{app:jumps}, we show how a ``block-encoding'' of $g_\eta(H/|E^*|)$, which is diagonal, can be constructed by explicilty computing $H(z)$ into an ancilla register, computing $\theta_z = \arcsin(g_\eta(H(z)/|E^*|))$ with classical arithmetic in superposition, and subsequently performing a rotation by angle $\theta_z$. This only works because $H$ is a classical operator.} Finally, the algorithm prepares $\frac{\Pi^*\ket{\psi_b}}{\lVert \Pi^* \ket{\psi_b}\rVert}$ using a second jump, from Hamiltonian $H_b$ to the (classical) Hamiltonian $H/|E^*|$. The state $\Pi^*\ket{\psi_b}$ is a superposition of optimal solutions to the cost function $H$; one of these solutions can be retrieved by measurement in the computational basis. Pseudocode for the algorithm appears in Algorithm \ref{algo:main_simple}, where a ``jump'' from Hamiltonian $K_1 \rightarrow K_2$ refers to the procedure in Prop.~\ref{prop:jump_simplified}. 

\SetKwInput{Input}{Input}
\SetKwInput{Output}{Output}
\SetKwInput{Set}{Set}
\SetKwFunction{ApprSolve}{ApprSolve}
\SetStartEndCondition{ }{}{}%
\SetKwProg{Fn}{def}{\string:}{}
\SetKwFunction{Range}{range}
\SetKw{KwTo}{in}\SetKwFor{For}{for}{\string:}{}%
\SetKwIF{If}{ElseIf}{Else}{if}{:}{elif}{else:}{}%
\SetKwFor{While}{while}{:}{fintq}%
\newcommand{\forcond}{$i$ \KwTo\Range{$n$}}
\AlgoDontDisplayBlockMarkers\SetAlgoNoEnd\SetAlgoNoLine%
\begin{algorithm}
\caption{Pseudocode for main algorithm \label{algo:main_simple}}
\DontPrintSemicolon
\SetAlgoLined
\LinesNumbered
\Input{$H$, $E^*$, $\eta$, $b$, which together define $H_b$ in Eq.~\eqref{eq:Hb}}
\Output{an optimal assignment $z^*$ for $H$}
\BlankLine
Prepare $\ket{\boldplus} \equiv \ket{+}^{\otimes n} \equiv \sum_{i=1}^{2^n} \ket{i}$, the ground state of $-\frac{X}{n}$ \;
Prepare $\ket{\psi_b}$ up to exponentially small error with jump $-\frac{X}{n}\rightarrow H_b$ \;
Prepare $\frac{\Pi^* \ket{\psi_b}}{\lVert \Pi^* \ket{\psi_b} \rVert}$ up to exponentially small error with jump $H_b \rightarrow \frac{H}{|E^*|}$ \;
Measure in the computational basis to produce $\ket{z^*}$
\end{algorithm}

\subsection{Condition for success and overall runtime}\label{sec:conditions}

Note that implementing the jumps in steps 2 and 3 of the algorithm requires an upper bound on $E_b$, a lower bound on the excited energy of $H_b$, and lower bounds on the quantities $\lvert \braket{\boldplus}{\psi_b}\rvert$ and $\lVert \Pi^* \ket{\psi_b}\rVert$, where here $\ket{\boldplus} \equiv \ket{+}^{\otimes n}$ (see Prop.~\ref{prop:jumps} in App.~\ref{app:jumps} for a formal presentation of jump implementation). The bounds on $\lvert \braket{\boldplus}{\psi_b}\rvert$ and $\lVert \Pi^* \ket{\psi_b}\rVert$ can be ``guessed'' in the sense that one can try one value, see if the algorithm succeeds (it is easy to check if the output $y$ of the algorithm satisfies $H(y) = E^*$), and if not, repeat with a guess that is smaller by a fixed constant factor. This procedure contributes at most $\poly(n)$ overhead compared to if these quantities were known ahead of time. On the other hand, bounds on the eigenenergies of $H_b$ must be shown separately, and are not guaranteed to hold for every possible cost function $H$ and choice of parameters $b$, $\eta$. Accordingly, we define a condition under which success can be proved (and later we prove that the condition holds in specific cases).

\begin{condition}[large-excited-energy condition]\label{cond:large-excited-energy}
We say the Hamiltonian $H_b$, as defined in Eq.~\eqref{eq:Hb}, satisfies the \emph{large-excited-energy condition} if the ground-state of $H_b$ is non-degenerate and all excited states have energy greater than $-1 + 1/n$. 
\end{condition}

Observe that when $b=0$, the Hamiltonian $H_b$ is equal to $-X/n$ and the large-excited-energy condition is satisfied since all excited states have energy at least $-1+2/n$. In the situations for which we prove rigorous bounds, the large-excited-energy condition will continue to hold as $b$ is increased, up until some $n$-independent threshold. For example, in Fig.~\ref{fig:spectrum_n20}, we plot the numerically computed eigenvalues of $H_b$ as a function of $b$ for an $n=20$ instance drawn from the $3$-spin ensemble, with $\eta=0.5$. For this instance, the large-excited-energy condition persists past $b=0.8$.

\begin{figure}[h]
    \centering
    \includegraphics[draft=false,width = 0.6\textwidth]{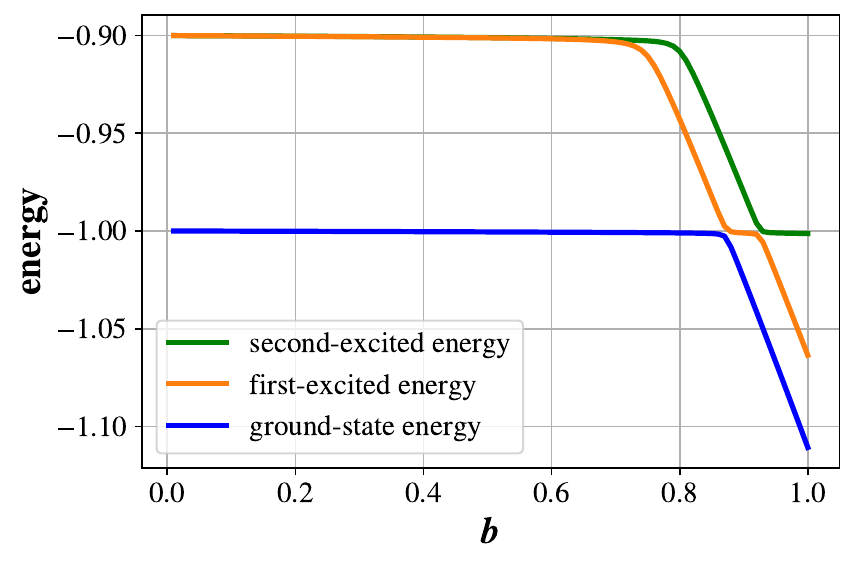}
    \caption{Plot of the lowest three eigenvalues of $H_b$ as a function of $b$, for an $n=20$ instance randomly chosen from the $3$-spin ensemble, with $\eta = 0.5$. Eigenvalues were computed numerically using exact diagonalization. The two key features are that the spectral gap remains large and the ground state energy barely shifts from $-1$ until a relatively large value of $b$, namely $b \approx 0.85$.}
    \label{fig:spectrum_n20}
\end{figure}

\begin{theorem}[runtime]\label{thm:runtime}
    Fix a cost function $H$ and parameters $\eta$ and $b$, which defines $H_b$ through Eq.~\eqref{eq:Hb}. If $H_b$ has the large excited-energy condition (Cond.~\ref{cond:large-excited-energy}), then with probability at least $1-\exp(-\Omega(n))$, the algorithm outputs an optimal solution $z^*$ of $H$ and runs in time at most
    \begin{equation}
        \poly(n)\left(\lvert \braket{\boldplus}{\psi_b}\rvert^{-1}+ \lVert \Pi^* \ket{\psi_b} \rVert^{-1} \right)\,,
    \end{equation}
    where $\ket{\psi_b}$ is the ground state of $H_b$ and $\Pi^*$ is the projector onto the ground space of $H$. 
\end{theorem}
\begin{proof}
    Note that since $g_\eta(H/|E^*|)$ is a negative semidefinite operator, we have $\bra{\boldplus}H_b\ket{\boldplus} \leq -1$ and thus, by the variational principle, the ground state energy $E_b$ of $H_b$ satisfies $E_b \leq -1$. The large-excited-energy condition then implies that $\Delta = 1/n$ is a lower bound on the spectral gap of $H_b$. We now refer to the steps described in the pseudocode of Algorithm \ref{algo:main_simple}. The only non-trivial steps are steps 2 and 3. Step 2 performs the jump $-X/n \rightarrow H_b$. Note that $-X/n$ is a classical Hamiltonian as it is diagonal in the Hadamard basis with efficiently computable entries. Thus, by Prop.~\ref{prop:jump_simplified}, we may choose $\delta = e^{-\Omega(n)}$ and assert that step 2 prepares $\ket{\psi_b}$ up to error $\delta$ (in standard Euclidean vector norm) and runs in time $\poly(n)\Delta^{-1}\lvert \braket{\boldplus}{\psi_b}\rvert^{-1}$. Similarly, step 3 performs the jump $H_b \rightarrow H/|E^*|$. Noting that $H/|E^*|$ is also a classical Hamiltonian, and that $\Pi^*$ is its ground-space projector, Prop.~\ref{prop:jump_simplified} implies that it prepares $\Pi^*\ket{\psi_b}/\lVert \Pi^* \ket{\psi_b} \rVert$ up to error at most $2\delta$ and runs in time $\poly(n)\Delta^{-1}\lVert\Pi^*\ket{\psi_b}\rVert^{-1}$.
    This is true even if the spectral gap of $H/|E^*|$ is exponentially small. A subsequent computational basis measurement produces an optimal solution $z^*$ with probability at least $1-2\delta = 1-\exp(-\Omega(n))$. As $\Delta^{-1}$ is $\poly(n)$, the theorem statement follows.
\end{proof}

Theorem \ref{thm:runtime} shows that the large-excited-energy condition is sufficient for algorithmic success and a bound on its runtime. However, the large-excited-energy condition is not a necessary condition: a similar statement would follow for a relaxed version of the condition, where the excited energy is at least $-1 + 1/\poly(n)$. Additionally, the algorithm could still succeed even if the excited energy falls beneath $-1$, as long as there is a sizable gap between the ground and excited energy and a good approximation to ground/excited energy. We focus on the large-excited-energy condition because we will be able to prove that it holds under certain circumstances.

\section{Proving speedup over Grover}

The runtime statement in Theorem \ref{thm:runtime} implies that a super-Grover speedup can be shown given sufficient control of the spectrum of $H_b$ (and in particular the first-excited energy) in combination with an upper bound on the quantity $\lvert \braket{\boldplus}{\psi_b}\rvert^{-1} + \lVert \Pi^* \ket{\psi_b}\rVert^{-1}$, which is essentially equivalent to a lower bound on the overlaps $\braket{\boldplus}{\psi_b}$ and $\lVert \Pi^* \ket{\psi_b}\rVert$. These tasks are accomplished separately, and in this section, we illustrate all the technical steps involved while deferring many of the mathematical proofs to the Appendices. In Sec.~\ref{sec:additional_conditions} we define additional conditions needed to organize our technical results. In Sec.~\ref{sec:bounding_overlap}, we show how, when these conditions are met, the runtime enjoys a super-Grover speedup. In Sec.~\ref{sec:showing_conditions}, we show how a tail bound on the spectral density of the cost function implies that all of the conditions are met. We also discuss when such a tail bound is guaranteed to hold, and what can be said in the case there is no such tail bound. 

\subsection{Additional conditions and properties}\label{sec:additional_conditions}
In addition to the large-excited-energy condition, we define the small-ground-energy-shift condition and the $\alpha$-subdepolarizing property, which will be needed to bound the runtime of the algorithm.

\begin{condition}[small-ground-energy-shift condition]\label{cond:small-ground-energy-shift}
We say that the Hamiltonian $H_b$, as defined in Eq.~\eqref{eq:Hb}, satisfies the small-ground-energy-shift condition if the ground-state energy $E_b$ of $H_b$ satisfies $-1-1/n^3 \leq E_b \leq -1$.
\end{condition}

To get a qualitative sense of the idea behind Cond.~\ref{cond:small-ground-energy-shift}, observe the remarkable flatness of the ground state energy of the 3-spin instance depicted in Fig.~\ref{fig:spectrum_n20}: $E_b$ stays very close to $-1$ until relatively large values of $b$.

We also define a property of cost functions that we call the $\alpha$-subdepolarizing property, which is important for establishing a lower bound on the overlap that determines the algorithm's runtime. To define $\alpha$-subdepolarizing, we introduce the notation $y \sim x$ to denote that bit string $y \in \{+1,-1\}^n$ is generated from $x$ by flipping a single bit chosen uniformly at random. First we define $\alpha$-depolarizing before generalizing to $\alpha$-subdepolarizing. 
\begin{definition}[$\alpha$-depolarizing]\label{def:alpha-depolarizing}
We say a cost function $H$ is $\alpha$-depolarizing if for every bit string $x$
\begin{equation}\label{eq:alpha-depolarizing}
    \EV_{y \sim x} H(y) = (1-\alpha)H(x)
\end{equation}
\end{definition}
We call the property $\alpha$-depolarizing because it states that flipping a single bit at random brings the energy toward zero by exactly a fixed factor $1-\alpha$, on average. We can immediately note that all MAX-E$k$-LIN2 instances are $\alpha$-depolarizing, due to the fact that every term has the same degree.
\begin{proposition}\label{prop:MAX-Ek-LIN2-alpha-depolarizing}
    Any MAX-E$k$-LIN2 instance has the $\alpha$-depolarizing property with $\alpha = 2k/n$.
\end{proposition}
\begin{proof}
If one of the $n$ bits is flipped at random, then the sign of any degree-$k$ monomial will flip with probability $k/n$. Thus the expectation value of the monomial is brought toward zero by a factor $1-2k/n$, and the monomial is $2k/n$-depolarizing. Moreover, the sum of cost functions all of which possess the $\alpha$-depolarizing property is also $\alpha$-depolarizing, by linearity of Eq.~\eqref{eq:alpha-depolarizing}. This extends the property to all MAX-E$k$-LIN2 instances. 
\end{proof}
Although Hastings \cite{hastings2018shortPath} did not use the same terminology, it was precisely the $\alpha$-depolarizing property that led to an upper bound on the runtime of the short-path algorithm that suggested the possibility of super-Grover speedup.  We now define the $\alpha$-subdepolarizing property. 
\begin{definition}[$\alpha$-subdepolarizing]\label{def:alpha-subdepolarizing}
Consider a pair $(H,g)$, where $H$ is a cost function with optimal value $E^* < 0$ and $g:[-1,\infty) \rightarrow [-1,0]$ is a monotonic non-decreasing, concave function that is twice-differentiable at every point where it is nonzero. Let $f(x) := -g(-x)$, so that $f$ is monotonically non-decreasing and convex. We say that $(H,g)$ is $\alpha$-\textit{subdepolarizing} if for any set of constants $0 < c_1,c_2,\ldots,c_T < 1$, 
\begin{equation}
    \EV_{y \sim x} \prod_{t=1}^T f\left(\frac{c_tH(y)}{E^*}\right) \geq \prod_{t=1}^T f\left(\frac{c_t(1-\alpha) H(x)}{ E^*}\right)\,.
\end{equation}
\end{definition}
The definition appears complex, but it attempts to capture the same idea as $\alpha$-depolarizing, with minor relaxations that allow for MAX-$k$-CSP cost functions to be included. First, note that if $H$ is $\alpha$-depolarizing, then $(H,g)$ is $\alpha$-subdepolarizing for any function $g$ satisfying the criteria in Def.~\ref{def:alpha-subdepolarizing}, which includes $g_\eta$ from Eq.~\eqref{eq:g_eta} for any $\eta$ (see Prop.~\ref{prop:depolarizing_implies_subdepolarizing} in App.~\ref{app:alpha_(sub)depolarizing}). Second, note that any MAX-$k$-CSP instance has the property for the function $g_\eta$ for any $\eta$, as stated in the following proposition, which is proved in App.~\ref{app:alpha_(sub)depolarizing}.

\begin{proposition}\label{prop:CSP_alpha-sub-depolarizing}
Suppose $H$ is a MAX-$k$-CSP instance with $m$ terms and optimal value $E^*$. Then, for any $\eta$, $(H,g_\eta)$ is $\alpha$-subdepolarizing with 
\begin{equation}
    \alpha=\frac{m}{\lvert E^*\rvert }\frac{k2^k}{(1-\eta)n}
\end{equation}
In particular, if $H$ is frustration free, i.e.~fully satisfiable, then $\lvert E^* \rvert = m$ and $\alpha = k 2^k/(1-\eta)n$.  
\end{proposition}

\subsection{Bounding the runtime with an approximate ground-state projector}\label{sec:bounding_overlap}

Per Theorem \ref{thm:runtime}, the runtime of the algorithm (assuming the large-excited-energy condition) depends on the quantity $\lvert \braket{\boldplus}{\psi_b}\rvert^{-1} + \lVert \Pi^* \ket{\psi_b}\rVert^{-1}$. We wish to upper bound this quantity. In particular, we want to show that it is at most $2^{(0.5-c)n}$ for some constant $c$, implying a super-Grover speedup. 

Let $z^*$ be any optimal bit string and note that $\lVert \Pi^*\ket{\psi_b}\rVert \geq \braket{z^*}{\psi_b}$ (recall we take the convention that all of the entries of $\ket{\psi_b}$ are positive in the computational basis). Then we have
\begin{equation}
\begin{split}
    \braket{\boldplus}{\psi_b}^{-1} + \lVert \Pi^* \ket{\psi_b} \rVert^{-1} &\leq  \braket{\boldplus}{\psi_b}^{-1} + \braket{\psi_b}{z^*}^{-1} \leq 2\left(\braket{\boldplus}{\psi_b}\braket{\psi_b}{z^*}\right)^{-1}  \\
    &=2\bra{\boldplus} \Pi_b \ket{z^*}^{-1} \label{eq:runtime_bound_projector}
    \end{split}
\end{equation}
where $\Pi_b = \ket{\psi_b}\bra{\psi_b}$ is the ground-state projector for the Hamiltonian $H_b$. We will replace $\Pi_b$ by an \emph{approximate ground state projector}, a tool that has been used successfully in the completely different context of proving area laws  for ground states of many-body Hamiltonians \cite{aharonov2011AreaLaw1D,arad2012improved,arad2017rigorousRG}. As in the context of area laws, our approximate ground state projector will be a degree-$\ell$ polynomial in $H_b$; however, where they used Chebyshev polynomials, we need only examine the simpler polynomial

\begin{equation}\label{eq:P_ell}
    P_\ell := \left(\frac{H_b}{E_b}\right)^\ell\,.
\end{equation}
The operator $P_\ell$ approximates $\Pi_b$ since $\ket{\psi_b}$ is an eigenstate with eigenvalue 1, and, assuming $\ell$ is sufficiently large, the other eigenvalues of $P_\ell$ will be close to zero. We show that $\ell = \Omega( n^2)$ is sufficiently large, assuming the large-excited-energy condition.\footnote{One could consider using Chebyshev-like polynomials for $P_\ell$, which could reduce the requirement of the degree of $P_\ell$ to $\ell = \Omega(n^{3/2}$). However, this would provide very limited benefit to the proof, and it would bring its own additional complications; for example, in lower bounding $\bra{+}P_\ell \ket{z}$, it is helpful that all the terms in the polynomial in $H_b/E_b$ (there is only one term) have a positive coefficient.}

\begin{lemma}\label{lem:overlap_bound_Pl}
    If $H_b$ satisfies the large-excited-energy condition (Cond.~\ref{cond:large-excited-energy}), then for any $z$ and any $L \geq (\mu+1.5\ln(2)) n^2$, the following equation holds either for $\ell = L$ or $\ell = L+1$ (or both):
    \begin{equation}
    \braket{\boldplus}{\psi_b}\braket{\psi_b}{z} \geq \bra{\boldplus} P_\ell \ket{z} - 2^{-n/2}e^{-\mu n}\,.
\end{equation}
\end{lemma}
\begin{proof}
Consider the operator $-X/n$, which is the first term of $H_b$, as defined in Eq.~\eqref{eq:Hb}. Its maximum eigenvalue is 1, associated with eigenvector $\ket{\boldminus} \equiv 2^{-n/2}(\ket{0}-\ket{1})^{\otimes n}$, and its second-largest eigenvalue is $1-2/n$. Meanwhile, the second term $b\,g_\eta(H/|E^*|)$ is a negative semidefinite operator. Denoting the largest eigenvalue of $H_b$ by $E'_b$, and the associated eigenvector by $|\psi'_b\rangle$, we can say that $E'_b \leq 1$. Additionally, we can assert that all other eigenvalues of $H_b$ are smaller than $1-2/n$. To see this, suppose for contradiction that there were two eigenvectors $\ket{p}$ and $\ket{q}$ of $H_b$ with eigenvalue greater than $1-2/n$. Then the state $\ket{p} - \frac{\braket{\boldminus}{p}}{\braket{\boldminus}{q}}\ket{q}$ is orthogonal to $\ket{\boldminus}$ and has average energy larger than $1-2/n$. However, this is impossible, since the average value of the $-X/n$ term can be at most $1-2/n$ (as the state is orthogonal to $\ket{-}$) and the average energy of the $b\,g_\eta(H/|E^*|)$ term can be at most 0. 

The operator $P_\ell$ has the same eigenvectors as $H_b$, and for each eigenvalue $\lambda$ of $H_b$, $(\lambda/E_b)^\ell$ is an eigenvalue of $P_\ell$. Thus, in the limit of $\ell \rightarrow \infty$, $P_\ell$ approaches the projector $\ket{\psi_b}\bra{\psi_b}$. By assumption, $\ell \geq L \geq \nu n^2$ with $\nu = \mu + 1.5\ln(2)$. Hence, by the large-excited-energy condition, all $2^n$ eigenvalues of $P_\ell$ have magnitude at most $(1-1/n)^{\nu n^2} \leq e^{-\nu n} \leq 2^{-3n/2}e^{-\mu n}$, except for the eigenvalue associated with $\ket{\psi_b}$, which is 1, and (if $E'_b > 1-1/n$) the eigenvalue associated with $|\psi_b'\rangle$ which is $(E_b'/E_b)^\ell$. Denote these $2^{n}-2$ eigenvalues by $\lambda_i$ and associated eigenvectors by $\ket{\lambda_i}$ for $i=1,\ldots, 2^{n}-2$. Thus, the quantity $\bra{\boldplus} P_\ell \ket{z}$ is equal to
\begin{equation}
\nonumber
 \braket{\boldplus}{\psi_b}\braket{\psi_b}{z} + \left(\frac{E_b'}{E_b}\right)^\ell \braket{\boldplus}{\psi'_b}\braket{\psi'_b}{z} + \sum_{i=1}^{2^n-2} \lambda_i \braket{\boldplus}{\lambda_i}\braket{\lambda_i}{z}
\end{equation}
which is upper bounded by
\begin{equation}
\braket{\boldplus}{\psi_b}\braket{\psi_b}{z} + \left(\frac{E_b'}{E_b}\right)^\ell \braket{\boldplus}{\psi'_b}\braket{\psi'_b}{z} + (2^n-2) 2^{-1.5n}e^{-\mu n} \label{eq:<+|Pl|z>_intermediate} 
\end{equation}
In the case where $E'_b > 1-1/n>0$, we have $E_b'/E_b < 0$ (since $E_b < 0$) and hence $(E_b'/E_b)^\ell \langle\boldplus|\psi'_b\rangle\langle\psi'_b|z\rangle$ is non-positive for exactly one of the choices $\ell = L$ or $\ell = L+1$. For this choice of $\ell$, we have
\begin{equation}
 \braket{\boldplus}{\psi_b}\braket{\psi_b}{z} \geq \bra{\boldplus} P_\ell \ket{z} - 2^{-n/2}e^{-\mu n}\,. \label{eq:<+|Pl|z>_final}
\end{equation}
In the case where $E'_b \leq 1-1/n$, the second term of Eq.~\eqref{eq:<+|Pl|z>_intermediate} can be combined with the third term to arrive at the same result in Eq.~\eqref{eq:<+|Pl|z>_final}.
\end{proof}

Lemma \ref{lem:overlap_bound_Pl}, together with Eq.~\eqref{eq:runtime_bound_projector}, reduces the task of upper bounding the runtime of the algorithm to lower bounding the quantity $\bra{\boldplus} P_\ell \ket{z^*}$ for $\ell = \Omega(n^2)$. We produce a lower bound by expanding $P_\ell$ as a sum of $2^\ell$ terms using its definition in Eq.~\eqref{eq:P_ell} and the definition of $H_b$ in Eq.~\eqref{eq:Hb}. Assuming the small-ground-energy-shift condition (Cond.~\ref{cond:small-ground-energy-shift}) and that $\ell<O(n^3)$, we can say that the magnitude of the denominator $|E_b|^\ell$ is at most a constant, as $(1+O(1/n^3))^{O(n^3)} = O(1)$. Each of these $2^\ell$ terms contributes a positive amount to the sum; we lower bound the sum by lower bounding each individual term under the assumption that $(H,g_\eta)$ has the $\alpha$-subdepolarizing property (Def.~\ref{def:alpha-subdepolarizing}).   The result of this calculation is captured in Lemma \ref{lem:<+|Pl|0>}, which is proved in App.~\ref{app:overlap_bound}.

\begin{lemma}\label{lem:<+|Pl|0>}
Given positive parameters $\eta < 1$, $b < 1$, $ \alpha < (1-b)/2$, and integer $\ell$, suppose that $(H,g_\eta)$ has the $\alpha$-subdepolarizing property (Def.~\ref{def:alpha-subdepolarizing}), that $3/\alpha^2 \leq \ell < n^3$, and that $H_b$ satisfies the small-ground-energy shift condition (Cond.~\ref{cond:small-ground-energy-shift}).
Define the function $F:[0,1]\rightarrow[0,1]$ as follows (plot of $F$ included for convenience).
\begin{equation}\label{eq:F_def}
    F(x) = 1-x+x\ln\left(x\right) \,. \qquad \qquad \includegraphics[width=0.3\textwidth,valign=c,draft=false]{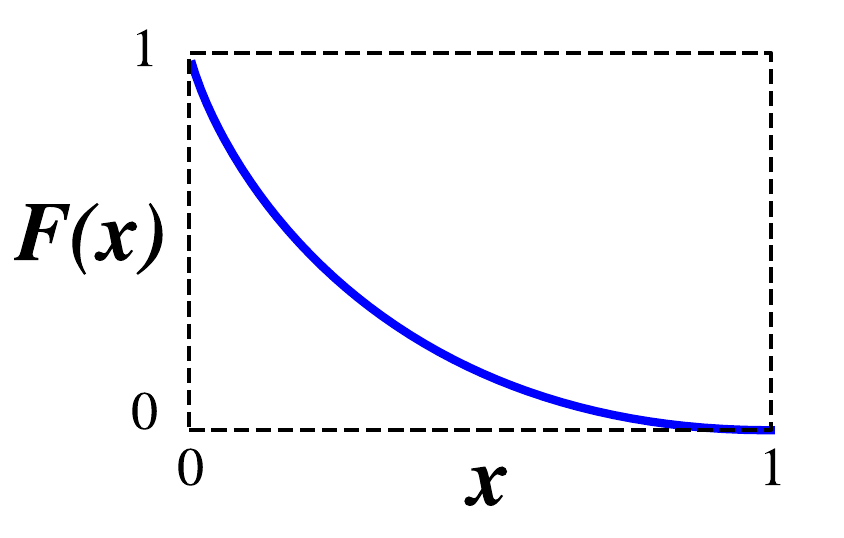}
\end{equation}
Let $z \in \{+1,-1\}^n$ be any assignment for which $\mathcal{E}:=H(z)/|E^*|$ satisfies $\mathcal{E}  \leq -(1-\eta)$. Then
    \begin{equation}\label{eq:AGSP_exp_advantage}
        \bra{\boldplus}P_\ell \ket{z} \geq 2^{-n/2}\exp\left(\frac{b}{\alpha}\frac{\mcEbar}{\eta }F\left(\frac{1-\eta}{\mcEbar}\right)\right)(e^{-1}-2e^{-2})\,,
    \end{equation}
    noting that $e^{-1}-2e^{-2} \approx 0.0972 = \Omega(1)$.
\end{lemma}

These lemmas together imply that the following conditions are sufficient for a super-Grover speedup: (i) the cost function has the $\alpha$-subdepolarizing property for $\alpha = O(1/n)$, as is the case for MAX-E$k$-LIN2 and MAX-$k$-CSP with limited frustration, and (ii)  $b$ and $\eta$ are constants independent of $n$ chosen such that $H_b$ satisfies the large-excited-energy and small-ground-energy-shift conditions. Generally speaking, the larger $b$ and $\eta$ are, the larger the speedup that can be shown. This is formally captured in the following Theorem. 

\begin{theorem}[Super-Grover speedup]\label{thm:super-Grover-speedup}
    Let $H$ be a cost function on $n$ variables with $n \geq 4$. Fix parameters $\eta$, $b$ and $a$. Suppose that $H$ has the $\alpha$-subdepolarizing property (Def.~\ref{def:alpha-subdepolarizing}) with $\alpha = a/n$ and $1 \leq a < n(1-b)/2$. Suppose further that $H_b$ satisfies the large-excited-energy condition (Cond.~\ref{cond:large-excited-energy}) and the small-ground-energy-shift condition (Cond.~\ref{cond:small-ground-energy-shift}). Then, Algorithm \ref{algo:main_simple} successfully produces an optimal solution with probability $1-e^{-\Omega(n)}$ while running in time bounded above by
    \begin{equation}
        \poly(n) 2^{(0.5-c)n}
    \end{equation}
    where
    \begin{equation}
        c = \frac{bF(1-\eta)}{a\eta\ln(2)} \geq \frac{b\eta}{2a\ln(2)}\,.
    \end{equation}
\end{theorem}
\begin{proof}
    Theorem \ref{thm:runtime}, together with Eq.~\eqref{eq:runtime_bound_projector} and Lemma \ref{lem:overlap_bound_Pl} imply that for $L > (\mu + 1.5)n^2$, the runtime is upper bounded by
    \begin{equation}
        \poly(n) (\bra{\boldplus} P_\ell \ket{z^*} - 2^{-n/2}e^{-\mu n})^{-1}
    \end{equation}
    for either $\ell=L$ or $\ell=L+1$, where $\ket{z^*}$ is any optimal solution to $H$. We choose $L = 3.5n^2$ (i.e.~$\mu = 2$) and note that, by the assumptions of the lemma, $3/\alpha^2 \leq \ell < n^3$ holds for both $\ell = L$ and $\ell = L+1$. Thus all the conditions of Lemma \ref{lem:<+|Pl|0>} hold. As $\mcEbar = 1$ in this case, we have that the runtime is upper bounded by
    \begin{align}
        &\poly(n) 2^{n/2}(e^{n\frac{bF(1-\eta)}{\eta a}}(e^{-1}-2e^{-2}) - e^{-2n})^{-1} \,.
    \end{align}
    The $e^{-2n}$ term will be smaller than the first term by a constant factor whenever $n \geq 4$. We can thus absorb it, as well as the $e^{-1}-2e^{-2}$ factor, into the $\poly(n)$ expression. This proves the theorem. Note that $F(1-\eta)/\eta \geq \eta/2$ holds for $0\leq \eta \leq 1$. 
\end{proof}

Our results actually say something stronger: every bit string $z$ lying in a deep cost valley, that is, those for which $H(z) \leq (1-\eta)E^*$, will have overlap with $\ket{\psi_b}$ that is $2^{c'n}$ larger than $2^{-n/2}$ for some $n$-independent constant $c'$. Thus, if slightly suboptimal solutions are acceptable, the probability that our algorithm will find one of these bit strings upon measurement of $\ket{\psi_b}$ is also boosted by an amount $2^{c'n}$ compared to measurement of $\ket{\boldplus}$. This is irrelevant for our particular task of finding the exactly optimal solution, but could be relevant in other situations.

\subsection{Tail bound on spectral density implies conditions are met}\label{sec:showing_conditions}
For a family of cost functions with the $\alpha$-subdepolarizing property with $\alpha = O(1/n)$, Theorem \ref{thm:super-Grover-speedup} reduces the task of showing super-Grover speedup to the task of proving that the large-excited-energy and small-ground-energy-shift conditions are met. In general, we do not show that the conditions are always satisfied. Rather, we show that they are satisfied whenever there is a tail bound on the cumulative number of assignments to $H$ beneath a certain cost (i.e.~integral of spectral density of $H$), as follows. 

\begin{lemma}\label{lem:tail_bound_implies_conditions}
Let $H$ be a cost function over assignments $\{+1,-1\}^n$, and let $C(E)$ denote the number of assignments to $H$ with cost value beneath $E$, i.e.
\begin{equation}
    C(E) := \lvert \{z: H(z) \leq E\} \rvert \,.
\end{equation}
Suppose that for a certain choice of $\eta < 1$
\begin{equation}\label{eq:tail_bound}
    C\left((1-\eta)E^*\right) \leq 2^{(1-\gamma) n}\,.
\end{equation}
Suppose further that $\gamma \geq (1+4\log_2(n))/n$ (which holds for any $n$-independent $\gamma$ and sufficiently large $n$). Then the large-excited energy condition (Cond.~\ref{cond:large-excited-energy}) and the small-ground-energy-shift condition (Cond.~\ref{cond:small-ground-energy-shift}) are met for the Hamiltonian $H_b$ (as defined in Eq.~\eqref{eq:Hb}) whenever
\begin{equation}
    b \leq \frac{\ln(2)\gamma}{2+\ln(2)} \approx 0.257\gamma
\end{equation}
\end{lemma}

Lemma \ref{lem:tail_bound_implies_conditions} is proved in App.~\ref{sec:proof_of_lem_tail_bound}; the proof utilizes the log-Sobolev inequality and tools from statistical mechanics. 

\begin{proof}[Proof sketch of Lemma \ref{lem:tail_bound_implies_conditions}]
The basic idea is as follows. For the large-excited-energy condition to be violated, there must exist a state $\ket{\phi}$ orthogonal to $\ket{\boldplus}$ for which $\bra{\phi}H_b\ket{\phi} \leq -1+1/n$. Thus, referring to the definition of $H_b$ in Eq.~\eqref{eq:Hb}, there must be a number $\mathcal{U}$ satisfying $-1 \leq \mathcal{U} \leq 0$ such that the following two relations hold simultaneously:
\begin{align}
    \mathcal{U} &= \bra{\phi}g_\eta(H/|E^*|)\ket{\phi} \label{eq:U_def}\\
    \bra{\phi}X/n \ket{\phi} &\geq 1-1/n + b\, \mathcal{U} \,.\label{eq:X/n_bound}
\end{align}
Let $p(z) = \lvert \braket{z}{\phi}\rvert^2$ denote the probability of obtaining $z$ when measuring $\ket{\phi}$ in the computational basis. Viewing $p$ as a probability distribution over assignments, Eq.~\eqref{eq:U_def} says that the average cost for the function $g_\eta(H/|E^*|)$ of a sample from $p$ is equal to $\mathcal{U}$. Meanwhile, the log-Sobolev inequality \cite{samorodnitsky2008modified, hastings2018shortPath} allows us to turn the lower bound on $\bra{\phi}X/n \ket{\phi}$ in Eq.~\eqref{eq:X/n_bound} into a lower bound on $S$, the (log base-2) entropy of $p$, which reads:
\begin{equation}
    S \geq n\left( 1-\ln(2)^{-1}(- b\,\mathcal{U} + 1/n) \right) \,.
\end{equation}
This turns the question into a statistical mechanics problem: is there a distribution $p$ with entropy greater than $n(1-O(|\mathcal{U}|))$, yet average cost of $g_\eta(H/|E^*|)$ equal to $\mathcal{U}$? Here, the tail bound constrains what is possible: since $g_\eta$ zeroes out the cost of any assignment $z$ for which $H(z) > (1-\eta)E^*$, there are at most $2^{(1-\gamma)n}$ assignments $z$ that have negative cost for $g_\eta(H/|E^*|)$. In order to have a distribution with average cost $\mathcal{U} < 0$, a large portion of the probability must be concentrated on this subset of $2^{(1-\gamma)n}$ assignments, constraining the entropy. Formally, for this step we use an elementary fact from statistical mechanics: the distribution that maximizes entropy for a fixed average energy (average cost) is a Gibbs distribution, where assignments with cost $E$ are allocated probability proportional to $e^{-\beta E}$ for some value of $\beta$, which physically corresponds to the inverse temperature. Applying these tools, we are able to show that whenever $b \leq 0.257\gamma$, it is not possible for Eq.~\eqref{eq:U_def} and Eq.~\eqref{eq:X/n_bound} to simultaneously be true, implying that the large-excited-energy condition holds. 

Separately, the small-ground-energy-shift condition follows from a perturbation-theory argument. Viewing $b\,g_\eta(H/|E^*|)$ as a perturbation to $-X/n$ in the expression for $H_b$ in Eq.~\eqref{eq:Hb}, the magnitude of the first-order shift in energy is $\lvert\bra{\boldplus}b\,g_\eta(H/|E^*|)\ket{\boldplus}\rvert$, which is smaller than $b2^{-\gamma n}$ (an exponentially small number), simply due to the tail bound, and the fact that only $2^{(1-\gamma)n}$ entries of $g_\eta(H/|E^*|)$ have non-zero cost. The full proof also bounds higher-order contributions. 
\end{proof}

To get a sense of the interplay between the tail bound and the magnitude of the super-Grover speedup, suppose $H(z)$ comes from a family of cost functions with a unique optimal solution and for which exactly half of the eigenvalues are negative, so that $C(E^*)=1$ and $C(0) = 2^{n-1}$. Thus, asymptotically speaking, the tail bound is satisfied with $\gamma = 1$ at $\eta = 0$, and with $\gamma = 0$ at $\eta =1$. However, from Theorem \ref{thm:super-Grover-speedup}, we see that the magnitude of the super-Grover speedup is zero if either $\eta = 0$ or if $\gamma = 0$ (since $\gamma = 0$ implies $b = 0$, by Lemma \ref{lem:tail_bound_implies_conditions}), so in both of these cases, these observations are not sufficient to show a super-Grover speedup. To show a super-Grover speedup, we need a non-trivial tail bound for non-zero $\eta$, $\gamma$ to be true.

What happens if the tail bound condition is not satisfied for any choice of $\eta$, $\gamma$? Then, regardless how small we make $\eta$, there must be \textit{many} assignments for which $H(z) \leq (1-\eta)E^*$. There are so many assignments that a classical algorithm could produce an assignment achieving a $1-\eta$ approximation ratio by simple repetition in sub-exponential time for any constant value of $\eta$. Thus, instances with the $\alpha$-subdepolarizing property can be partitioned into a set where our quantum algorithm has a super-Grover speedup, and a set that are unusually classically easy in a precise sense.

When can the tail bound be shown? For spin-glass-like cost functions consisting of random local terms, we often expect the spectral density to be roughly Gaussian, with the minimum cost assignment among all $2^n$ assignments lying $\Omega(\sqrt{n})$ standard deviations beneath 0. In this case, an exponential tail bound can be shown for any $\eta > 0$. As $\eta$ increases, the corresponding bound $\gamma$ decreases. 

\begin{proposition}\label{prop:tail_bound_cases}
    In each of the following situations, a tail bound of the form in Eq.~\eqref{eq:tail_bound} holds for any $\eta$ and some choice of $\gamma$ that depends on $\eta$ but is independent of $n$. 
    \begin{enumerate}
        \item If the cost function $H$ is a MAX-$k$-CSP instance, a tail bound holds with 
        \begin{equation}
            \gamma = \left(\frac{|E^*|}{m}\right)^2\frac{(1-\eta)^2}{\ln(2)2^{2k}k^2D}
        \end{equation}
        where $D$ is a measure of how irregular the constraint hypergraph is (e.g., $D=1$ if all $n$ variables participate in exactly $km/n$ constraints); specifically $D = n k^{-2}m^{-2}\sum_{j=1}^n d_j^2$, with $d_j$ the number of constraints in which variable $j$ participates. If the clauses $\mathcal{C}_j$ are each independently chosen to act on a set of $k$ uniformly random bits, then $D \leq 3$ holds for most instances. 
        \item If the cost function $H$ is randomly chosen from the $k$-spin ensemble of MAX-E$k$-LIN2 instances, then for any $\eta$ a tail bound holds with probability at least $1-2^{-\gamma n +1}$ over choice of instance, with
        \begin{equation}
            \gamma = \frac{(1-\eta)^2}{32\pi \ln(2)k^2}
        \end{equation}
    \end{enumerate}
\end{proposition}
\begin{proof}
Item 1 is shown in Prop.~\ref{prop:CSP_tail_bound} and item 2 is shown in Prop.~\ref{prop:k-spin_tail_bound}, which appear in App.~\ref{app:tail_bounds}.
\end{proof}

\subsection{Formal statement of main results}\label{sec:formal_results}

\begin{theorem}
    Let $H$ be an instance of Quadratic Unconstrained Binary Optimization (QUBO) with $n \geq 12$ or MAX-E$k$-LIN2 with $n \geq 6k$. Let $\gamma_0 = (1+4\log_2(n))/n$ (and note that $\gamma_0 \rightarrow 0$ as $n \rightarrow \infty$). For every $\gamma \in [\gamma_0,1]$ and every $\eta \in [0,1]$, one of the following must be true.
    \begin{enumerate}[(a)]
    \item There is a quantum algorithm running in time $O^*(2^{(0.5-c)n})$ that produces an optimal solution to $H$ with probability at least $1-e^{-\Omega(n)}$, where
    \begin{equation}
        c = \frac{1}{4(2+\ln(2))}\frac{\gamma\eta}{k} \approx 0.0928 \frac{\gamma\eta}{k}
    \end{equation}
    with $k=2$ in the case of QUBO. 
    \item  There is a classical algorithm which repeatedly samples assignments uniformly at random that outputs a solution $z$ for which $H(z) \leq E^*(1-\eta)$ in expected time $O^*(2^{\gamma n})$. 
    \end{enumerate}
    Thus, if $(b)$ is the case for arbitrarily small $\eta$ and arbitrarily small $\gamma$, we can say that the classical algorithm can achieve an arbitrarily good approximation to the optimal cost value in sub-exponential time. 
\end{theorem}
\begin{proof}
    Rather than analyze QUBO directly, we employ a well-known reduction from QUBO to MAX-E2-LIN2: given a QUBO instance $H$ on $n$ variables, construct a MAX-E2-LIN2 instance $H'$ on $n+1$ variables by introducing a binary variable $z_0$ and multiplying it with all degree-1 terms of $H$. Since all terms of $H'$ have even degree, there is a $\mathbb{Z}_2$ symmetry where the value of $H'(z)$ is unchanged under a flip of all bits of $z$; thus, the spectrum of $H'$ is the same as the spectrum of $H$, with each energy appearing twice. The ground state of $H$ can be determined by computing a ground state $z$ of $H'$ and flipping the bits $z_1,\ldots,z_n$ if $z_0 = -1$. Since $H$ and $H'$ have the same spectrum, one will obey a tail bound of the form of Eq.~\eqref{eq:tail_bound} if and only if the other obeys the same tail bound. Thus all subsequent statements made about MAX-E$k$-LIN2 will also apply to QUBO with the substitution $k=2$. 
    
    Note that by Prop.~\ref{prop:MAX-Ek-LIN2-alpha-depolarizing}, any MAX-E$k$-LIN2 instance $H$ is $\alpha$-depolarizing (Def.~\ref{def:alpha-depolarizing}) with $\alpha = 2k/n$, and thus by Prop.~\ref{prop:depolarizing_implies_subdepolarizing}, for every $\eta$ the pair $(H,g_\eta)$ is $\alpha$-subdepolarizing (Def.~\ref{def:alpha-subdepolarizing}) with the same value of $\alpha$. For a certain choice of $\eta$ and $\gamma$, either $H$ satisfies a tail bound of the form in Eq.~\eqref{eq:tail_bound} (case (a)), or it does not (case (b)). 
    
    In case (a), Lemma \ref{lem:tail_bound_implies_conditions} implies that $H_b$ satisfies the large-excited-energy and small-ground-energy-shift conditions when we choose $b = \gamma\ln(2)/(2+\ln(2))$. Note that if $\gamma < 1$ then  $b < 0.3$. So long as $n \geq 6k$, the condition $1 \leq (\alpha n) \leq n (1-b)/2$ will hold, and we may apply Theorem \ref{thm:super-Grover-speedup} to prove the theorem statement. 
    
    In case (b),  the fact that there are more than $2^{(1-\gamma)n}$ assignments with cost at most $(1-\eta)E^*$ among $2^n$ total assignments implies that the average number of samples the classical algorithm must make before finding one of these assignments is at most $2^{\gamma n}$. Assuming the algorithm terminates as soon as it finds one such assignment, and noting that each sample can be drawn in $\poly(n)$ time, the statement follows. 
\end{proof}

\begin{theorem}\label{thm:CSP_main_result}
    There is a quantum algorithm which, for any instance of MAX-$k$-CSP, produces an optimal solution with probability at least $1-e^{-\Omega(n)}$ while running in time $O^*(2^{(0.5-c)n})$, where
    \begin{equation}
        c = 0.0289\left(\frac{|E^*|}{m}\right)^3\frac{1}{2^{3k}k^3D}
    \end{equation}
    with $m$ denoting the number of clauses, $E^*$ denoting the optimal value of the instance, and $D$ denoting the irregularity parameter defined by $D = nk^{-2}m^{-2}\sum_{j=1}^{n} d_j^2$, where $d_j$ is the number of constraints that involve variable $j$. ($ D \leq 3$ for most random MAX-$k$-CSP instances.) In particular, the question of whether or not the instance is fully satisfiable can be answered with probability $1-e^{-\Omega(n)}$ in time $O^*(2^{(0.5-c)n})$ setting $|E^*|/m = 1$ in the above expression for $c$. If additionally $k=3$ and $D=3$, we find that $c = 5.87 \times 10^{-7}$.
\end{theorem}
\begin{proof}
    We may assume $E^*$ is known. If it is not known, we may simply loop through all possible values that $E^*$ could possibly take, which incurs only polynomial overhead by the following argument. Note that, in the general setting, there are at most $m \leq n^k$ clauses and for each clause $\mathcal{C}_t$, the value of $\mathcal{C}_t$ is $-1$ when satisfied and $s_t/(2^k-s_t)$ when not satisfied, where $1 \leq s_t \leq 2^k-1$ denotes the number of satisfying assignments to $\mathcal{C}_t$. Thus, for every $t$ and every assignment, the value of $\mathcal{C}_t$ is an integral multiple of $[(2^k-1)!]^{-1}$, and the number of distinct values $E^*$ can possibly take is upper bounded by $(2^k-1)!m = \poly(n)$.
    
    We fix a value of $\eta$ to be specified later. By Prop.~\ref{prop:CSP_alpha-sub-depolarizing}, the pair $(H,g_\eta)$ has the $\alpha$-subdepolarizing property with 
    \begin{equation}
        \alpha = \frac{m}{|E^*|}\frac{k2^k}{(1-\eta)n}\, .
    \end{equation}
    Furthermore, by Prop.~\ref{prop:tail_bound_cases}, $H$ obeys a tail bound of the form of Eq.~\eqref{eq:tail_bound} with 
    \begin{equation}
        \gamma = \left(\frac{|E^*|}{m}\right)^2\frac{(1-\eta)^2}{\ln(2)2^{2k}k^2D}\,,
    \end{equation}
    where $D = nk^{-2}m^{-2}\sum_{j=1}^n d_j^2$, with $d_j$ the number of constraints involving binary variable $j$.
    By Lemma \ref{lem:tail_bound_implies_conditions}, this implies that the large-excited-energy and small-ground-energy-shift conditions are satisfied when we choose $b = \gamma\ln(2)/(2+\ln(2))$. By Theorem \ref{thm:super-Grover-speedup}, and substituting our choices for $b$ and $\gamma$ above, the runtime is $O^*(2^{(0.5-c)n})$ with
    \begin{equation}\label{eq:c_kCSP_proof}
        c = \left[\frac{1}{\ln(2)(2+\ln(2))}\frac{(1-\eta)^3F(1-\eta)}{\eta}\right]\left(\frac{|E^*|}{m}\right)^3\frac{1}{D2^{3k}k^3}\, .
    \end{equation}
    The expression in brackets achieves its maximum of $0.0289$ at $\eta = 0.189$, leading to the quoted statement. 

    Now consider the question of whether a certain $k$-CSP instance is fully satisfiable, which is equivalent to whether $E^* = -m$, or not. To solve this problem, it suffices to give an algorithm that can find a satisfying assignment, promised that one exists: we may run the algorithm assuming that $|E^*| = m$; if the instance is satisfiable then the output $y$ will satisfy $H(y)=-m$ with high probability, whereas if the instance is not satisfiable, the output $y$ will satisfy $H(y) > -m$ with probability 1.
\end{proof}

\begin{theorem}
    There is a quantum algorithm which, for at least a fraction $1-e^{-\Omega(n/k)}$ instances drawn from the $k$-spin ensemble, produces an optimal solution with probability at least $1-e^{-\Omega(n)}$ while running in time $O^*(2^{(0.5-c)n})$, where
    \begin{equation}
        c = (2.24 \times 10^{-4})\frac{1}{k^3}
    \end{equation}
\end{theorem}
\begin{proof}
    By Prop.~\ref{prop:MAX-Ek-LIN2-alpha-depolarizing}, any MAX-E$k$-LIN2 instance $H$ is $\alpha$-depolarizing (Def.~\ref{def:alpha-depolarizing}) with $\alpha = 2k/n$, and thus by Prop.~\ref{prop:depolarizing_implies_subdepolarizing}, for every $\eta$ the pair $(H,g_\eta)$ is $\alpha$-subdepolarizing (Def.~\ref{def:alpha-subdepolarizing}) with the same value of $\alpha$. 
    Furthermore, by Prop.~\ref{prop:tail_bound_cases}, $H$ obeys a tail bound of the form of Eq.~\eqref{eq:tail_bound} with 
    \begin{equation}\label{eq:gamma_kspin}
        \gamma = \frac{(1-\eta)^2}{32\pi \ln(2)k^2}\,.
    \end{equation}
    By Lemma \ref{lem:tail_bound_implies_conditions}, this implies that the large-excited-energy and small-ground-energy-shift conditions are satisfied when we choose $b = \gamma\ln(2)/(2+\ln(2))$. By Theorem \ref{thm:super-Grover-speedup}, and substituting our choices for $b$ and $\gamma$ above, the runtime is $O^*(2^{(0.5-c)n})$ with
    \begin{equation}
        c = \left[\frac{1}{64\ln(2)\pi(2+\ln(2))}\frac{(1-\eta)^2F(1-\eta)}{\eta}\right]\frac{1}{k^3}\,.
    \end{equation}
    The expression in brackets achieves its maximum of $2.24 \times 10^{-4}$ at $\eta = 0.405$, proving the theorem.
\end{proof}

\section{Physical intuition for algorithm and speedup}

This section gives an intuition for why the algorithm works and why it is able to give a speedup over Grover's algorithm.

\subsection{A rugged cost landscape}

A local structure can be established by placing the $2^n$ assignments to the cost function $H$ on the corners of a hypercube in $n$-dimensions. Flipping a bit in the assignment corresponds to traversing an edge of the hypercube from one corner to an adjacent corner. The Hamming distance between two assignments is defined as the minimum number of edges that must be traversed to get from one to the other. An assignment is a local minimum if its cost is lower than all adjacent assignments on the hypercube. Random instances of cost functions like MAX-E$k$-LIN2 and MAX-$k$-SAT typically give rise to a rugged landscape, where there are an exponential number of local minima and an exponential number of deep cost valleys (i.e.~achieving a constant fraction of the optimal cost) separated from each other by Hamming distance at least order $n$. This property limits the effectiveness of any algorithm that aims to find the global minimum by local search, an intuitive claim that has also been formalized in works such as Ref.~\cite{gamarnik2021overlap}.

By convention, the cost function is offset such that the average cost across all assignments is 0; thus the optimal cost $E^*$ is negative. Dividing $H$ by $|E^*|$ normalizes the cost function such that its deepest valley has cost $-1$. Applying $g_\eta$ to $H/|E^*|$ smooths out the cost function, setting nearly all assignments to zero cost, while preserving the local structure of the deepest valleys. One way to visualize it is that $g_\eta$ has the effect of flooding the entire landscape with water, such that the water level lies at a factor $1-\eta$ times the maximum valley depth. All dry land is set to zero cost; meanwhile, the local structure of the cost landscape underwater is preserved up to a stretching factor of $1/\eta$ such that the range of the cost function $g_\eta(H/|E^*|)$ remains $[-1,0]$. The reason for introducing this function $g_\eta$ is elaborated upon in Sec.~\ref{sec:reason_for_eta}.

\subsection{Quantum phase transition from extended to localized phase}

With $\eta$ fixed, $H_b$ is defined in Eq.~\eqref{eq:Hb} as a sum of the transverse field $-X/n$ and the diagonal operator $g_\eta(H/|E^*|)$. We may view $b$ as a free parameter that drives one or more quantum phase transitions. When $b=0$, we have $H_b = -\frac{X}{n}$ and $\ket{\psi_b} = \ket{\boldplus}$.
The state $\ket{\boldplus}$ might be called an \emph{extended} state because it has substantial amplitude, in fact \emph{equal} amplitude, on all $2^n$ computational basis states.

As we increase $b$, the diagonal term  $b \, g_\eta(H_b/|E^*|)$ begins to play a role. The spectrum of this term lies in the range $[-b,0]$, whereas the first term $-\frac{X}{n}$ has range $[-1,1]$. Thus, when $b$ is small, the second term may be viewed as a perturbation of the first, and we expect $\ket{\psi_b}$ to retain high overlap with $\ket{\boldplus}$ and continue to be an extended state with substantial overlap across all $2^n$ computational basis states. However, the distribution of amplitude is no longer uniform; assignments that lie within the valleys of $g_\eta(H_b/|E^*|)$ will collect more amplitude than assignments not lying within a valley. Note that, as previously mentioned, since $H_b$ is a stoquastic Hamiltonian, we may take $\ket{\psi_b}$ to have positive amplitudes in the computational basis \cite{bravyi2006complexity}.  

On the other end, in the limit $b \rightarrow \infty$, the first term vanishes relative to the second. The ground state $\ket{\psi_b}$ becomes \emph{localized} onto a single computational basis state $\ket{z^*}$, associated with the optimal assignment to the cost function (or if there are multiple optimal values, then there will be a degenerate ground state in this limit). For large but finite $b$, the state $\ket{\psi_b}$ will have some support everywhere, but it will still be localized in the sense that most of the amplitude will be on $\ket{z^*}$ with amplitude on other computational basis states exponentially small in the Hamming distance from $\ket{z^*}$.

These observations suggest that increasing $b$ drives a quantum phase transition from an extended phase to a localized phase at some critical value of $b$. This is a second-order phase transition, analogous to the thermal phase transition in a spin glass. When does the critical point occur? We know that it must occur at some value $b < 1$. This observation follows from the fact that, at $b=1$, both terms $-X/n$ and $b \, g_\eta(H/|E^*|)$ have equal operator norm, and $\ket{\boldplus}$ and $\ket{z^*}$ have similar expectation values for $H_b$ despite the former being extended and the latter being localized. Our formal proofs, in particular Lemma \ref{lem:tail_bound_implies_conditions}, showed that under the tail-bound assumption, the ground and first-excited states of $H_b$ remain well separated up until an $n$-independent value of $b$, which puts a lower bound on where the critical point for this quantum phase transition occurs. 

Note that there can also be first-order phase transitions within the localized phase, where increasing $b$ can cause the wavefunction $\ket{\psi_b}$ to change from being localized within one valley of the landscape to being localized in another far-away valley. In the context of the quantum adiabatic algorithm, these first-order phase transitions have been identified as the truly problematic part of the protocol \cite{amin2009firstOrderAdiabatic,young2010firstOrderAdiabatic,altshuler2010AndersonLocalization}, as the size of the minimum spectral gap is hard to study and exponentially small in $n$. Our algorithm avoids the complication of first-order phase transitions within the localized phase by \textit{jumping straight from the extended phase to the end of the algorithm}: as long as $b$ is chosen at some point smaller than the critical value, $H_b$ will have a large spectral gap, and $\ket{\psi_b}$ will be extended. The speedup leverages the fact that while $\ket{\psi_b}$ has substantial amplitude on all $2^n$ computational basis states, it has a factor $2^{cn}$ larger amplitude on the computational basis states lying within the deepest valleys of the landscape. Thus, by measuring $\ket{\psi_b}$ in the computational basis (and performing amplitude amplification), we can find the optimal assignment with a $2^{cn}$ advantage over Grover's algorithm. 

\subsection{The $\alpha$-depolarizing property and the super-Grover speedup}

As Grover's algorithm is optimal in a black-box model, any algorithm achieving a super-Grover speedup must exploit structure in the cost function. In Hastings' work on the short-path algorithm \cite{hastings2018shortPath}, the cost function studied was MAX-E$k$-LIN2, in part because it has the following nice property: for any assignment $z$ with cost value $H(z)$, the average cost of the $n$ neighbors of $z$ on the hypercube is exactly $(1-\alpha)H(z)$, with $\alpha = 2k/n$; in other words, traversing a random edge brings the expected cost value toward zero by a small and consistent factor. Hastings used this property in combination with a perturbation theory analysis to show that a certain overlap that determines the runtime of the short-path algorithm is $2^{\Omega(b/\alpha)}$ larger than what occurs in Grover's algorithm, under certain assumptions. This advantage is exponentially large in $n$ as long as the assumptions hold for constant $b$. Proving that these additional assumptions hold was the weaker part of Hastings' analysis.

We call this property the $\alpha$-depolarizing property (Def.~\ref{def:alpha-depolarizing}), and then we define a generalized version: the $\alpha$-\textit{subdepolarizing} property (Def.~\ref{def:alpha-subdepolarizing}). The idea is that cost functions consisting of $k$-local terms may not exactly satisfy the $\alpha$-depolarizing property, but there is still a sense in which a single random bit flip will not cause the energy to change much on average. Indeed, we are able to show that MAX-$k$-SAT instances that are not highly frustrated (together with the function $g_\eta$) have the $\alpha$-subdepolarizing property where $\alpha = O(k2^k/n)$ (see Prop.~\ref{prop:CSP_alpha-sub-depolarizing}). While the overlap we want to compute, $\braket{\psi_b}{z^*}$, is different than the one Hastings' needed to compute in Ref.~\cite{hastings2018shortPath}, we show that the $\alpha$-subdepolarizing property, along with the small-ground-energy-shift condition (Cond.~\ref{cond:small-ground-energy-shift}), implies a $2^{\Omega(b/\alpha)} = 2^{cn}$ advantage over Grover's algorithm, with $c = \Omega(b/(k2^k))$. Thus, as long as the small-ground-energy shift is met for an $n$-indepenent choice of $b$, the super-Grover speedup quantity $c$ will be independent of $n$. Where Hastings' main tool was perturbation theory, we use an approximate ground-state projector and a combinatorial analysis (see Sec.~\ref{sec:bounding_overlap}). 

\subsection{The role of $\eta$ and establishing a lower bound on the critical point}\label{sec:reason_for_eta}

Since the magnitude of the advantage in the exponent over Grover's algorithm is $cn$ for a value $c$ proportional to $b$, we need to establish that we can take $b$ to be a constant independent from $n$ and still be in the extended phase if we want to have a super-Grover speedup. It is not obvious that this would be the case, and previous studies of phase transitions in the adiabatic algorithm suggested the possibility that the critical point could occur at a value of $b$ that shrinks with $n$ \cite{altshuler2010AndersonLocalization,wecker2016trainingQuantumOptimizer}, although this is not expected to be the typical behavior \cite{knysh2010relevance}. 

By applying $g_\eta$ to the cost function $H/|E^*|$, we can provably avoid this possibility and show that the extended phase persists up to an $n$-independent choice of $b$. The idea is again to view $b\,g_\eta(H/|E^*|)$ as a perturbation to $-X/n$. Assuming a tail bound of the form of Eq.~\eqref{eq:tail_bound}, $g_\eta(H(z)/|E^*|) = 0$ for all but an exponentially small fraction of assignments $z$. As a result, the ground state does not \emph{see} the perturbation in the sense that $\bra{\boldplus}g_\eta(H(z)/|E^*|) \ket{\boldplus} $ is exponentially small. The smaller we make $\eta$, the weaker the perturbation appears, and the larger the value of $b$ we can choose and still be sure the ground state will be extended. However, the magnitude of the speedup is roughly proportional to $\eta$, so we do not want to make $\eta$ too small. The formal proof of this ingredient utilizes the log-Sobolev inequality and tools from statistical mechanics (see Sec.~\ref{sec:showing_conditions}).

\subsection{Mechanism for speedup: localization in 1-norm vs.~2-norm}\label{sec:speedup_mechanism}
In quantum mechanics, the Born rule dictates that probabilities of measurement outcomes are the square of amplitudes of the wavefunction. In other words, if we want to measure the probability that the wavefunction lies within a certain subset of the $2^n$ vertices of the hypercube, we should use the 2-norm. This contrasts with standard probability theory, where the 1-norm would be used.  

In describing the mechanism for speedup for the short-path algorithm \cite{hastings2018weaker}, Hastings pointed to precisely this fact. In the short-path algorithm, a state $\ket{\phi}$ was prepared that was simultaneously localized in the 2-norm and de-localized in the 1-norm, in the following sense. Define
\begin{equation}
w_1(z) = \frac{\lvert \braket{z}{\phi}\rvert}{\sum_z \lvert \braket{z}{\phi}\rvert} \qquad \qquad
    w_2(z) = \frac{\lvert \braket{z}{\phi}\rvert^2}{\sum_{z} \lvert \braket{z}{\phi}\rvert^2}\,,
\end{equation} 
both of which are probability distributions over inputs $z$. For the state $\ket{\phi}$ prepared by the algorithm, the distribution $w_2$ had nearly all its mass concentrated on the optimal solution $z^*$, while the distribution $w_1$ had an exponentially small $2^{-cn}$ fraction of its mass concentrated on $z^*$. In fact, if we break the hypercube up into subsets of equal Hamming distance from $z^*$, the subset with the most $w_1$ mass was order-$n$ bit flips away from $z^*$. This 1-norm delocalization is what gave rise to the super-Grover speedup during the large jump made by the algorithm, while the 2-norm localization allowed the small jump to succeed with high probability. 

Our algorithm exploits a similar effect to produce a super-Grover speedup, although it does so in the Hadamard basis. That is, the same comments for $\ket{\phi}$ above apply to the state $\ket{\psi_b}$ if we redefine
\begin{equation}\label{eq:w1_w2_our_algo}
w_1(u) = \frac{\lvert \braket{u}{\psi_b}\rvert}{\sum_u \lvert \braket{u}{\psi_b}\rvert} \qquad \qquad
    w_2(u) = \frac{\lvert \braket{u}{\psi_b}\rvert^2}{\sum_{u} \lvert \braket{u}{\psi_b}\rvert^2}\,,
\end{equation} 
where the states $\ket{u}$ are the $2^n$ $n$-fold tensor products of $\ket{+}$ and $\ket{-}$. We have already discussed how the state $\ket{\psi_b}$ is extended and has high overlap with $\ket{\boldplus}$: thus, it is 2-norm localized in the Hadamard basis. Now, without loss of generality, we assume the optimal solution is the computational basis state $\ket{0^n} = 2^{-n/2}(\ket{+} + \ket{-})^{\otimes n}$. Then, the overlap that determines the runtime of our algorithm is given by
\begin{equation}
    \braket{0^n}{\psi_b} = 2^{-n/2}\sum_{u \in \{+,-\}^n} \braket{u}{\psi_b}\,.
\end{equation}
The statement that $\braket{0^n}{\psi_b} \geq 2^{-(0.5-c)n}$ is thus equivalent to the statement $\sum_u \braket{u}{\psi_b} \geq 2^{cn}$. This implies the denominator of the definition of $w_1$ in Eq.~\eqref{eq:w1_w2_our_algo} is at least $2^{cn}$, but meanwhile $\braket{\boldplus}{\psi_b}\approx 1$, meaning $w_1(\boldplus)\leq 2^{-cn}$. In other words, an exponentially small amount of 1-norm weight lies on the $\ket{\boldplus}$ basis state.

Crucially, this dichotomy between the distribution of weight in the 1-norm and the 2-norm is only possible in quantum mechanics; it is not obvious how one would exploit this phenomenon with a classical algorithm. As a result, this mechanism has the potential to deliver super-Grover speedups, and, potentially, super-quadratic speedups compared to best-known classical algorithms. 

\subsection{Performance of analogous classical Markov chain methods}

While a classical algorithm cannot exactly replicate the same 1-norm vs.~2-norm dichotomy that our quantum algorithm exploits, it is valuable to scrutinize whether a classical algorithm might nonetheless emulate our quantum algorithm. We identify two categories of classical algorithms based on sampling Markov processes that might, in a certain sense, be regarded as classical analogues of our algorithm. However, in both cases, they lack rigorous runtime guarantees similar to those we have shown about our algorithm.

The first candidate analgoue is the classical algorithm mentioned in Sec.~\ref{sec:comparison} that repeatedly samples the high-temperature classical Gibbs distribution of the cost function (or some transformation applied to the cost function), a feat which can be accomplished via classical Metropolis sampling or simulated annealing. Showing that this classical algorithm runs in time $O^*(2^{(1-c)n})$ would require proving that the associated Markov chain remains rapidly mixing up to some sufficiently large inverse temperature $\beta$.  As rapid-mixing proofs often leverage the log-Sobolev inequality and its relatives (e.g., \cite{stroock1992equivalence,gheissari2019sphericalSpinGlass,eldan2022spectral}), perhaps it is unsurprising that the log-Sobolev inequality is also helpful for us in showing lower bounds on the spectral gap of the Hamiltonian $H_b$. However, despite a heuristic connection between the inverse temperature $\beta$ in the Metropolis sampling algorithm and the inverse transverse-field strength parameter $b$ in our algorithm, there is no direct relationship between the runtime of the two algorithms.  
Indeed, a clearer quantum analogue of simulated annealing is given by applying the ``quantum simulated annealing'' algorithm proposed in  Ref.~\cite{somma2008annealing} (see also Refs.~\cite{aharonov2003adiabaticStateGeneration,wocjan2008speedupQSampling,boixo2015fastMethodsOptimization,lemieux2020efficient,sanders2020compilationHeuristics}), which, when combined with amplitude amplification as in Ref.~\cite{montanaro2015MonteCarlo},  gives a direct quadratic (but no larger) quantum speedup of the classical algorithm described above. 

The second candidate analogue of our quantum algorithm is formed by applying a Quantum Monte Carlo (QMC) approach to the Hamiltonian $H_b$ (see, e.g., \cite{farhi2009smallGapsDifferentPaths,bravyi2015qmcStoquastic,jarret2016diffusionMonteCarlo, crosson2016simulatedQuantumAnnealing,crosson2020classicalSimulationIsingModels}). A reason to think that this might be successful is that our Hamiltonian $H_b$ is stoquastic and thus it does not have the sign problem \cite{bravyi2006complexity}. Indeed, the core step of our algorithm is the application of quantum phase estimation to the stoquastic Hamiltonian $H_b$, a situation that was shown in Ref.~\cite{bravyi2015qmcStoquastic} to admit a classical QMC simulation whenever there exists a ``guiding state.'' In the case of $H_b$, there is no obvious state that satisfies the criteria of Ref.~\cite{bravyi2015qmcStoquastic}: the state $\ket{\boldplus}$ falls short since it differs from $\ket{\psi_b}$ by an exponentially large factor on some of the amplitudes. More generally, showing that the QMC algorithm can sample (up to inverse exponential precision) the same distribution as computational basis measurements on the ground state of $H_b$ requires proving that a certain Markov chain representing the imaginary time evolution $e^{-\beta H_b}$ remains rapidly mixing up to imaginary time $\beta$ of order-$n^2$ (this is necessary to guarantee that the second largest eigenvalue of $e^{-\beta H_b}$ is exponentially smaller than the largest eigenvalue). For comparison, Ref.~\cite{crosson2020classicalSimulationIsingModels} showed rapid mixing for transverse-field Ising models up to imaginary time of order $(J(\Delta+2))^{-1}$, where $J$ is the maximum interaction strength and $\Delta$ is the interaction-degree. This encompasses the SK model (i.e.~Eq.~\eqref{eq:k-spin} with $k=2$), where, after normalization by $|E^*|$ (order-$n$), we have $J^{-1}=O(n\sqrt{n})$ and $\Delta^{-1} = 1/n$, so this fails to prove a relevant statement for our situation. However, it remains an interesting question to examine the performance of QMC as a classical simulation of our algorithm. Toward that end, it is interesting to note that previous work (see fourth counterexample of Ref.~\cite{hastings2013obstructions}, as well as Refs.~\cite{jarret2016diffusionMonteCarlo,gilyen2021ExpAdvAdiabStoqQCSTOC}) has demonstrated examples where the classical QMC algorithm struggles to simulate its quantum counterpart due in part to a discrepancy between wavefunction localization in the 1-norm versus the 2-norm, the exact effect exploited by our quantum algorithm (albeit in a different basis). 

\section{Numerical results}\label{sec:numerics}

\subsection{Example instance}

\begin{figure}[h]
    \centering
    \includegraphics[draft=false,width = 0.6\textwidth]{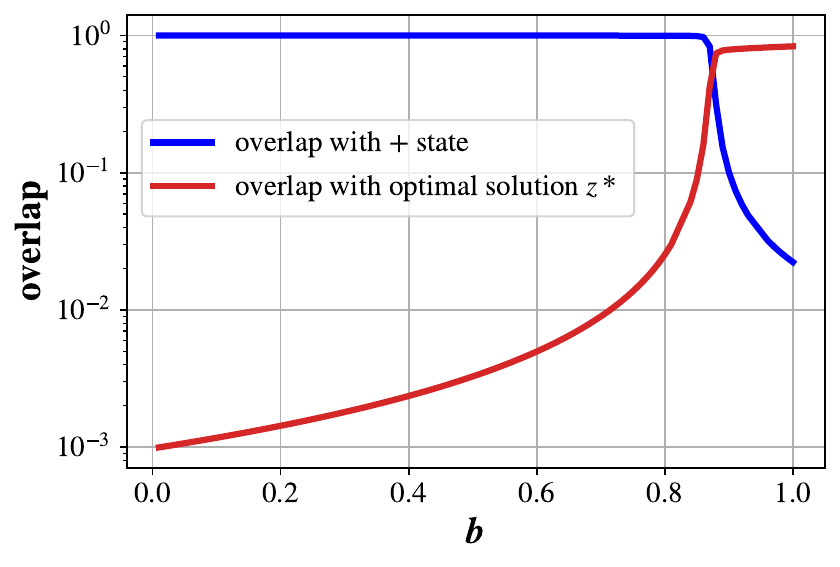}
    \caption{Plot of the relevant overlap values for the same $n=20$, $\eta=0.5$ 3-spin instance from Fig.~\ref{fig:spectrum_n20}. Overlaps are determined by exact diagonalization of $H_b$. The overlap $\lvert\braket{\boldplus}{\psi_b}\rvert$ (blue) determines the runtime of step 2 of Algorithm \ref{algo:main_simple}, and the overlap $\lvert\braket{z^*}{\psi_b}\rvert$ (red) determines the runtime of step 3.
    }
    \label{fig:overlap_n20}
\end{figure}

We present an example instance to confirm that the theoretical picture of the algorithm holds in practice. Recall that in Fig.~\ref{fig:spectrum_n20}, we presented the lowest three energy levels of $H_b$ as a function of $b$ for an $n=20$ instance chosen from the 3-spin ensemble, with $\eta=0.5$. In Fig.~\ref{fig:overlap_n20}, we show the overlap profile for the same instance, that is $\lvert\braket{\boldplus}{\psi_b}\rvert$ and $\lvert\braket{z^*}{\psi_b}\rvert$ as a function of $b$.

Our theoretical picture is largely confirmed. The ground state energy stays remarkably unshifted from $-1$ and the spectral gap remains open until $b$ is quite large, around $b=0.85$ for this instance.  Meanwhile, the overlap $\lvert \braket{z^*}{\psi_b}\rvert$ which determines the runtime of the algorithm grows exponentially with $b$, even as the overlap $\lvert\braket{\boldplus}{\psi_b}\rvert$ remains close to 1 (the wavefunction remains localized in the 2-norm even as it becomes increasingly de-localized in the 1-norm, see Sec.~\ref{sec:speedup_mechanism}). This confirms that for this instance the first step of the algorithm will succeed with little need for amplification even when we take $b$ as large as 0.8, and that the second step has success probability significantly better than Grover's algorithm. 

This instance also illustrates how there might be considerable room for improvement in the theoretical analysis. Firstly, in Fig.~\ref{fig:overlap_n20}, the quantity $\lvert \braket{z^*}{\psi_b}\rvert$ appears to grow \textit{super}-exponentially with $b$, even as $\lvert \braket{\boldplus}{\psi_b}\rvert$ remains close to 1, whereas the theoretical lower bound from Lemma \ref{lem:overlap_bound_Pl} predicts only exponential growth. Secondly, the numerics indicate that the algorithm will work well for this $\eta=0.5$ instance all the way up to $b=0.8$, whereas the theoretical analysis only guarantees success up to a much smaller value of $b$. Looking back at Eq.~\eqref{eq:gamma_kspin} and the text beneath, for $k=3$ and $\eta=0.5$, the large-excited-energy condition (Cond.~\ref{cond:large-excited-energy}) is shown to hold only for $b \leq 1.02 \times 10^{-4}$ (and only for \textit{most} random instances in the ensemble). From the theory perspective, the important conclusion is that this value is independent of $n$, but these numerics suggest that there is a considerable gap between the theoretical bounds and the empirical reality. 

\subsection{Estimation of actual size of super-Grover speedup}

Our theoretical analysis shows a super-Grover speedup but with a very small lower bound on the size of the speedup. We illustrate numerically that the actual super-Grover speedup might be much more substantial. By performing exact diagonalization on 30 random $\eta=0.5$, $b=0.7$ instances from the 3-spin ensemble at each value $n=17,18,\ldots,23$, we determine the growth of the advantage with $n$. We chose $b=0.7$ to give some breathing room between the maximum value at which the large-excited-energy condition (Cond.~\ref{cond:large-excited-energy}) failed (roughly $b=0.8$) for the instance depicted in Fig.~\ref{fig:spectrum_n20}. Indeed, we found that for all of the instances we diagonalized, $b=0.7$ fell comfortably below the phase transition point and the large-excited-energy condition was satisfied. Thus, by Theorem \ref{thm:runtime}, the algorithm would succeed with runtime dependent on $\lvert \braket{z^*}{\psi_{0.7}} \rvert ^{-1}$. In Fig.~\ref{fig:inv_overlap_with_n_b0.7_eta0.5}, we plot the quantity $\lvert \braket{z^*}{\psi_{0.7}} \rvert ^{-1}$ for each instance, as well as the median among all 30 instances at each value of $n$. Fitting the medians to a line on a log scale, we find that the best fit is $\lvert \braket{z^*}{\psi_{0.7}} \rvert ^{-1} \approx 0.28 \times 2^{0.427n}$; the 95\% confidence interval on the 0.427 value is $[0.415,0.439]$. While the quality of the numerical fit is encouraging, we caution that robust conclusions are difficult to assert based on an exponential fit to just 7 data points that span the small range from roughly $40$ to roughly $300$, less than even one order of magnitude. 


A scaling of $O^*(2^{0.43n})$ would represent a material improvement over Grover, if not a practically transformative one. A larger speedup, closer to quartic compared with exhaustive enumeration (corresponding to a quantum algorithm with $2^{0.25n}$ runtime), would be needed to make the prospect of actual quantum advantage more likely \cite{babbush2021focus}. However, as we have not made much effort at optimizing the parameters $b$ and $\eta$, we are optimistic that future examination will reveal additional speedup over Grover than what we report here. 

\begin{figure}[h]
    \centering
    \includegraphics[draft=false,width = 0.6\textwidth]{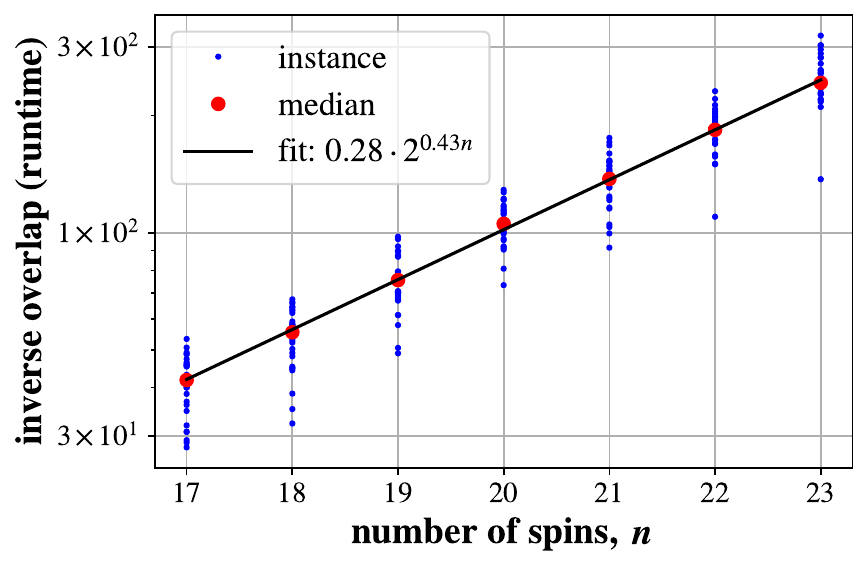}
    \caption{Plot of $\lvert \braket{z^*}{\psi_b} \rvert ^{-1}$ for several randomly chosen instances from the $3$-spin model at $b=0.7$ and $\eta = 0.5$ at values of $n$ ranging from $n=17$ to $n=23$. For each value of $n$, the median value is also plotted. A fit of the medians to an exponential yields the fit $\lvert \braket{z^*}{\psi_b} \rvert ^{-1} = 0.276 \cdot 2^{0.427n}$. }
    \label{fig:inv_overlap_with_n_b0.7_eta0.5}
\end{figure}

\section{Conclusion}

We present a quantum algorithm that has a rigorous proof of a speedup compared to Grover's algorithm in certain cases. On the one hand, as it stands, the algorithm is not likely to lead to a practical advantage. For one, the size of the speedup in our proof is very small, although we believe that further optimization of parameters and extensions of the theory would lead to improvements. Additionally, while numerical experiments suggest that the actual speedup is non-negligible, it still appears to fall short of a cubic or quartic speedup (in comparison to exhaustive enumeration), and thus, after accounting for other overheads, it is unlikely to lead to practical benefit without further improvements \cite{babbush2021focus}. Finally, in most cases where our algorithm has rigorous guarantees, the speedup over the \textit{best} classical algorithm is less than quadratic, owing to the existence of clever classical algorithms that exploit structure to beat exhaustive enumeration.

On the other hand, the mechanism by which our algorithm generates a speedup has no obvious classical analogue. This contrasts with many other speedups over exhaustive enumeration, which are essentially a Grover speedup combined with a classical technique. The inherent quantumness of our speedup positions our algorithm to at least have a fighting chance at achieving a super-quadratic speedup, even if it cannot yet do so when compared to mature classical algorithms that have been improved over the course of decades. Toward that end, an important open question is whether our algorithm and the speedup mechanism behind it can be combined with any of the classical techniques for exploiting problem structure, such as backtracking \cite{montanaro2018backtracking} and branch-and-bound \cite{montanaro2020branchAndBound}.

\section*{Acknowledgments}

We are grateful to Mario Berta, Andr\'as Gily\'en, Michael Kastoryano, Sam McArdle, Ashley Montanaro, and Grant Salton for helpful conversations. We thank Fran\c{c}ois Le Gall for pointing out an error in Theorem 7 of version 1 (see Footnote \ref{footnote:error}). 
\appendix
\section{Deferred proofs}
\subsection{Tail bound implies large-excited-energy and small-ground-energy-shift conditions}

\subsubsection{Short-path condition definition and consequences}
First, we define a condition we call the ``short-path condition'' since it is analagous to the main condition used by Hastings in the analysis of the short-path algorithm Ref.~\cite{hastings2018shortPath}.

\begin{condition}[Short-path condition]\label{cond:short_path}
Let $H_b$ be a Hamiltonian as defined in Eq.~\eqref{eq:Hb}. Let $\bar{\Pi} = I - \ket{\boldplus}\bra{\boldplus}$ to be a projector onto the excited space of $-X/n$ and let $\bar{E}_{b}$ denote the ground state energy of the Hamiltonian
\begin{equation}
    \bar{\Pi} H_b\bar{\Pi} \,.
\end{equation}
Then $H_b$ is said to obey the \textit{short-path condition} if
\begin{equation}
    \text{Short-path condition:   } \bar{E}_{b} \geq -1+\frac{1}{n} \,.
\end{equation}
\end{condition}

\begin{proposition}\label{prop:short_path->(i)}
If $H_b$ satisfies the short-path condition (Cond.~\ref{cond:short_path}), then $H_b$ satisfies the large-excited-energy condition (Cond.~\ref{cond:large-excited-energy}).
\end{proposition}
\begin{proof}
    Suppose for contradiction that the large-excited-energy condition were false. Then there are at least two eigenvectors of $H_b$ with eigenvalue beneath $-1+1/n$. Denote two such eigenvectors by $\ket{\psi_b}$ and $\ket{\eta_b}$. Furthermore, since $\bra{\boldplus}H_B\ket{\boldplus} \leq -1$, we may assume without loss of generality that at least one of these eigenvectors has nonzero overlap with $\ket{\boldplus}$. We define the normalized nonzero state
    \begin{equation}
        \ket{Y} = \frac{\braket{\boldplus}{\eta_b}\ket{\psi_b} - \braket{\boldplus}{\psi_b}\ket{\eta_b}}{\lvert \braket{\boldplus}{\eta_b}\rvert^2 + \lvert \braket{\boldplus}{\psi_b} \rvert^2}
    \end{equation}
    such that it is orthogonal to $\ket{\boldplus}$ and thus lies in the support of $\bar{\Pi}$, and meanwhile has expectation value of $\bra{Y}\bar{\Pi} H_b \bar{\Pi}\ket{Y}$ that is less than $-1+1/n$. By the variational principle, the ground state of $\bar{\Pi} H_b \bar{\Pi}$ (restricted to the support of $\bar{\Pi}$) must have smaller expectation value than $\ket{Y}$, which contradicts the short-path condition (Cond.~\ref{cond:short_path}).
\end{proof}

\begin{proposition}\label{prop:short_path->(ii)}
    Suppose that $H_b$ satisfies the short-path condition (Cond.~\ref{cond:short_path}) and $H$ obeys the cumulative states tail bound
    \begin{equation}
        C((1-\eta)E^*) \leq 2^{(1-\gamma) n}
    \end{equation}
    where $C(E) = |\{z: H(z) \leq E\}|$. Then 
    \begin{equation}
        E_b \geq -1 - 2^{-\gamma n}(b+b^2n)
    \end{equation}
    Thus, if $\gamma$, $b$ are independent of $n$, the small-ground-energy-shift condition (Cond.~\ref{cond:small-ground-energy-shift}) is satisfied for sufficiently large $n$. 
\end{proposition}
\begin{proof}
    Let
    \begin{equation}
        W_b = \left(\bar{\Pi}(E_b-H_b)\bar{\Pi}\right)^{-1}
    \end{equation}
    where the inverse is the Moore-Penrose pseudo-inverse. In the following, let $G_\eta$ denote the operator $g_\eta(H/\lvert E^*\rvert)$. We claim that
    \begin{equation}
        E_b = -1 - b\bra{\boldplus}G_\eta \ket{\boldplus} - b^2 \bra{\boldplus}G_\eta W_b G_\eta \ket{\boldplus}\,.
    \end{equation}
    This is justified as follows. First, define the (non-normalized) state $\ket{U}$ as follows:
    \begin{equation}
        \ket{U} := \ket{\boldplus} + \frac{\bar{\Pi}\ket{\psi_b}}{\braket{\boldplus}{\psi_b}} \qquad \implies \qquad \braket{\boldplus}{\psi_b}\ket{U} = \left(\ketbra{\boldplus}+\bar{\Pi}\right)\ket{\psi_b} = \ket{\psi_b}
    \end{equation}
    which by noting the resolution of the identity $I = \ketbra{\boldplus} + \bar{\Pi}$, is seen to be proportional to $\ket{\psi_b}$. Since by definition $\ket{\psi_b}$ has eigenvalue $E_b$ under $H_b$, we can make the following assertions. 
    \begin{align}
        H_b\ket{U} &= E_b \ket{U} \\
        H_b\left(\ket{\boldplus} + \frac{\bar{\Pi}\ket{\psi_b}}{\braket{\boldplus}{\psi_b}}\right)&= E_b\left(\ket{\boldplus} + \frac{\bar{\Pi}\ket{\psi_b}}{\braket{\boldplus}{\psi_b}}\right) \\
        (E_b-H_b)\ket{\boldplus} &= (H_b-E_b)\bar{\Pi} \frac{\ket{\psi_b}}{\braket{\boldplus}{\psi_b}}\\
        \bar{\Pi}(E_b-H_b)\ket{\boldplus} &= \bar{\Pi}(H_b-E_b)\bar{\Pi} \frac{\ket{\psi_b}}{\braket{\boldplus}{\psi_b}}\\
        \bar{\Pi}(E_b+(X/n) - b G_\eta)\ket{\boldplus} &= \left(\bar{\Pi}(H_b-E_b)\bar{\Pi}\right)\bar{\Pi} \frac{\ket{\psi_b}}{\braket{\boldplus}{\psi_b}} \\
        -\bar{\Pi}(b G_\eta)\ket{\boldplus} &= \left(\bar{\Pi}(H_b-E_b)\bar{\Pi}\right)\bar{\Pi} \frac{\ket{\psi_b}}{\braket{\boldplus}{\psi_b}} \\
        -bW_b G_\eta\ket{\boldplus} &=  \frac{\bar{\Pi}\ket{\psi_b}}{\braket{\boldplus}{\psi_b}}\,,
    \end{align}
    so, referring back to the definition of $\ket{U}$, we have
    \begin{equation}
        \ket{U} = \ket{\boldplus} - b W_b G_\eta \ket{\boldplus}\,.
    \end{equation}
    Note also that since $H_b\ket{U} = E_b\ket{U}$, we have
    \begin{align}
        E_b &= \frac{\bra{\boldplus} H_b \ket{U}}{\braket{\boldplus}{U}} = \bra{\boldplus} H_b \ket{U}  \\
        &= \bra{\boldplus} \Big(-X/n + bG_\eta\Big)\Big(\ket{\boldplus} - b W_b G_\eta\ket{\boldplus}\Big)\\
        &= -1-b \bra{\boldplus} G_\eta \ket{\boldplus} - b^2 \bra{\boldplus}G_\eta W_b G_\eta \ket{\boldplus}\,.
    \end{align}
Now, by the short-path assumption, the minimum eigenvalue of $\bar{\Pi}H_b\bar{\Pi}$ is greater than $-1+1/n$. Since $E_b \leq -1$, we can assert that the minimum eigenvalue of $\bar{\Pi}(H_b-E_b)\bar{\Pi}$ is greater than $1/n$, and thus $W_b$ has operator norm upper bounded by $n$. Thus,
\begin{equation}\label{eq:Eb_perturbation_bound}
    E_b \geq -1 + b \bra{\boldplus} G_\eta \ket{\boldplus} - b^2 n \bra{\boldplus} G_\eta^2 \ket{\boldplus} \,.
\end{equation}

Recall that $G_\eta$ is a diagonal operator, but by assumption, most of its $2^n$ diagonal entries are zero. There are at most $2^{(1-\gamma)n}$ strings $z$ for which $G_\eta \ket{z} \neq 0$, and $\bra{z}G_\eta \ket{z} \geq -1$ for all $z$. Since $\ket{\boldplus} = 2^{-n/2}\sum_z{\ket{z}}$, this implies 
\begin{align}
    0 \leq -\bra{\boldplus} G_\eta \ket{\boldplus} \leq 2^{-\gamma n} \\
    0 \leq \bra{\boldplus} G_\eta^2 \ket{\boldplus} \leq 2^{-\gamma n}\,,
\end{align}
which together with Eq.~\eqref{eq:Eb_perturbation_bound} gives the proposition. 
\end{proof}

\begin{proposition}[log-Sobolev inequality]\label{prop:log-sobolev}
    Given a state $\ket{\phi}$, let $p:\{+1,-1\}^n\rightarrow [0,1]$ denote the probability distribution over measurement outcomes when $\ket{\phi}$ is measured in the computational basis, that is
    \begin{equation}
        p(z) = \lvert \braket{z}{\phi}\rvert^2\,.
    \end{equation}
    and let $S = -\sum_z p(z)\log_2(p(z))$ be the base-2 Shannon entropy of $p$. Then
\begin{equation}\label{eq:log-Sobolev}
    S/n \geq \tau^{-1}(\bra{\phi}X/n \ket{\phi})
\end{equation}
where 
\begin{equation}\label{eq:tau}
    \tau^{-1}(y) = H_2\left(\frac{1}{2} - \frac{1}{2}\sqrt{1-y^2}\right)
\end{equation}
with $H_2$ denoting the binary entropy function $H_2(q) = -q\log_2(q) - (1-q) \log_2(1-q)$. Moreover, $\tau^{-1}(y) \geq 1+(y-1)/\ln(2)$, and thus
\begin{equation}\label{eq:log-Sobolev-first-order}
    S/n \geq 1-\ln(2)^{-1}(1-\bra{\phi}X/n \ket{\phi})\,.
\end{equation}
\end{proposition}
\begin{proof}
    We import the log-Sobolev inequality directly from prior literature. The bound we use originates in Ref.~\cite{samorodnitsky2008modified}, and we state it identically as it appeared in Hastings' short-path analysis \cite[Lemma 11]{hastings2018shortPath}:
    \begin{equation}
        \tau(S/n) \geq \bra{\phi}X/n \ket{\phi}
    \end{equation}
    where $\tau(x):= 2\sqrt{H_2^{-1}(x)(1-H_2^{-1}(x))}$ can be verified to be inverse of the $\tau^{-1}$ function defined above. This implies Eq.~\eqref{eq:log-Sobolev}. To show the relation $\tau^{-1}(y) \geq 1-\ln(2)^{-1}(1-y)$, we first show that $\tau^{-1}$ is convex by proving that its second derivative is non-negative. Define $q(y) := (1-\sqrt{1-y^2})/2$ and note that $q(1-q) = y^2/4$. 
    \begin{align}
        \frac{d^2 \tau^{-1}(y)}{dy^2} &= \frac{d^2 H_2(q(y))}{dy^2} = \frac{d^2 H_2(q)}{dq^2} \left(\frac{dq}{dy}\right)^2 + \frac{d H_2(q)}{dq} \frac{d^2q}{dy^2} \\
        &= -\frac{1}{\ln(2)q(1-q)}\left(\frac{y}{2\sqrt{1-y^2}}\right)^2 + \frac{1}{\ln(2)}\ln\left(\frac{1-q}{q}\right) \left(\frac{1}{2(1-y^2)^{3/2}}\right) \\
        &= -\frac{1}{\ln(2)}\frac{1}{1-y^2} + \frac{1}{\ln(2)}\ln\left(\frac{1+\sqrt{1-y^2}}{1-\sqrt{1-y^2}}\right) \left(\frac{1}{2(1-y^2)^{3/2}}\right) \\
        &\geq 0
    \end{align}
    where the last line follows using the inequality $\ln((1+x)/(1-x)) \geq 2x$ when $0 \leq x < 1$. Since the second derivative is non-negative on $y \in [0,1]$ and $\tau^{-1}(1) = 1$, the function satisfies $\tau^{-1}(y) \geq 1 + (y-1)a$, where $a$ is the derivative of the function at $y=1$. In our case, $a = (\ln(2))^{-1}$. This yields Eq.~\eqref{eq:log-Sobolev-first-order}, which was also stated explicitly in Hastings \cite[Lemma 10]{hastings2018shortPath}.
\end{proof}

\begin{proposition}\label{prop:short_path_false_implication}
    If the short-path condition (Cond.~\ref{cond:short_path}) is false for $H_b$, then there exists a probability distribution over assignments to the cost function $p:\{+1,-1\}^n \rightarrow [0,1]$ and a number $\mathcal{U}$ for which
    \begin{align}
        \frac{S}{n} &\geq 1 + \frac{b}{\ln(2)} \mathcal{U} - \frac{1}{n\ln(2)} \\
       \text{and} \qquad   -\frac{1}{bn} \geq \mathcal{U} &\geq -1
    \end{align}
    where
    \begin{align}
        \mathcal{U} &= \sum_z p(z) g_\eta\left(\frac{H(z)}{|E^*|}\right) \qquad \text{is the average energy of } p  \text{ for } g_\eta\left(\frac{H}{|E^*|}\right)\\
        S &= \sum_z -p(z)\log_2 p(z) \qquad \hspace{5.5 pt} \text{is the entropy of } p
    \end{align}
\end{proposition}
\begin{proof}
    Define short-hand $G_\eta = g_\eta(H/|E^*|)$ so that $H_b = -X/n + b \, G_\eta$. If the short-path condition is false, then the ground state $\ket{\bar{\psi}_b}$ of $\bar{\Pi} H_b \bar{\Pi}$ is orthogonal to $\ket{\boldplus}$ and has energy $\bar{E}_b$ satisfying $\bar{E}_b < -1+1/n$. We choose the distribution $p$ to be the one induced by measuring $\ket{\bar{\psi}_b}$ in the computational basis, i.e.~$p(z) = \lvert \braket{z}{\bar{\psi}_b}\rvert^2$. 
    This implies $\mathcal{U} =\bra{\bar{\psi}_b} G_\eta  \ket{\bar{\psi}_b}$ which is bounded below by $-1$. We can then say
    \begin{align}
        \bar{E}_b = \bra{\bar{\psi}_b} \bar{H}_b \ket{ \bar{\psi}_b} = -\bra{\bar{\psi}_b} (X/n)\ket{ \bar{\psi}_b}+ b\bra{\bar{\psi}_b} G_\eta \ket{ \bar{\psi}_b} = -\bra{\bar{\psi}_b} (X/n)\ket{ \bar{\psi}_b}+ b\,\mathcal{U} 
    \end{align}
    Thus, we have
    \begin{equation}
         \bra{\bar{\psi}_b} (X/n)\ket{ \bar{\psi}_b} > 1 - 1/n +b\,\mathcal{U} \,.
    \end{equation}
    Since $\ket{\bar{\psi}_b}$ has zero overlap with $\ket{\boldplus}$, and the second largest eigenvalue of $-X/n$ is $1-2/n$, the left-hand-side is bounded above by $1-2/n$ and we find that $-1 \leq \mathcal{U} \leq -1/(bn)$. Moreover, using Eq.~\eqref{eq:log-Sobolev-first-order} from Prop.~\ref{prop:log-sobolev} to relate $S$ to the left-hand side, we have
    \begin{equation}
        \frac{S}{n} > 1 + \frac{b\,\mathcal{U}}{\ln(2)} + \frac{1}{n\ln(2)}
    \end{equation}
    proving the proposition. 
\end{proof}

\subsubsection{Tools from statistical mechanics}

The main technical problem solved in this section is the following question.
\begin{quote}
    Suppose we have a cost function over $2^n$ assignments and the only information we know about it is that (i) the minimum cost value is $-1$, (ii) the maximum cost value is 0, and (iii) the number of assignments with nonzero cost value is at most $2^{(1-\gamma) n}$, where $0 < \gamma < 1$ is a known number. If $p$ is a probability distribution over assignments with average cost equal to $\mathcal{U}$, with $-1 < \mathcal{U} < 0$, what is an upper bound on the entropy of $p$?
\end{quote}

One difficult aspect of this question is that we know that there are at most $2^{(1-\gamma)n}$ assignments with negative cost value, but we have no additional information about how these costs are distributed on the interval $[-1,0)$, and we want an upper bound on the entropy that will hold regardless of the distribution. To answer this question, we first review some elementary machinery from statistical mechanics culminating in Prop.~\ref{prop:entropy_increases}, which gives a simple condition that allows us to say that one stat mech system has a larger entropy than another. This proposition is then used in the proof of Prop.~\ref{prop:tail_bound_implies_short_path}, where it is shown that the entropy-per-site $S/n$ is upper bounded by $1+\gamma \mathcal{U} + O(1/n)$.  Together with Prop.~\ref{prop:short_path_false_implication}, this implies to that a tail bound implies the short-path condition as long as $b$ is sufficiently small. 

\paragraph{Discrete systems}

Let us recall basic stat mech definitions. A discrete system with $\Omega$ total states is defined by its energy levels $\{E_j\}_{j=0}^{\Omega-1}$ (which we assume are listed in non-decreasing order). A probability distribution over the $\Omega$ states is given by a function $p:\{0,1,2,\ldots,\Omega-1\}\rightarrow [0,1]$ that assigns each state a positive real number, such that $\sum_{j=0}^{\Omega-1} p(j) = 1$. The average energy of the distribution is
\begin{equation}\label{eq:U_statmech_def}
    U = \sum_{j=0}^{\Omega-1} p(j) E_j
\end{equation}
and the entropy is
\begin{equation}\label{eq:S_statmech_def}
    S = -\sum_{j=0}^{\Omega-1} p(j) \log_2(p(j))\,.
\end{equation}
Note that we define the entropy with logarithm base-2 to maintain consistency with the statement of the log-Sobolev inequality in Prop.~\ref{prop:log-sobolev}, rather than following the standard stat mech convention of using natural logarithms. 
Standard results of statistical mechanics dictate that for any fixed choice of $\mathcal{U}$ satisfying $E_0 < \mathcal{U} < E_{\Omega-1}$, there exists a distribution $p$ with average energy $U = \mathcal{U}$, and among all such distributions, the one that maximizes $S$ is uniquely given by $p(j) = e^{-\beta E_j}/Z(\beta)$ for some real, finite value of $\beta$, with $Z(\beta)$ the \textit{partition function}, given by
\begin{equation}\label{eq:Z_statmech_def}
    Z(\beta) = \sum_{j=0}^{\Omega-1} e^{-\beta E_j}\,.
\end{equation}
This is the \textit{Gibbs distribution} or \textit{Boltzmann distribution} and, physically, $\beta$ represents the inverse temperature (which can be negative). For the remainder of this section, we always assume systems are in their Gibbs distribution for a certain inverse temperature $\beta$. We may then view $U$ and $S$ as a function of $\beta$.
Note the following standard identities
\begin{align}
    U(\beta) &= -\frac{1}{Z(\beta)}\frac{dZ(\beta)}{d\beta} \label{eq:U_Zderivative_identity}\\
    S(\beta) &= \ln(2)^{-1}\left[\ln(Z(\beta)) + \beta U(\beta)\right] \label{eq:discrete_entropy}
\end{align}

\paragraph{Cumulative state function}

We can define a cumulative state function $C:(-\infty,\infty)\rightarrow [0,\Omega]$ by
\begin{equation}
    C(E) := |\{j: E_j \leq E\}| = \sum_{j=0}^{\Omega-1} \mathbbm{1}(E_j \leq E)\,,
\end{equation}
where $\mathbbm{1}$ denotes the indicator function. 
Note that $C(E)$ has a discontinuity at each value of $E$ that appears in the list of energies $\{E_j\}_{j=0}^{\Omega-1}$, where it steps up by an integer equal to the multiplicity of $E$ in the sequence. The cumulative state function $C(E)$ contains all of the information about a stat mech system, and we may rewrite the partition function and internal energy directly in terms of $C$, as shown in the following proposition.  

\begin{proposition}\label{prop:ZU_cumulative}
    For a system with $\Omega$ states, we can rewrite $Z(\beta)$ and $U(\beta)$ in terms of the cumulative state functions as follows:
        \begin{align}
        Z(\beta) &= \begin{cases}
            \beta  \int_{-\infty}^{\infty} dE \; C(
        E) e^{-\beta E} & \text{if } \beta > 0 \\
        C(\infty) & \text{if } \beta = 0 \\
        \beta  \int_{-\infty}^{\infty} dE \; \left(C(E)-C(
        \infty)\right) e^{-\beta E} & \text{if } \beta < 0
        \end{cases} \label{eq:Z_cumulative}\\
        U(\beta) &=
        \begin{cases}
            Z(\beta)^{-1}\int_{-\infty}^{\infty} dE \; C(E) (\beta  E -1)e^{-\beta E} & \text{if } \beta > 0 \\
            C(\infty)^{-1}\left(-\int_{-\infty}^0 dE \;  C(E) + \int_{0}^\infty dE \; \left(C(\infty)-C(E)\right) \right) & \text{if } \beta = 0 \\
            Z(\beta)^{-1}\int_{-\infty}^{\infty} dE \; \left(C(E)-C(
        \infty)\right)(\beta E -1)e^{-\beta  E} & \text{if } \beta < 0 \label{eq:U_cumulative}\\
        \end{cases}
        \end{align}
\end{proposition}
\begin{proof}
When $\beta=0$, we see from the definition of $Z$ in Eq.~\eqref{eq:Z_statmech_def} that $Z(0) = \Omega = C(\infty)$. Meanwhile, from the definition of $U$ in Eq.~\eqref{eq:U_statmech_def},  
\begin{align}
U(0) &= \frac{1}{Z(0)}\sum_{j=0}^{\Omega-1}E_j= \frac{1}{C(\infty)}\sum_{j=0}^{\Omega-1}\left(-\int_{\min(E_j,0)}^{0} dE +\int_{0}^{\max(0,E_j)} dE\right) \\
&= \frac{1}{C(\infty)}\sum_{j=0}^{\Omega-1}\left(-\int_{-\infty}^{0} \mathbbm{1}(E_j \leq E)dE +\int_{0}^{\infty} (1-\mathbbm{1}(E_j \leq E)) dE\right)   \\
&= \frac{1}{C(\infty)}\left(-\int_{-\infty}^{0} \sum_{j=0}^{\Omega-1}\mathbbm{1}(E_j \leq E)dE +\int_{0}^{\infty} \sum_{j=0}^{\Omega-1}(1-\mathbbm{1}(E_j \leq E)) dE\right) \\
&= \frac{1}{C(\infty)}\left(-\int_{-\infty}^0 dE \;  C(E) + \int_{0}^\infty dE \; \left(C(\infty)-C(E)\right) \right)
\end{align}
When $\beta > 0$, we can rewrite
\begin{align}
    Z(\beta) = \sum_{j=0}^{\infty}e^{-\beta E_j} &= \beta \sum_{j=0}^{\Omega-1}\int_{E_j}^{\infty}dE  e^{-\beta E} = \beta \sum_{j=0}^{\Omega-1}\int_{-\infty}^{\infty}dE  \mathbbm{1}(E_j \leq E)e^{-\beta E} \\
    &=\beta \int_{-\infty}^{\infty}dE  \sum_{j=0}^{\Omega-1}\mathbbm{1}(E_j \leq E)e^{-\beta E} = \beta \int_{-\infty}^{\infty}dE C(E)e^{-\beta E} \,.
\end{align}
and
\begin{align}
    U(\beta) &= \sum_{j=0}^{\infty}E_je^{-\beta E_j} =  \sum_{j=0}^{\Omega-1}\int_{E_j}^{\infty}dE  (\beta E_j-1)e^{-\beta E} = \sum_{j=0}^{\Omega-1}\int_{-\infty}^{\infty}dE  \mathbbm{1}(E_j \leq E) (\beta E_j-1)e^{-\beta E} \\
    &=\int_{-\infty}^{\infty}dE  \sum_{j=0}^{\Omega-1}\mathbbm{1}(E_j \leq E) (\beta E_j-1)e^{-\beta E} = \int_{-\infty}^{\infty}dE C(E)(\beta E_j-1)e^{-\beta E} \,.
\end{align}
The formulas for $\beta < 0$ can be derived directly from the formulas for $\beta > 0$. Define a new stat mech system, denoted with a prime on each symbol, by negating all of the energies, i.e.~$E'_j = -E_j$. We have $C'(E) = C(\infty) - C(E)$. Moreover, the primed system at inverse temperature $\beta'$ is equivalent to the original system at inverse temperature $\beta$ when $\beta = -\beta'$. Thus, negative $\beta$ corresponds to positive $\beta'$ and the same formulas apply with the substitution $C(E)\rightarrow C'(E) = C(\infty) - C(E)$. 
\end{proof}

So far, we have restricted to $C(E)$ that come from a system with a discrete set of energy levels $\{E_j\}$. We now generalize slightly and consider a wider class of $C(E)$ which need not always correspond to a system with discrete energy levels. 

\begin{definition}[finite-stepping]\label{def:finite-stepping}
We say that a cumulative state function $C(E)$ has the \textit{finite-stepping} property if it (i) is monotonically non-decreasing, (ii) has a finite number of discontinuities, and (iii) takes a constant value in between consecutive discontinuities.
\end{definition}

Thus, when $C(E)$ comes from a discrete system of $\Omega$ energy levels, it will have the finite-stepping property and additionally each step size will be an integer; however, other functions for which the step size is non-integer also have the finite-stepping property.

The equations for $Z(\beta)$ and $U(\beta)$ in Prop.~\ref{prop:ZU_cumulative} are well-defined for any function $C(E)$ with the finite-stepping property. We will also extend the definition of $S(\beta)$ to all functions $C(E)$ with the finite-stepping property by continuing to impose the identity in Eq.~\eqref{eq:discrete_entropy}.

\paragraph{Comparing the entropy of different systems}

We show a condition under which one can assert that one stat mech system has larger entropy than another for the same average energy. 

\begin{proposition}\label{prop:entropy_increases}
Consider two statistical mechanical systems defined by their cumulative state functions $C_1(E)$ and $C_2(E)$, and assume that both have the finite-stepping property. Let $U_1(\beta)$ and $U_2(\beta)$ denote their average energy functions. Given a fixed choice of $\mathcal{U} \in \mathbb{R}$, let $\beta_1$ and $\beta_2$ be the solutions to $U_1(\beta_1) = \mathcal{U}$ and $U_2(\beta_2)= \mathcal{U}$, which we assume exist.  If for every real value of $E$, it holds that $C_1(E) \geq C_2(E)$ and $C_1(\infty) - C_1(E) \geq C_2(\infty) - C_2(E)$, then
\begin{equation}
    S_1(\beta_1) \geq S_2(\beta_2)\,,
\end{equation}
where $S_1(\beta_1)$ and $S_2(\beta_2)$ denote the entropy of each system at fixed internal energy $\mathcal{U}$. 
\end{proposition}
\begin{proof}
    For $j \in \{1,2\}$, let $A_j$ ($B_j$) be the minimum (maximum) argument of $C_j$ where it has a discontinuity. Note that $U_j(\beta)$ is a monotonically decreasing and continuous function of $\beta$ and that 
    \begin{align}
    \lim_{\beta \rightarrow \infty}U_j(\beta) &= A_j \\ 
    \lim_{\beta \rightarrow -\infty}U_j(\beta) &= B_j \,.
    \end{align}
    Thus, since we have assumed $U_j(\beta_j) = \mathcal{U}$ has a solution, it must be the case that $A_j \leq \mathcal{U} \leq B_j$ for $j \in \{1,2\}$. The limits above correspond to where all the probability mass is concentrated on the lowest and highest energy levels, respectively, and the statement can be proved by evaluation of Eqs.~\eqref{eq:Z_cumulative} and \eqref{eq:U_cumulative} in the $\beta \rightarrow \pm \infty$ limits. 
    
    We now consider a one-parameter family of systems that interpolates between system 1 and system 2, with cumulative states function
    \begin{equation}
        C(E,t) = (1-t) C_1(E) + t C_2(E)
    \end{equation}
    Let $U(\beta,t)$ denote its average energy function and let $S(\beta,t)$ denote its entropy function.  Since $C_1(E)$ and $C_2(E)$ have the finite-stepping property, it is easy to verify that $C(E,t)$ also has the finite-stepping property for any $t \in [0,1]$. Moreover, if $A(t)$ ($B(t)$) is the minimum (maximum) argument of $C(\cdot,t)$ for which it has a discontinuity, we can say that $A(t) \leq \min_j A_j$ and $B(t) \geq \max_j B_j$. Thus $A(t) \leq \mathcal{U} \leq B(t)$ and for every $t$, there will exist a value $\beta^*(t)$ for which $U(\beta^*(t),t) = \mathcal{U}$. That is, as the cumulative states function changes with $t$, $\beta^*$ will also change with $t$ such that the average energy remains fixed at $\mathcal{U}$. 
    
    Starting from Eq.~\eqref{eq:discrete_entropy}, we now compute the total derivative of $S$ with respect to $t$, for fixed average energy:
    \begin{align}
        \ln(2)\frac{dS(\beta^*(t),t)}{dt} &= \frac{1}{
        Z(\beta^*(t),t)}\frac{dZ(\beta^*(t),t)}{dt} + \mathcal{U} \frac{d\beta^*(t)}{dt} \\
        &= \frac{1}{
        Z(\beta^*(t),t)}\frac{\partial Z(\beta^*(t),t)}{\partial t} +\left( \frac{1}{
        Z(\beta^*(t),t)}\frac{\partial Z(\beta^*(t),t)}{\partial \beta} + \mathcal{U} \right)\frac{d\beta^*(t)}{dt}\,.
    \end{align}
    From Eq.~\eqref{eq:U_Zderivative_identity}, we see that the quantity in parentheses vanishes, and we have
    \begin{align}
        \frac{dS(\beta^*(t),t)}{dt} &= \frac{1}{
        \ln(2)Z(\beta^*(t),t)}\frac{\partial Z(\beta^*(t),t)}{\partial t} 
    \end{align}
    Using Eq.~\eqref{eq:Z_cumulative} and assuming $\beta>0$, this yields
    \begin{align}
        \frac{dS(\beta^*(t),t)}{dt}&= \frac{\beta}{
        \ln(2)Z(\beta^*(t),t)}  \int_{-\infty}^{\infty} dE \; \frac{\partial C(
        E,t)}{\partial t} e^{-\beta^*(t) E} \\
        &= \frac{\beta}{
        \ln(2)Z(\beta^*(t),t)} \int_{-\infty}^{\infty} dE \; (C_2(E)-C_1(E)) e^{-\beta E}  \\
        &\leq 0
    \end{align}
    since, by assumption, $C_1(E) \geq C_2(E)$ for all $E$. Assuming $\beta < 0$, it yields
    \begin{align}
        \frac{dS(\beta^*(t),t)}{dt}&= \frac{\beta}{
        \ln(2)Z(\beta^*(t),t)} \int_{-\infty}^{\infty} dE \; \left(\frac{\partial C(
        E,t)}{\partial t}-\frac{\partial C(
        \infty,t)}{\partial t}\right) e^{-\beta^*(t) E} \\
        &= \frac{\beta}{
        \ln(2)Z(\beta^*(t),t)}  \int_{-\infty}^{\infty} dE \; \left[(C_1(\infty)-C_1(E))-(C_2(\infty)-C_2(E))\right] e^{-\beta E}  \\
        &\leq 0
    \end{align}
    since, by assumption, $C_1(\infty)-C_1(E) \geq C_2(\infty)-C_2(E)$ (recall that $\beta < 0$).
    
    Since the derivative of $S$ with $t$ is always non-positive, we have $S(\beta^*(1),1) \leq S(\beta^*(0),0) $, which is equivalent to $S_2(\beta_2) \leq S_1(\beta_1)$. 
\end{proof}

\subsubsection{Tail bound implies short-path condition}

\begin{proposition}\label{prop:tail_bound_implies_short_path}
    Let $H$ be a cost function on $n$ bits with optimal value $E^*$ and let $C(E)$ be its cumulative state function
    \begin{equation}
        C(E) = |\{z: H(z) \leq E\}|\,.
    \end{equation}
    Fix a value of $\eta$ and $b$, which defines the Hamiltonian $H_b$ in Eq.~\eqref{eq:Hb}. If $C(E)$ obeys a tail bound $C(E^*(1-\eta)) \leq 2^{(1-\gamma)n}$, and
    \begin{equation}
        b \leq \frac{\ln(2)\gamma}{2+\ln(2)}
    \end{equation}
    then $H_b$ satisfies the short-path condition (Cond.~\ref{cond:short_path}). 
\end{proposition}
\begin{proof}
The function $C(E)$ is the cumulative state function for the cost function $H$. Define $C_\eta(v)$ to be the cumulative state function for the cost function $g_\eta(H/|E^*|)$. Recalling that $g_\eta(x) = \min(0,(x+1-\eta)/\eta)$, we can say that
\begin{equation}
    C_\eta(v) = 
    \begin{cases}
        0   & \text{if } v < -1   \\
        C(E^*(1-\eta v - \eta)) \leq 2^{(1-\gamma)n} & \text{if } -1 \leq v < 0   \\
        2^n & \text{if } v \geq 0\,,
    \end{cases} 
\end{equation}
and note that, because it is associated with a discrete system, $C_\eta(v)$ has the finite-stepping property defined in Def.~\ref{def:finite-stepping}. 
Now consider an alternate system which has $2^{(1-\gamma)n}$ states at energy $-1$, and $2^n$ states at energy $0$. Denote quantities related to this system with an overline, and define its cumulative states function by 
\begin{equation}
    \overline{C}(v) = 
    \begin{cases}
        0 & \text{if }  v < -1\\
        2^{(1-\gamma)n} & \text{if }  -1 \leq v < 0\\
        2^{(1-\gamma)n} + 2^n &\text{if } v \geq 0\,.
    \end{cases}
\end{equation}
The function $\overline{C}$ is constant except at two discontinuities where it steps up by a finite amount; thus it also has the finite-stepping property. Inspecting the above two expressions, we see that, due to the assumption $C(E^*(1-\eta))\leq 2^{(1-\gamma)n}$, we have $C_\eta(v) \leq \overline{C}(v)$. We also see that, by construction, $\overline{C}(\infty)-\overline{C}(v) \geq 2^n$ whenever $v < 0$, and thus $C_\eta(\infty) - C_\eta(v) \leq \overline{C}(\infty) - \overline{C}(v)$ for all $v$. Hence, the requirements of Prop.~\ref{prop:entropy_increases} are satisfied. 

Let $S(\mathcal{U})$ denote the maximum entropy of the system described by cumulative state function $C_\eta(\cdot)$ for fixed expected cost value (average energy) of $g_\eta(H/|E^*|)$ equal to $\mathcal{U}$. Let $\overline{S}(\mathcal{U})$ denote the same for the cumulative state function $\overline{C}(\cdot)$. Prop.~\ref{prop:entropy_increases} implies that
\begin{equation}\label{eq:SvsSoverline}
    S(\mathcal{U}) \leq \overline{S}(\mathcal{U})
\end{equation}
for all $\mathcal{U}$ in which there exists a distribution for both systems with average energy $\mathcal{U}$, i.e.~when $-1 \leq \mathcal{U} \leq 0$.

The entropy of the $\overline{C}$ system can be analyzed exactly. We compute the partition function
\begin{equation}\label{eq:Z_overline}
    \overline{Z}(\beta) = \beta \int_{-1}^\infty dv\; \overline{C}(v) e^{-\beta v} = 2^{(1-\gamma)n}e^\beta + 2^n
\end{equation}
which, via Eq.~\eqref{eq:U_Zderivative_identity}, determines the average energy 
\begin{equation}\label{eq:U_overline}
    \overline{U}(\beta) = -\frac{1}{\overline{Z}(\beta)}\frac{d\overline{Z}(\beta)}{d\beta} = -\frac{2^{(1-\gamma)n}e^{\beta}}{2^{(1-\gamma)n}e^\beta + 2^n}\,.
\end{equation}
For a given $\mathcal{U}$ satisfying $-1 \leq \mathcal{U} \leq 0$, we choose $\overline{\beta}$ such that $\overline{U}(\overline{\beta}) = \mathcal{U}$, giving the relation
\begin{align}
    e^{\overline{\beta}} &= -\frac{2^{\gamma n}\mathcal{U}}{1+\mathcal{U}} \\
    \implies \overline{\beta} &= \gamma n \ln(2) + \ln\left(\frac{-\mathcal{U}}{1+\mathcal{U}}\right)\,.
\end{align}
Plugging this value of $\beta=\overline{\beta}$ in to Eq.~\eqref{eq:Z_overline}, as well as the definition of entropy in Eq.~\eqref{eq:discrete_entropy}, we find:
\begin{align}
    \overline{Z} &= \frac{2^n}{1+\mathcal{U}} \\
    \overline{S}(\mathcal{U}) &= \ln(2)^{-1}\left[\ln(\overline{Z}) + \overline{\beta} \mathcal{U}\right] \\
    &= n - \log_2(1+\mathcal{U}) + \gamma n \,\mathcal{U}  - (-\mathcal{U})\log_2(-\mathcal{U}) + (-\mathcal{U})\log_2(1+\mathcal{U}) \\
    &= n(1 + \gamma \, \mathcal{U}) + H_2(-\mathcal{U})
\end{align}
where $H_2(q) = -q\log_2(q) - (1-q)\log_2(1-q)$ is the binary entropy. Note that $H_2(q) \leq 1$ for all $q$. 

By assumption, we have $b$ satisfies $b \leq \frac{\ln(2)\gamma}{2+\ln(2)}$. Thus, we can say
\begin{align}
    \frac{S(\mathcal{U})}{n} &\leq \frac{\overline{S}(\mathcal{U})}{n} = 1 + \gamma \, \mathcal{U} + \frac{H_2(-\mathcal{U})}{n} \leq 1 + \gamma \,\mathcal{U} + \frac{1}{n} \leq 1+\frac{(2+\ln(2))b}{\ln(2)}\mathcal{U} + \frac{1}{n} \\
    &= \left[1+\frac{b}{\ln(2)}\mathcal{U} - \frac{1}{n\ln(2)}\right] + \left[(b\,\mathcal{U} + 1/n)(1+\ln(2)^{-1})\right]
\end{align}
Thus, whenever $-1 \leq \mathcal{U} \leq -\frac{1}{bn}$, the final term is negative and $S(\mathcal{U})/n \leq 1+\ln(2)^{-1}b\,\mathcal{U} - (n\ln(2))^{-1}$. By Prop.~\ref{prop:short_path_false_implication}, this implies that the short-path condition must hold. 
\end{proof}

\subsubsection{Proof of Lemma \ref{lem:tail_bound_implies_conditions}}\label{sec:proof_of_lem_tail_bound}

\begin{proof}[Proof of Lemma \ref{lem:tail_bound_implies_conditions}]
We assume that a tail bound of the form $C(E^*(1-\eta)) \leq 2^{(1-\gamma)n}$ holds and that $b \leq \ln(2)\gamma/(2+\ln(2))$. By Prop.~\ref{prop:tail_bound_implies_short_path}, this means that $H_b$ obeys the short-path condition.  By Prop.~\ref{prop:short_path->(i)}, the short-path condition implies the large-excited-energy condition (Cond.~\ref{cond:large-excited-energy}). By Prop.~\ref{prop:short_path->(ii)}, the short-path condition (Cond.~\ref{cond:short_path}) combined with the tail bound implies that $E_b \geq -1-2^{-\gamma n}(b+b^2 n)$. Noting $b < 1$ and the assumption that $2^{\gamma n} \geq 2n^4 \geq n^3(b+b^2 n)$, we have that $E_b \geq -1-1/n^3$, giving the small-ground-energy-shift condition (Cond.~\ref{cond:small-ground-energy-shift}) and proving the lemma.
\end{proof}


\subsection{$\alpha$-(sub)depolarizing}\label{app:alpha_(sub)depolarizing}

\begin{proof}[Proof of Proposition \ref{prop:CSP_alpha-sub-depolarizing}]
    We wish to show $(H,g_\eta)$ is subdepolarizing. Let $f(x) = -g_\eta(-x)$, and note that $f$ is monotonically non-decreasing and convex. Note also that it is twice differentiable at every place where it is nonzero. The cost function $H$ can be written as $H = \sum_{j=1}^m \mathcal{C}_j$ for $m$ different constraints $\mathcal{C}_j$ each depending on at most $k$ variables. Recall that each constraint takes value $-1$ for some integral number $s_j$ of the $2^k$ settings of these $k$ variables, and it takes value $s_j/(2^k-s_j)$ on the other $2^k-s_j$ settings, with $0 \leq s_j < 2^k$. Given an assignment $x$, let $\mathcal{S} \subset [M]$ be the subset of constraints that $x$ satisfies (i.e.,~$j \in \mathcal{S}$ if $\mathcal{C}_j(x) = -1$).
    Suppose that $y$ is generated by flipping one bit of $x$ at random. For each $j\in \mathcal{S}$, there is a $(n-k)/n$ chance that the flipped bit is not involved in constraint $j$ and $\mathcal{C}_j(y) = \mathcal{C}_j(x) = -1$. Meanwhile, there is a $k/n$ chance that the flipped bit is involved, and the clause could become unsatisfied, but even in this case, we can at least say that $\mathcal{C}_j(y) \leq 2^k-1$, as this is an upper bound for the maximum value the constraint can take on any input (achieved when $s_j = 2^k-1$). Thus, for any $j \in \mathcal{S}$,
    \begin{equation}
        \EV_{y \sim x}\mathcal{C}_j(y) \leq \frac{n-k}{n}(-1) + \frac{k}{n}(2^k-1) = -1 + \frac{k2^k}{n} = \mathcal{C}_j(x) +\frac{k2^k}{n}
    \end{equation}
    For each $j \not\in \mathcal{S}$, the value of $\mathcal{C}_j(x)$ is already the maximum value of $\mathcal{C}_j$, so in this case, we can assert that $\EV_{y \sim x}\mathcal{C}_j(y) \leq \mathcal{C}_j(x)$. Overall, this gives
    \begin{equation}\label{eq:energy_increase_bound}
        \EV_{y \sim x}H(y) \leq H(x) + \lvert \mathcal{S} \rvert \frac{k2^k}{n} \leq H(x) + m\frac{k2^k}{n} \,.
    \end{equation}
    Let $c_1,\ldots, c_T$ be non-negative constants less than 1. We can then say the following, where the first step uses Jensen's inequality and the fact that $\prod_{t=1}^Tf(c_tx)$ is a convex function from Proposition \ref{prop:product_is_convex}, while the second step uses the monotonicity of $f$ along with Eq.~\eqref{eq:energy_increase_bound} (recalling that $E^* < 0$).
    \begin{align}
    \EV_{y\sim x} \prod_{t=1}^T f\left(\frac{c_t H(y)}{E^*}\right) &\geq \prod_{t=1}^T f\left(c_t\EV_{y \sim x} \frac{H(y)}{E^*}\right) \geq \prod_{t=1}^T f\left(\frac{c_tH(x)}{E^*} + \frac{c_t m}{E^*}\frac{k2^k}{n}\right) \\
    &=  \prod_{t=1}^T f\left(\frac{c_tH(x)}{E^*}\left(1 + \frac{m}{H(x)}\frac{k2^k}{n}\right) \right) \label{eq:boundfinal}
\end{align}
 If $H(x) \geq E^*(1-\eta)/c_t$ for at least one value of $t$, then $0 = \prod_{t=1}^T f(c_tH(x)/E^*)=\prod_{t=1}^T f(c_t(1-\alpha)H(x)/E^*)$, and $\EV_{y\sim x} \prod_{t=1}^T f(c_t H(y)/E^*) \geq \prod_{t=1}^T f(c_t(1-\alpha)H(x)/E^*))$ is true. On the other hand, if $H(x) \leq E^*(1-\eta)/c_t$ for all $t$, then we can certainly say that $H(x) \leq E^*(1-\eta)$ and hence $1/H(x) \geq -1/|E^*|(1-\eta)$ and by substitution in Eq.~\eqref{eq:boundfinal} and the fact that $f$ is monotonic, we have
\begin{align}
    \EV_{y\sim x} \prod_{t=1}^T f\left(\frac{c_tH(y)}{E^*}\right)
    &\geq  \prod_{t=1}^T f\left(\frac{c_tH(x)}{E^*}\left(1 - \frac{m}{|E^*|}\frac{1}{(1-\eta)}\frac{k2^k}{n}\right)\right) 
\end{align}
which proves the lemma. 
\end{proof}
 \begin{proposition}\label{prop:product_is_convex}
    Suppose $f: (-\infty,1] \rightarrow [0,1]$ is a monotonically non-decreasing, convex function. Suppose that $f$ is twice differentiable for every $x$ in which $f(x) \neq 0$. Then for any non-negative constants $c_1,\ldots, c_m$, the function $h(x) = \prod_{t=1}^m f(c_t x)$ is also a convex function.
 \end{proposition}
 \begin{proof} 
 Since $f$ is non-negative and monotonically non-decreasing, if $h(y) = 0$, then $h(x)=0$ for all $x < y$ as well. If $h(x) \neq 0$, then $h$ is twice-differentiable at $x$, since it is the product of twice-differentiable functions. To show that $h$ is convex, it suffices to show that the second derivative of $h$ is non-negative for all points $x$ in which $h(x)$ is non-zero. 
\begin{align}
    h''(x) = \sum_{t} c_t^2f''(c_{t} x) \frac{h(x)}{f(c_t x)} + \sum_{t_1} \sum_{t_2 \neq t_1} c_{t_1} c_{t_2} f'(c_{t_1}x)f'(c_{t_2}x) \frac{h(x)}{f(c_{t_1}x)f(c_{t_2}x)}
\end{align}
We observe that $h''(x)$ is always non-negative, since $c_t, h(x), f(c_t x), f''(c_t x), f'(c_t x) \geq 0$ for all $t$ and all $x$.
\end{proof}
\begin{proposition}\label{prop:depolarizing_implies_subdepolarizing}
    If $H$ has the $\alpha$-depolarizing property (Definition \ref{def:alpha-depolarizing}) and $g:[-1,\infty) \rightarrow [-1,0]$ is a monotonic non-decreasing, concave function that is twice-differentiable at every point where it is nonzero, then $(H,g)$ has the $\alpha$-subdepolarizing property.
\end{proposition}
\begin{proof}
    This follows from application of Jensen's inequality and Prop.~\ref{prop:product_is_convex}. Let $f(x) = -g(-x)$ and note that $f$ satisfies the constraints of Prop.~\ref{prop:product_is_convex}.  The proposition implies that for any non-negative constants $c_1,\ldots,c_T$
\begin{align}
    \EV_{y\sim x} \prod_{t=1}^{T}f(c_tH(y)/E^*) \geq \prod_{t=1}^T f\left(c_t\EV_{y \sim x} H(y)/E^*\right) = \prod_{t=1}^T f\left(c_t(1-\alpha)H(x)/E^*\right)
\end{align}
where the second equality results from invoking the $\alpha$-depolarizing property. This proves the proposition. 
\end{proof}

\subsection{Overlap lower bound with approximate ground state projector}\label{app:overlap_bound}

\begin{proof}[Proof of Lemma \ref{lem:<+|Pl|0>}]
First, let $\mathcal{E} = H(z)/|E^*|$ and recall the assumption $H(z) \leq (1-\eta)E^*$. Thus $\mathcal{E} < 0$. Next, define $A$, $B$, and $f$ by the following equations. 
\begin{align}
A &= X/n \\
B &= -b\, g_\eta(H/|E^*|) = b\,f(H/E^*)
\end{align}
such that, referring to the definition of $H_b$ in Eq.~\eqref{eq:Hb}, $H_b = -A-B$ and thus
\begin{equation}\label{eq:P_ell_expanded}
    P_\ell = \frac{(A+B)^\ell}{|E_b|^\ell}.
\end{equation}
When the small-ground-energy-shift condition (Cond.~\ref{cond:small-ground-energy-shift}) holds, the denominator of the above expression is no larger than $(1+1/n^3)^{\ell} \leq e$, where the inequality holds since $\ell < n^3$ by assumption. If we expand the numerator, we get a sum over $2^\ell$ strings, where each string is a length-$\ell$ sequence of $A$ and $B$ operators (note that $A$ and $B$ do not commute). We will evaluate this sum by computing  $\bra{\boldplus}\sigma \ket{z}$ where $\sigma$ is a string of $A$ and $B$ operators.  To get a feel for this, we start simple and compute
\begin{align}
    \bra{\boldplus}A\ket{z} &= 2^{-n/2} \\
    \bra{\boldplus}B\ket{z} &= b 2^{-n/2}f(\mcEbar)\,.
\end{align}
Here we have utilized the eigenvalue equations $\bra{\boldplus}A = \bra{\boldplus}$ and $B \ket{z} = bf(H(z)/E^*)\ket{z}$, and the inner product $\braket{\boldplus}{z} = 2^{-n/2}$. 

The first non-trivial computation is $\bra{\boldplus}BA\ket{z}$. We evaluate this by noting that applying $A$ to $\ket{z}$ yields a uniform superposition over the $n$ computational basis states that differ from $z$ by one bit flip. As in the definition of $\alpha$-subdepolarizing (Definition \ref{def:alpha-subdepolarizing}), we use the notation $y \sim z$ to denote that $y$ is generated from $z$ by choosing one of these $n$ neighboring assignments at random. We have
\begin{align}
    \bra{\boldplus}BA\ket{z} &= b\EV_{y\sim z} \bra{\boldplus}f\left(\frac{H}{E^*}\right)\ket{y} = b\, 2^{-n/2}\EV_{y\sim z}f\left(\frac{H(y)}{E^*}\right) \\
    &\geq b \, 2^{-n/2}f\left(\frac{H(z)}{E^*}(1-\alpha)\right) = b\, f\left(\mcEbar(1-\alpha)\right)2^{-n/2}
\end{align}
where the inequality uses the $\alpha$-subdepolarizing property. Similar logic yields the relation
\begin{align}
    \bra{\boldplus}BA^k\ket{z} \geq b f(\mcEbar(1-\alpha)^k) 2^{-n/2}
\end{align}
for any $k$, by applying the $\alpha$-subdepolarizing property $k$ times.

We can also consider
\begin{align}
    \bra{\boldplus}BABA\ket{z} &= b^22^{-n/2}\EV_{y \sim z}\left[f\left(\frac{H(y)}{E^*}\right) \EV_{w\sim y}f\left(
    \frac{H(w)}{E^*}\right)\right]\\
    &\geq b^2 2^{-n/2} \EV_{y \sim z}\left[f\left(\frac{H(y)}{E^*}\right) f\left(\frac{H(y)(1-\alpha)}{E^*}\right)\right] \\
    &\geq b^2 2^{-n/2} f\left(\frac{H(z)(1-\alpha)}{E^*}\right)f\left(\frac{H(z)(1-\alpha)^2}{E^*}\right) \\
    &=b^2 2^{-n/2} f\left(\mcEbar(1-\alpha)\right)f\left(\mcEbar(1-\alpha)^2\right)
\end{align}
where the second-to-last line again uses the $\alpha$-subdepolarizing property. 

In general we may write any string of $A$s and $B$s as
\begin{equation}
     \ldots AB^{x_3}A B^{x_2} A B^{x_1} AB^{x_0}
\end{equation}
and thus associate every string $\sigma$ with a sequence of non-negative integers $\sigma=(x_0,x_1,\ldots )$. Let this sequence be infinitely long by padding it with an infinite number of zeros (note that no two distinct strings $\sigma$ and $\sigma'$ both of length $\ell$ can map to the same sequence of integers). By generalizing the calculations above, we can show the following proposition. 

\begin{proposition}\label{prop:w(sigma)}
If $(H,g_\eta)$ is $\alpha$-subdepolarizing and $B = -b\, g_\eta(H/|E^*|) =: b\,f(H/E^*)$, then
    \begin{equation}
    2^{n/2}\bra{\boldplus} \ldots AB^{x_3}A B^{x_2} A B^{x_1} AB^{x_0}\ket{z} \geq  b^{\sum_{j=0}^\infty x_j} \prod_{j=0}^\infty f(\mcEbar(1-\alpha)^j)^{x_j} =: w(\sigma)
\end{equation}
where the right-hand side of the above equation defines $w(\sigma)$, with $\sigma = (x_0,x_1,\ldots)$. 
\end{proposition}
The proof of Prop.~\ref{prop:w(sigma)} appears after the conclusion of this proof.
Note that if $\sigma = \sigma_1 + \sigma_2$ (with entry-wise addition at each location in the sequence), then $w(\sigma) = w(\sigma_1)w(\sigma_2)$. We can now compute a sum over $w(\sigma)$ for all $\sigma$. 

\begin{proposition}\label{prop:sum_sigma_in_Gamma}
Let $\Gamma$ be the set of all sequences $\sigma$ that have non-negative entries.
If $\mcEbar \geq 1-\eta$, then
\begin{equation}
    \sum_{\sigma \in \Gamma} w(\sigma) \geq \exp\left(\frac{b\mcEbar}{\eta \alpha}F\left(\frac{1-\eta}{\mcEbar}\right)\right)\,
\end{equation}
where the function $F(x) := 1-x + x \ln(x)$, as defined previously in Eq.~\eqref{eq:F_def}.
\end{proposition}
The proof of Prop.~\ref{prop:sum_sigma_in_Gamma} appears after the conclusion of this proof. The sum computed in Prop.~\ref{prop:sum_sigma_in_Gamma} includes all strings, but we only want to sum over strings that appear in the expansion of $(A+B)^\ell$ for a particular choice of $\ell$. Let the set of sequences that appear in this expansion be denoted $\Lambda_\ell$. We have 
\begin{equation}\label{eq:sequences_difference}
    \sum_{\sigma \in \Lambda_\ell} w(\sigma) = \sum_{\sigma \in \Gamma} w(\sigma)-\sum_{\sigma \not\in  \Lambda_\ell} w(\sigma)\,.
\end{equation}
The following proposition upper bounds the final term. 

\begin{proposition}\label{prop:not_in_expansion}
\begin{equation}
    \sum_{\sigma \not\in \Lambda_\ell} w(\sigma) \leq b\mcEbar(1-\alpha)^{\ell+1}\left(\frac{e^{2/\alpha}}{1-\alpha-b}+  \alpha^{-1}\sum_{\sigma \in \Gamma}w(\sigma) \right)\,.
\end{equation}    
\end{proposition}
The proof of Prop.~\ref{prop:not_in_expansion} appears after the conclusion of this proof.
Using Prop.~\ref{prop:not_in_expansion}, Prop.~\ref{prop:sum_sigma_in_Gamma} and Eq.~\eqref{eq:sequences_difference} we find that
\begin{align}
    \sum_{\sigma \in \Lambda_l} w(\sigma) &\geq \sum_{\sigma \in \Gamma} w(\sigma) (1-b\mcEbar\alpha^{-1}(1-\alpha)^{\ell+1}) - b\mcEbar(1-\alpha)^{\ell+1}(1-\alpha-b)^{-1}e^{2/\alpha} \\
    &\geq \exp\left(\frac{b\mcEbar}{\alpha \eta}F\left(\frac{1-\eta}{\mcEbar}\right)\right)\left(1-b\mcEbar\frac{(1-\alpha)^{\ell+1}}{\alpha}\right) - \exp\left(\frac{2}{\alpha}\right)\frac{b\mcEbar(1-\alpha)^{\ell+1}}{1-\alpha-b}  \label{eq:first_bound_Lambda_l}
\end{align}
Now, using the assumption $\ell > 3/\alpha^2$, we note that $(1-\alpha)^\ell \leq e^{-3/\alpha}$, and using the assumption $\alpha < (1-b)/2$, we have $1-\alpha-b > \alpha$. Further noting that 
\begin{equation}
    1 \leq \exp\left(\frac{b\mcEbar}{\alpha \eta}F\left(\frac{1-\eta}{\mcEbar}\right) \right) 
\end{equation}
allows us to rewrite Eq.~\eqref{eq:first_bound_Lambda_l} as
\begin{equation}
     \sum_{\sigma \in \Lambda_l} w(\sigma) \geq \exp\left(\frac{b\mcEbar}{\alpha \eta}F\left(\frac{1-\eta}{\mcEbar}\right)\right)\left(1-b\mcEbar(1-\alpha)e^{-3/\alpha}/\alpha - b\mcEbar(1-\alpha)e^{-1/\alpha}/\alpha\right)\,. 
\end{equation}
Finally noting that $e^{-1/\alpha}/\alpha \leq 1/e$ while $b, \mcEbar,(1-\alpha) < 1$, the final term in parentheses is greater than $1-2e$. Recalling the observation that the denominator in Eq.~\eqref{eq:P_ell_expanded} is bounded above by $e$, the lemma statement follows. 
\end{proof}

\begin{proof}[Proof of Prop.~\ref{prop:w(sigma)}]
    Let $t$ be the largest index for which $x_t \neq 0$. Following the examples described above the proof of Lemma \ref{lem:<+|Pl|0>}, we can say that
    \begin{align}
        2^{n/2}\bra{\boldplus} \sigma \ket{z} = b^{\sum_{j=0}^{t}x_j}\EV_{z_1\sim z}\EV_{z_2\sim z_1}\cdots\EV_{z_t\sim z_{t-1}}\left[\prod_{j=0}^t f\left(\frac{H(z_j)}{E^*}\right)^{x_j}\right]\,.
    \end{align}
    Invoking the $\alpha$-subdepolarizing property on the $\EV_{z_t\sim z_{t-1}}$ expectation allows the simplification
        \begin{align}
        2^{n/2}\bra{\boldplus} \sigma \ket{z} \geq b^{\sum_{j=0}^{t}x_j}\EV_{z_1\sim z}\cdots\EV_{z_{t-1}\sim z_{t-2}}\left[f\left(\frac{H(z_{t-1})}{E^*}(1-\alpha)\right)^{x_t}\prod_{j=0}^{t-1} f\left(\frac{H(z_j)}{E^*}\right)^{x_j}\right]\,.
    \end{align}
    Invoking it again on the $\EV_{z_{t-1}\sim z_{t-2}}$ allows the expectation values on the right-hand-side to be simplifed to
    \begin{align}
    \EV_{z_1\sim z}\cdots\EV_{z_{t-2}\sim z_{t-3}}\left[f\left(\frac{H(z_{t-2})}{E^*}(1-\alpha)^2\right)^{x_t}f\left(\frac{H(z_{t-2})}{E^*}(1-\alpha)\right)^{x_{t-1}}\prod_{j=0}^{t-2} f\left(\frac{H(z_j)}{E^*}\right)^{x_j}\right]\,.
    \end{align}
    Continuing this until all expectations are evaluated yields the expression $w(\sigma)$. 
\end{proof}

\begin{proof}[Proof of Prop.~\ref{prop:sum_sigma_in_Gamma}]
 Let $j_0 := \lceil [\ln(1-\eta)-\ln(\mcEbar)]/\ln(1-\alpha)\rceil 
 \leq \lceil[ -\ln(1-\eta)+\ln(\mcEbar)]/\alpha \rceil $, so that $\mcEbar(1-\alpha)^{j_0} \leq 1-\eta$. Recall that $f(x) =0$ when $x \leq 1-\eta$ and note that $(1-\alpha)^j \geq 1-j\alpha$. Working from the definition of $w(\sigma)$, we have
\begin{align}
    \sum_{\sigma \in \Gamma} w(\sigma) &= \prod_{j=0}^{\infty}\sum_{x_j=0}^\infty  b^{x_j} f(\mcEbar(1-\alpha)^j)^{x_j} = \prod_{j=0}^\infty \frac{1}{1-bf(\mathcal{E}(1-\alpha)^{j})} \\
    &\geq \prod_{j=0}^\infty \exp(bf(\mcEbar(1-\alpha)^{j})) = \prod_{j=0}^{j_0-1} \exp(bf(\mcEbar(1-\alpha)^{j})) \\
    &=  \exp\left(\frac{b}{\eta}\sum_{j=0}^{j_0-1}\left(\mcEbar(1-\alpha)^{j}-(1-\eta)\right)\right) \\
    &= \exp\left(\frac{b}{\eta}\left( \frac{\mcEbar-\mcEbar(1-\alpha)^{j_0}}{\alpha}-(1-\eta)j_0\right)\right) \\
    &\geq \exp\left(\frac{b}{\eta}\left( \frac{\mcEbar-(1-\eta)}{\alpha}-(1-\eta)j_0\right)\right) \\
    &\geq \exp\left(\frac{b}{\eta}\left( \frac{\mcEbar-(1-\eta)}{\alpha}-\frac{(1-\eta)\ln\left(\frac{\mcEbar}{1-\eta}\right)}{\alpha}\right)\right) \\
    &= \exp\left(\frac{b\mcEbar}{\eta\alpha}\left(1-\frac{1-\eta}{\mcEbar}+\frac{1-\eta}{\mcEbar}\ln\left(\frac{1-\eta}{\mcEbar}\right) \right)\right)
\end{align}
which, recalling the definition of $F$ in Eq.~\eqref{eq:F_def}, proves the proposition. 
\end{proof}

\begin{proof}[Proof of Prop.~\ref{prop:not_in_expansion}]
For any sequence $\sigma = (x_0,x_1,\ldots)$ not in the set $\Lambda_\ell$, there must exist some integer $a$ such that $x_a > 0$ and $\sum_{j=0}^a x_j > \ell-a$ (otherwise it appears in the expansion of $(A+B)^\ell$). For each $\sigma$, let $a_\sigma$ denote the minimum such integer. For each $a > 0$, let $\tau_a$ be the sequence with $x_0 = \max(0,\ell-a)$ and $x_a = 1$, and $x_j = 0$ for all other $j$. Let $\tau_0$ be the sequence with $x_0 = \ell+1$ and $x_j=0$ for all other $j$. Then we can let $\rho_{\sigma} = \sigma - \tau_{a_\sigma}$ (under entry-wise subtraction). Note that the $x_0$ entry for $\rho_{\sigma}$ can be negative, but no smaller than $\min(0,-\ell + a_\sigma)$ and also no smaller than $-\sum_{j=1}^\infty x_j$. We have
\begin{equation}
    w(\sigma) = w(\tau_{a_{\sigma}}) w(\rho_{\sigma}) = b^{\max(1,\ell-a_\sigma+1)}f(\mcEbar(1-\alpha)^{a_\sigma}) w(\rho_{\sigma}) 
\end{equation}
We can break up the sum into two terms, associated with $a_\sigma \leq \ell$ or $a_\sigma > \ell$:
\begin{align}
    \sum_{\sigma \not\in \Lambda_\ell} w(\sigma) &=  \sum_{a=0}^\ell b^{\ell-a+1}f(\mcEbar(1-\alpha)^{a}) \sum_{\substack{\sigma \not\in \Lambda_\ell \\ a_\sigma = a }} w(\rho_{\sigma}) + \sum_{a=\ell+1}^\infty b f(\mcEbar(1-\alpha)^a) \sum_{\substack{\sigma \not\in \Lambda_\ell \\ a_\sigma = a }} w(\rho_{\sigma})
\end{align}
Examine the second term. When $a_\sigma > \ell$, $\rho_\sigma$ has no negative entries, meaning $\rho_\sigma \in \Gamma$. Moreover, for any fixed choice of $a$, the mapping $\sigma \mapsto \rho_\sigma$ restricted to $\sigma$ for which $a_\sigma = a$ is an injective map. Thus, the sum $\sum_{\substack{\sigma \not\in \Lambda_\ell \\ a_\sigma = a }} w(\rho_{\sigma})$ is less than $\sum_{\sigma \in \Gamma} w(\sigma)$. Additionally, since $f(x) \leq x$, and $\sum_{a=\ell+1}^\infty (1-\alpha)^a = \alpha^{-1}(1-\alpha)^{\ell+1}$, the second term above is bounded by $b\mcEbar\alpha^{-1}(1-\alpha)^{\ell+1} \sum_{\sigma \in \Gamma} w(\sigma)$. 

Now examine the first term. When $a \leq \ell$, the map $\sigma \mapsto \rho_\sigma$ is still injective, but the $x_0$ entry of $\rho_\sigma$ can be negative. If we sum over all possible sequences $(x_0,x_1,\ldots)$ for which $x_j \geq 0$ when $j >0$ and $x_0 \geq -\sum_{j=1}^\infty x_j$, every $\rho_{\sigma}$ for which $a_\sigma = a$ will appear (exactly) once in the sum. Thus, we have
\begin{align}
    \sum_{\substack{\sigma \not\in \Lambda_\ell \\ a_\sigma = a }} w(\rho_{\sigma}) \leq{}&\sum_{x_1=0}^\infty\sum_{x_2=0}^\infty\cdots \sum_{x_0 = -\sum_{j=1}^\infty x_j}^{\infty}  b^{\sum_{j=0}^\infty x_j}\prod_{j=0}^\infty f(\mcEbar(1-\alpha)^j)^{x_j }  \\
    ={}&\sum_{x_1=0}^\infty\sum_{x_2=0}^\infty\cdots \sum_{x_0' = 0}^{\infty}  b^{x_0'}f(\mcEbar)^{x_0'-\sum_{j=1}^\infty x_j}\prod_{j=1}^\infty f(\mcEbar(1-\alpha)^{j})^{x_j } \\
    ={}&\sum_{x_1=0}^\infty\sum_{x_2=0}^\infty\cdots \sum_{x_0' = 0}^{\infty}  (bf(\mcEbar))^{x_0'}\prod_{j=1}^\infty \left(\frac{f(\mcEbar(1-\alpha)^{j})}{f(\mcEbar)}\right)^{x_j } \\
    ={}& \frac{1}{1-bf(\mcEbar)}\prod_{j=1}^\infty \frac{1}{1-\frac{f(\mcEbar(1-\alpha)^j)}{f(\mcEbar)}}
\end{align}
Since $f$ is convex $f(\mcEbar (1-\alpha)^j)) \leq (1-\alpha)^j f(\mcEbar)$ and thus the final expression is less than 
\begin{equation}
    \frac{1}{1-bf(\mcEbar)}\prod_{j=1}^\infty \frac{1}{1-(1-\alpha)^j}\,.
\end{equation}
We evaluate a bound on the logarithm as follows. Note that for a monotonically decreasing function $g$ we have $\sum_{j=j_0}^\infty g(j) \leq \int_{j_0-1}^\infty g(u) du$.
\begin{align}
    \ln\left(\prod_{j=1}^\infty \frac{1}{1-(1-\alpha)^j}\right) &= -\sum_{j=1}^\infty \ln(1-(1-\alpha)^j) = \sum_{j=1}^\infty \sum_{k=1}^\infty \frac{(1-\alpha)^{jk}}{k} \\
    &= \sum_{j=1}^\infty (1-\alpha)^j + \sum_{k=2}^\infty \sum_{j=1}^\infty \frac{(1-\alpha)^{jk}}{k} = \frac{1-\alpha}{\alpha} + \sum_{k=2}^\infty \sum_{j=1}^\infty \frac{(1-\alpha)^{jk}}{k} \\ 
    &\leq \frac{1-\alpha}{\alpha} + \int_{1}^\infty dy \frac{1}{y}\int_{0}^\infty dx (1-\alpha)^{xy} \\
    &= \frac{1-\alpha}{\alpha} - \int_{1}^\infty dy \frac{1}{y}\frac{1}{y\ln(1-\alpha)} \\
    &=\frac{1-\alpha}{\alpha} - \frac{1}{\ln(1-\alpha)} \leq \frac{2}{\alpha}
\end{align}
Combining these observations, we have
\begin{align}
    \sum_{\sigma \not\in \Lambda_\ell} w(\sigma) &\leq \left(\sum_{a=0}^\ell b^{\ell-a+1}f(\mcEbar(1-\alpha)^a)\right)\frac{1}{1-bf(\mcEbar)}e^{2/\alpha} + b \mcEbar \alpha^{-1}(1-\alpha)^{\ell+1}\sum_{\sigma \in \Gamma}w(\sigma) \\
    &\leq \left(\sum_{a=0}^\ell b^{\ell-a+1}\mcEbar(1-\alpha)^a\right)\frac{1}{1-bf(\mcEbar)}e^{2/\alpha} + b \mcEbar \alpha^{-1}(1-\alpha)^{\ell+1}\sum_{\sigma \in \Gamma}w(\sigma) \\
    &\leq b\mcEbar(1-\alpha)^{\ell+1}\frac{1}{1-\alpha-b}e^{2/\alpha}+ b \mcEbar \alpha^{-1}(1-\alpha)^{\ell+1}\sum_{\sigma \in \Gamma}w(\sigma) 
\end{align}
which entails the proposition statement. 
\end{proof}

\subsection{Tail bound for certain cost functions}\label{app:tail_bounds}
\begin{proposition}\label{prop:CSP_tail_bound}
    Let $H = \sum_{j=1}^m \mathcal{C}_j$ be a MAX-$k$-CSP instance with $m$ clauses, and let $E^*$ be its optimal value. 
    Let $d_j$ be the total number of constraints involving the $j$th variable and let $C(E)$ be the number of assignments $z$ such that $H(z) \leq E$. Define $D = nk^{-2}m^{-2}\sum_{j=1}^n d_j^2$.   Then for any $0 \leq \eta \leq 1$
    \begin{equation}
        C(E^*(1-\eta)) \leq 2^{(1-\gamma)n} 
    \end{equation}
    where
    \begin{align}
        \gamma = \left(\frac{|E^*|}{m}\right)^2\frac{(1-\eta)^2}{\ln(2)2^{2k}k^2 D} 
    \end{align}
    If $d_j \leq 2km/n$ for all $j$, then $D \leq 2$ and we can write
    \begin{equation}
        \gamma \geq \left(\frac{|E^*|}{m}\right)^2\frac{(1-\eta)^2}{2\ln(2)2^{2k}k^2}\,.
    \end{equation}
    Furthermore, if each clause $\mathcal{C}_j$ involves exactly $k$ distinct bits (with $k \geq 2$), which are chosen independently and uniformly at random for each $j$, then the expected value of $D$ is upper bounded by $1+n/(km)$, implying most random instances obey $D \leq 3$ so long as $m \geq n$. 
\end{proposition}
\begin{proof}
    We will use McDiarmid's inequality \cite{mcdiarmid1989method}, which can be stated as follows. Suppose a real function $H$ of $n$ independent random variables $z_1,\ldots, z_n$ has the property that changing the value of $z_i$ while leaving the other $n-1$ variables constant can change the value of $H$ by at most an amount $\Delta_i$. Then, when the $z_i$ are chosen randomly, the probability that $H$ deviates by more than $\delta$ from its mean is at most $e^{-2\delta^2/\sum_i
    \Delta_i^2}$.
    
    Applying this to our problem, we may choose $z_i \in \{+1,-1\}$ uniformly at random for each $i =1,\ldots,n$, and note that since $z_i$ participates in at most $d_i$ clauses, changing $z_i$ can change the value of $H$ by at most $\Delta_i \leq 2^kd_i$ (recall that the value of a clause is restricted to the interval $[-1,2^k-1]$). The mean of $H$ is precisely 0, and thus for $E<0$, the quantity $C(E)/2^n$ is the probability that for a randomly chosen $z \in \{+1,-1\}^n$, $H(z)$ deviates beneath its mean by at least a value $|E|$. Thus, for $E < 0$, McDiarmid's inequality gives us the Gaussian tail bound
    \begin{align}
        C(E) &\leq 2^n \exp\left(-\frac{2E^2}{2^{2k}\sum_{j=1}^n d_j^2}\right) \\
        &\leq 2^n \exp\left(-\frac{2E^2n}{2^{2k}k^2m^2 D}\right) \label{eq:tail_bound_McDiarmid}
    \end{align}
    Plugging in $E = E^*(1-\eta)$ into Eq.~\eqref{eq:tail_bound_McDiarmid}  yields the quoted result. 
    Next, recalling that $\sum_{j=1}^n d_j \leq km$ for MAX-$k$-CSP instances, and imposing $d_j \leq 2km/n$ for all $j$, implies that $D \leq 2$. 
    
   Now consider the case where each $\mathcal{C}_j$ is chosen randomly as described. Each $d_j$ is given by the sum of $m$ independent Bernoulli random variables that are equal to 1 with probability $k/n$ and 0 otherwise (although note that the $d_j$ are not independent for different values of $j$). The average value of $d_j$ is given by $\EV[d_j] = km/n$, and the average value of $d_j^2$ is given by $\EV[d_j^2] = mk/n+m(m-1)k^2/n^2$. Thus, the expected value of $D$ is given by $\EV[D] = (m + n/k -1)/m  < 1.5$. A simple Markov inequality implies that the probability that $D$ exceeds 3 is upper bounded by  $\EV[D]/3 < 1/2$, implying the quoted statement. 
\end{proof}

\begin{proposition}\label{prop:k-spin_expected_energy}
     Let $\mathcal{J}_{n,k}$ denote the expected value of the optimal cost $E^*$ when the cost function $H$ is drawn from the $k$-spin model with $n$ spins, as defined in Eq.~\eqref{eq:k-spin}. Then
     \begin{equation}
         \mathcal{J}_{n,k} \leq -n \frac{\sqrt{2}}{\sqrt{\pi}k}
     \end{equation}
     for any $n,k$ for which $n$ is a multiple of $k$. When $n$ is not a mutliple of $k$, we have $\mathcal{J}_{n,k} \leq -k\lfloor n/k\rfloor \frac{\sqrt{2}}{\sqrt{\pi}k}$. Note that as long as $n \geq k$, k$\lfloor n/k\rfloor \geq n/2$, and thus for all $n,k$, we can say that
     \begin{equation}
         \mathcal{J}_{n,k} \leq -n \frac{1}{\sqrt{2\pi}k}
     \end{equation}
\end{proposition}
\begin{proof}

    First, let us fix $k$ and consider varying $n$. Suppose $n_A < n_B$. We now construct a (correlated) joint ensemble over cost functions $(H_A, H_B)$ where the marginal distribution on $H_A$ is the $k$-spin model with $n=n_A$ and the marginal distribution on $H_B$ is the $k$-spin model with $n=n_B$.  First $H_A$ is chosen randomly from the $k$-spin ensemble with $n=n_A$. Note that $H_A$ has $\binom{n_A}{k}$ terms. Then $H_B$ is generated by choosing the coefficients of these $\binom{n_A}{k}$ terms of $H_B$ to be the same as those of $H_A$ and choosing the other $\binom{n_B}{k}-\binom{n_A}{k}$ coefficients as i.i.d.~standard Gaussian variables. Let $E^*_A$ and $E^*_B$ denote the optimal costs of $H_A$ and $H_B$, and $z^*_A$ and $z^*_B$ the optimal assignments. We have $\EV[E^*_A] = \mathcal{J}_{n_A,k}$ and $\EV[E^*_B] = \mathcal{J}_{n_B,k}$, where $\EV$ denotes expectation over the joint distribution defined above,  since the marginal distributions over $H_A$ and $H_B$ are precisely the $k$-spin ensemble. However, we also note that, due to the way the joint distribution was constructed, $\EV[H_B(z^*_A)|H_A] = E^*_A$, where $\EV[\cdot|H_A]$ denotes expectation conditioned on a fixed value of $H_A$. Moreover, $E^*_B \leq H_B(z^*_A)$ always holds. These observations imply that 
    \begin{equation}
        \mathcal{J}_{n_A,k} \geq \mathcal{J}_{n_B,k} \qquad \text{ when } n_A \leq n_B
    \end{equation}
    This statement means that $\mathcal{J}_{n,k} \geq \mathcal{J}_{k\lfloor n/k \rfloor,k}$ and thus the proposition statement when $n$ is not a multiple of $k$ follows directly from the case when $n$ is a multiple of $k$. Henceforth assume $n$ is a multiple of $k$.

    Now, we prove the formula for increasing $k$ by induction. First, consider the base case.
    For $k=1$, the expected optimal value $\mathcal{J}_{n,1} = -n\sqrt{2/\pi}$. This is because, for every instance, one can choose the value of $z_1,\ldots, z_n$ such that each of the $n$ degree-1 terms is negative (no frustration), and $E^*$ is the sum of the absolute value of $n$ independent random Gaussian variables, which can be readily calculated. 
    
    Now, suppose the proposition statement holds for $\mathcal{J}_{n,k-1}$ for all choices of $n$. We will show that for all choices of $n$, $\mathcal{J}_{n,k} \leq  \mathcal{J}_{n',k-1}$ with $n' := (k-1)n/k$. Consider an instance $H$ chosen from the $k$-spin model with $n$ spins, where $n$ is a multiple of $k$, defined by its coefficients $J_{i_1,\ldots,i_k}$ as in Eq.~\eqref{eq:k-spin}. Let $t = n/k$.  There are $\binom{n}{k}$ terms of $H$, and $t \binom{n-t}{k-1}$ of these terms will include exactly one of the variables $z_1,\ldots, z_t$. Thus we can express $H = H_C + H_D$ where $H_C$ includes these terms and $H_D$ includes the rest of the terms. We can write $H_C$ as 
    \begin{align}
        H_C(z_1,\ldots,z_n) = \sqrt{\frac{k!}{n^{k-1}}}\sum_{i_1=1}^t z_{i_1} \left(\sum_{t < i_2 < \ldots < i_k \leq n} J_{i_1,\ldots,i_k} z_{i_2}\ldots z_{i_k}\right)
    \end{align}
    Now consider the MAX-E$(k-1)$-LIN2 instance on $n-t =: n'$ variables $z_{t+1}\ldots, z_n$
    \begin{equation}\label{eq:barH_C}
        \bar{H}_C(z_{t+1},\ldots, z_n) = \sqrt{\frac{(k-1)!}{tn^{k-2}}}\sum_{t < i_2 < \ldots < i_k \leq n} \left(\sum_{i_1 = 1}^t J_{i_1,\ldots,i_k}\right) z_{i_2}\ldots z_{i_k}
    \end{equation}
    and note that
    \begin{equation}
        H_C(1,\ldots,1,z_{t+1},\ldots z_n) = \sqrt{kt/n}\bar{H}_C(z_{t+1},\ldots,z_n).
    \end{equation}
    The sum over $J_{i_1,\ldots,i_k}$ from $i_1 = 1$ to $t$ in Eq.~\eqref{eq:barH_C} is a sum of $t$ independent standard Gaussian-distributed random variables, which is itself Gaussian-distributed with mean zero and variance equal to $t$. This factor of $t$ will precisely cancel the $1/\sqrt{t}$ in Eq.~\eqref{eq:barH_C}, allowing us to conclude that the ensemble over $\bar{H}_C$ is precisely the $(k-1)$-spin ensemble with $n'$ spins (note that $n'$ is a multiple of $k-1$). Let $\bar{z}^*$ be the optimal assignment to $\bar{H}_C$, with energy $\bar{E}^*_C$ and note that $\EV[\bar{E}^*_C] = \mathcal{J}_{n',k-1}$.  We let $z$ be equal to $\bar{z}^*$ with the additional assignments $z_1 = z_2 = \ldots =z_t=1$, so that $H_C(z) = \bar{E}^*_C$. Since the coefficients of $H_D$ are chosen independently from $H_C$, we have $\EV[H_D(z)] = 0$ when $z$ is chosen as above. Thus $\EV[H(z)] =\mathcal{J}_{n',k-1}$, and since the optimal value can only be smaller than $H(z)$, we have
    \begin{equation}
        \mathcal{J}_{n,k} \leq \mathcal{J}_{n',k-1}\,.
    \end{equation}
    We now insert the bound on $\mathcal{J}_{n',k-1}$, finding that
    \begin{equation}
        \mathcal{J}_{n,k} \leq -\left(n-\frac{n}{k} \right) \frac{\sqrt{2}}{\sqrt{\pi}(k-1)} = -n\frac{\sqrt{2}}{k\sqrt{\pi}}\,.
    \end{equation}
    which proves the proposition. 
\end{proof}

\begin{proposition}\label{prop:k-spin_expected_energy_concentration}
When $H$ is drawn randomly from the $k$-spin model with $n$ spins, let $\mathcal{J}_{n,k}$ denote the expected value of $E^*$. For any $\delta$
\begin{align}
    \Pr_H[E^* \geq \mathcal{J}_{n,k} + \delta n] &\leq e^{-\delta^2n/2} \,,
\end{align}
where $\Pr_H$ denotes probability over random choice of $H$ from the $k$-spin ensemble. 
\end{proposition}
\begin{proof}
We may think of the quantity $E^*\sqrt{n^{k-1}/k!}$ as a function of the $\binom{n}{k}$ coefficients $J_{i_1,\ldots,i_k}$,  as in Eq.~\eqref{eq:k-spin}. Note that this function is 1-Lipschitz as changing a single coefficient by some amount $\Delta$ can change $E^*\sqrt{n^{k-1}/k!}$ by at most the same amount $\Delta$. Applying Lemma 1 of Ref.~\cite{montanaro2020branchAndBound}, we have for any $t > 0$
\begin{equation}
    \Pr_H\left[E^*\sqrt{n^{k-1}/k!} \geq \EV_H[E^*]\sqrt{n^{k-1}/k!} + t\right] \leq \exp\left(-\frac{{t}^2}{2\binom{n}{k}}\right)
\end{equation}
which, noting $\EV_H[E^*] = \mathcal{J}_{n,k}$, and letting $\delta = t\sqrt{k!}/n\sqrt{n^{k-1}}$ can be rewritten as
\begin{align}
        \Pr_H\left[E^* \geq J_{n,k} + \delta n\right]   \leq \exp\left(-\frac{\delta^2n^{k+1}}{2k!\binom{n}{k}}\right) \leq e^{-\delta^2 n/2}\,.
\end{align}
\end{proof}

\begin{proposition}\label{prop:k-spin_tail_bound}
Suppose $H$ is drawn at random from the $k$-spin model. Let  $E^*$ denote its optimal value and let $C(E)$ be the number of assignments $z$ such that $H(z) \leq E$.  Then for any $0 \leq \eta \leq 1$, with probability at least $1-2^{-\gamma n +1}$ over choice of $H$,
    \begin{equation}
        C(E^*(1-\eta)) \leq 2^{(1-\gamma)n} 
    \end{equation}
    where
    \begin{equation}
        \gamma = \frac{(1-\eta)^2}{32\pi \ln(2)k^2}
    \end{equation}
\end{proposition}
\begin{proof}
For a fixed bit string $x$ and randomly chosen $k$-spin instance $H$, the values of each of the $m = \binom{n}{k}$ terms in the Hamiltonian are independently random since the coefficients are independent. Thus, $H(x)$ is distributed as the sum of $m$ i.i.d.~standard Gaussian variables with mean 0 and standard deviation $\sqrt{k!/n^{k-1}}$, which is equivalent to a Gaussian with mean 0 and standard deviation $\sigma$, with $\sigma^2 = mk!/n^{k-1} = n!/((n-k)!n^{k-1})<n$. Thus for any $t$ we have
\begin{equation}
    \Pr_H\left[H(x) \leq  -t \right] = \frac{1}{2}\Erfc\left[\frac{t}{\sqrt{2\sigma^2}}\right]  \leq \frac{1}{2}\Erfc\left[\frac{t}{\sqrt{2n}}\right] \leq e^{-t^2/2n} = 2^{-t^2/(2n\ln(2))}\,
\end{equation}
where $\Pr_H$ denotes probability over random choice of $H$ from the $k$-spin ensemble. 
Since $C(E) = \sum_{x} \mathbbm{1}(H(x) \leq E)$, where $\mathbbm{1}$ is the indicator function, we can compute a bound on the expectation value of $C(E)$ as
\begin{align}
    \EV_H[C(E)] = \sum_{x}\EV_H[\mathbbm{1}(H(x) \leq E)]=\sum_x \Pr_H\left[H(x) \leq E\right] \leq 2^n2^{-E^2/(2n\ln(2))}\label{eq:exp_C(E)_bound}\,.
\end{align}

Separately, let us apply Prop.~\ref{prop:k-spin_expected_energy} and Prop.~\ref{prop:k-spin_expected_energy_concentration} with $\delta = 1/k\sqrt{8\pi}$. We find
\begin{align}
    \Pr_H\left[E^* \geq -\frac{n}{2k\sqrt{2\pi}}\right] &= \Pr_H\left[E^* \geq -n\frac{1}{k\sqrt{2\pi}} + n\frac{1}{2k\sqrt{2\pi}}\right] \leq \Pr_H\left[E^* \geq \mathcal{J}_{n,k} + \delta n \right] \\
    &\leq e^{-\delta^2 n / 2} = \exp\left(-\frac{n}{16\pi k^2}\right)\label{eq:k-spin-high-prob-Estar-bound}
\end{align}
Conditioned on $E^* \geq -n/(2k\sqrt{2\pi})$, which is true with high probability by Eq.~\eqref{eq:k-spin-high-prob-Estar-bound}, we can say from Eq.~\eqref{eq:exp_C(E)_bound} that
\begin{equation}
    \EV_H[C(E^*(1-\eta))] \leq 2^{n\left(1-\frac{(1-\eta)^2}{16k^2\pi\ln(2)}\right)}
\end{equation}
and, by Markov's inequality
\begin{equation}
\Pr_H[C(E^*(1-\eta)) \geq 2^{(1-\gamma)n}] \leq 2^{-n\left(\frac{(1-\eta)^2}{16k^2\pi\ln(2)}-\gamma\right)}\label{eq:markov-inequality-result}
\end{equation}
Choosing $\gamma = (1-\eta)^2/(32 k^2 \pi \ln(2))$ yields the desired bound on $C(E^*(1-\eta))$. By the union bound, the chances this bound fails is at most the sum of the right-hand-sides of Eqs.~\eqref{eq:k-spin-high-prob-Estar-bound} and \eqref{eq:markov-inequality-result}, which gives
\begin{equation}
    2^{-n\left(\frac{(1-\eta)^2}{32k^2\pi\ln(2)}\right)} + e^{-n\frac{1}{16\pi k^2}} \leq 2 \exp\left(-\frac{n(1-\eta)^2}{32 k^2 \pi}\right) = 2^{-\gamma n + 1}\,,
\end{equation}
which proves the proposition. 
\end{proof}

\section{Implementation of jump steps}\label{app:jumps}

We assume we have access to a $\poly(n)$ size circuit that enacts a block-encoding of both the beginning and ending Hamiltonians, where (following standard convention) a $(\kappa,a)$-block encoding of an $n$-qubit Hamiltonian $K$ is defined to be an $(n+a)$-qubit unitary $U$ for which 
\begin{equation}
    \left(\bra{0}^{\otimes a} \otimes I \right) U \left(\ket{0}^{\otimes a} \otimes I\right)  = K/\kappa
\end{equation}

\begin{proposition}[Jump $K_1 \rightarrow K_2$ with parameters $(E_1,E_2,\Delta_1,\Delta_2,p,\delta)$]\label{prop:jumps}
    Suppose $K_1$ and $K_2$ are $n$-qubit Hamiltonians and let $E_1,E_2,\Delta_1,\Delta_2,p,\delta$ be known positive parameters. Let $U_i$ be a $(\kappa_i,a)$-block-encoding of $K_i$ for $i\in {1,2}$.  Suppose that the following properties are satisfied.
    \begin{enumerate}[(A)]
        \item $K_1$ has one eigenstate with energy at most $E_1$, denoted by $\ket{\psi_1}$
        \item $\lVert \Pi_2 \ket{\psi_1}\rVert^2 \geq p$, where $\Pi_2$ denotes the projector onto the subspace spanned by eigenvectors of $K_2$ with energy at most $E_2$
        \item For $j \in \{1,2\}$, at least one of the following holds:
        \begin{enumerate}[(i)]
        \item $K_j$ has no eigenvalues in the interval $(E_j, E_j+\Delta_j)$ 
        \item $K_j$ can be represented as a diagonal matrix in either the computational basis or the Hadamard basis and each of its entries in that basis can be computed in $\poly(n)$ classical time.
        \end{enumerate}
    \end{enumerate}
    
    Then, there is a universal constant $D$ such that we can construct a quantum circuit that implements a unitary $U$ for which 
    \begin{equation}
        \left \lVert U(\ket{\psi_1}\otimes\ket{0}^{\otimes a}) - \frac{\Pi_2\ket{\psi_1}\otimes\ket{0}^{\otimes a}}{\lVert \Pi_2\ket{\psi_1}\rVert } \right\rVert \leq \delta 
    \end{equation}
    and consists of $\poly(n)$ gates along with
    \begin{equation}
        \frac{D \kappa_i }{\Delta_i \sqrt{p}}\log(\delta^{-1})\log(p^{-1/2}\delta^{-1}\log(\delta^{-1}))
    \end{equation}
    calls to $U_j$ and controlled-$U_j$, for $j \in \{1,2\}$. If $K_j$ is diagonal in the computational basis and each of its entries can be computed in $\poly(n)$ classical time, then no calls to $U_j$ are necessary, and $\Delta_j$ is irrelevant.
    
    When acting with $U$ on $\ket{\psi_1}$, we say we are performing the jump $K_1 \rightarrow K_2$. 
\end{proposition}

\begin{proof}

We can construct $U$ through a combination of fixed-point amplitude amplification \cite{yoder2014fixedpoint} and what essentially amounts to phase estimation, but enacted with the quantum singular value transformation (QSVT) framework \cite{gilyen2019QSVT}. First assume that conditions (A),(B), and (C)(i) are satisfied. For $j \in \{1,2\}$, let $q_j:[-1,1] \rightarrow [-1,1]$ denote the step function 
\begin{equation}
q_j(x) = \begin{cases}
-1 & \text{if } x \leq 0 \\
1 & \text{if } x \geq 0
\end{cases}
\end{equation}
Let $R_1$ denote the unitary reflection operator which applies a $-1$ phase to the state $\ket{\psi_1}$ and acts as identity on the subspace orthogonal to $\ket{\psi_1}$. Similarly, let $R_2$ denote the unitary reflection that applies a $-1$ phase to the support of $\Pi_2$. By conditions (A), (B) and (C)(i) we can say that
\begin{equation}
R_j = q_j\left(\frac{K_j-E_j-\frac{\Delta_j}{2}}{2\kappa_j}\right)
\end{equation}

Let $L$ be an integer to be specified later. Let $\delta' = \delta/{4L}$. Then, there is an odd polynomial $q'_j:[-1,1]\rightarrow [-1,1]$ with degree at most $d = O(\log(1/\delta')\kappa_j/\Delta_j)$ for which \cite{gilyen2019QSVT}

\begin{align}
q_j'(x) \leq -1+\delta' \qquad & \text{ if } x \leq  -\frac{\Delta_j}{4\kappa_j} \label{eq:step_approx_low}\\
q_j'(x) \geq \;\;\;1-\delta' \qquad & \text{ if } x \geq  \frac{\Delta_j}{4\kappa_j}\label{eq:step_approx_high}
\end{align}
which approximates the step function $q_j(x)$. Define the operator
\begin{equation}
R'_j = q'_j\left(\frac{K_j-E_j-\frac{\Delta_j}{2}}{2\kappa_j}\right)
\end{equation}
Since $q'_j$ is an odd polynomial of degree $d$ whose range is in $[-1,1]$ on the domain $[-1,1]$, the quantum singular value transformation (QSVT) \cite{gilyen2019QSVT} allows a $(2\kappa_i,a+O(1))$ block-encoding of $R'_j$ to be implemented using $d$ calls to a block-encoding (and its inverse) of the operator $\frac{-E_j-\Delta_j/2}{2\kappa_j}I+\frac{K_j}{2\kappa_j}$, which itself can be constructed using a controlled-$U_j$ gate via linear combination of block encoding with the identity. 

Note that $K_j$, $R_j$ and $R'_j$ share a common set of eigenvectors. Using condition (C)(i), every eigenvalue of $\frac{-E_j-\Delta_j/2}{2\kappa_j}I+\frac{K_j}{2\kappa_j}$ satisfies either the condition in Eq.~\eqref{eq:step_approx_low} or \eqref{eq:step_approx_high}, and thus
\begin{equation}
\lVert R_j - R'_j \rVert \leq \delta'
\end{equation}

Fixed-point amplitude amplification \cite{yoder2014fixedpoint} describes how access to (controlled) reflection operators $R_1$ and $R_2$ allows one to construct a unitary $V$ for which 
\begin{equation}\label{eq:FPAA}
\left\lVert V(\ket{\psi_1}\otimes \ket{0}^{\otimes a})  - \frac{\Pi_2 \ket{\psi_1} \otimes \ket{0}^{\otimes a}}{\lVert \Pi_2 \ket{\psi_1} \rVert} \right\rVert \leq \delta/2
\end{equation}
as long as $\lVert \Pi_2 \ket{\psi_1}\rVert \geq \sqrt{p}$. The construction uses $L = O(\log(1/\delta)/\sqrt{p})$ calls to controlled-$R_1$ and the same number of calls to controlled-$R_2$. By replacing $R_1$ and $R_2$ by operators $R'_1$ and $R'_2$ to form $U$, at most $2L$ unitaries are modified, and each by at most an amount $\delta'$. By the triangle inequality, errors accrue linearly and we have
\begin{equation}
\lVert U- V \rVert \leq 2L\delta' = \delta/2
\end{equation}
and this equation, together with Eq.~\eqref{eq:FPAA}, proves that $U$ has the action claimed by the proposition. The number of calls $U$ makes to the block encodings of $K_1$ and $K_2$ is precisely the $L$ calls it makes to controlled-$R'_1$ and controlled-$R'_2$ times the $d$ calls $R'_1$ and $R'_2$ make to controlled-$K_1$ and controlled-$K_2$. Thus, the calls to $K_1$ and $K_2$ are doubly controlled.  

If condition (C)(ii) is true for either $K_1$, $K_2$, or both, we modify the argument to construct the associated approximate reflection operator $R'_j$ in a different way (actually, we will be able to construct an \textit{exact} reflection operator). First assume that $K_j$ is diagonal in the computational basis, and suppose that $T$ bits are sufficient to represent the diagonal entries of $K_j$ in binary. Let the operator $R_j$ be implemented by four steps: first, using at most $A = \poly(n)$ ancilla qubits initially in the state $\ket{0}$, we perform the operation
\begin{equation}
\ket{i}\ket{0}^{\otimes A} \mapsto \ket{i}\ket{(K_j)_{ii}}\ket{0}^{\otimes (A-T)}
\end{equation}
using reversible classical arithmetic in superposition, where $(K_j)_{ii}$ denotes the $i$th diagonal entry of the diagonal operator $K_j$. Second, we compute whether $(K_j)_{ii} \leq E_j$ or $(K_j)_{ii} > E_j$ into another ancilla bit
\begin{equation}
\ket{i}\ket{(K_j)_{ii}}\ket{0}^{\otimes (A-T)} \mapsto \ket{i}\ket{(K_j)_{ii}}\ket{(K_j)_{ii} \leq E_j}\ket{0}^{\otimes (A-T-1)}
\end{equation}
again using reversible classical arithmetic. Third, we apply a Pauli-$Z$ gate to the $\ket{(K_j)_{ii} \leq E_j}$ bit. Fourth, we uncompute all the ancillas, resetting them to $\ket{0}$. This yields the overall operation
\begin{equation}
\ket{i}\ket{0}^{\otimes A} \mapsto (-1)^{((K_j)_{ii} \leq E_j)}\ket{i}\ket{0}^{\otimes A}
\end{equation}
which is precisely the operation $R_j$. The complexity of implementing this operation is $\poly(n)$ since we have assumed in condition (C)(ii) that all arithmetic is efficient. Note that this is possible regardless of how small $\Delta_j$ is due to the fact that we can use classical arithmetic to distinguish very near eigenvalues. If $K_j$ is diagonal in the Hadamard basis, we may perform the exact same protocol with the modification that a layer of $n$ single-qubit Hadamard gates is applied at the beginning and end of the circuit to rotate into the basis where $K_j$ is diagonal. 
\end{proof}

Below, we show how to implement block-encodings of the Hamiltonian $H_b$ needed to apply Prop.~\ref{prop:jumps}. Actually, we only show that a block-encoding can be implemented up to exponentially small precision; since the block-encoding is called at most an exponential number of times, this block-encoding error can be made smaller than the other errors that arise in Prop.~\ref{prop:jumps}. The construction outlined below is certainly not the most efficient method. For those interested in producing a block-encoding with optimized circuit complexity, we draw the reader's attention to explicit constructions of useful circuit primitives in Ref.~\cite{sanders2020compilationHeuristics}.

\begin{proposition}[Block-encoding $H_b$]
Suppose $H$ is an efficiently computable cost function on $n$ bits with known optimal value $E^*$. Suppose that $b < 1$ and $g:[-1,\infty)\rightarrow [-1,0]$ is an efficiently computable function. Then for any $\delta \geq \exp(-\poly(n))$, there is a unitary $U$ that $\delta$-approximates (in operator norm) a $(1+b,\poly(n))$-block-encoding of the Hamiltonian $H_b$, defined as
\begin{equation}
    H_b = -\frac{X}{n} + b\, g\left(\frac{H}{|E^*|}\right)\,,
\end{equation}
and the  circuit complexity of $U$ is polynomial in $n$. Moreover, controlled-$U$ and controlled-controlled-$U$ can also be efficiently implemented. 
\end{proposition}
\begin{proof}
    The first term of $H_b$ is diagonal in the Hadamard basis and the second term is diagonal in the computational basis. The diagonal elements in these bases are efficiently computable. Consider the second term. Let $h(z) = g(H(z)/|E^*|)$ and suppose the value of $h(z)$ can be exactly represented with $T = \poly(n)$ bits. Note also that $\arcsin(h(x))$ can be represented to accuracy $\exp(-\Omega(P))$ with $P$ bits. For any $z$, let $\theta_z$ denote the $P$-bit approximation to $\arcsin(h(z))$ We can form a $(1,\poly(n))$-block-encoding of $g(H/|E^*|)$ up to $\exp(-\poly(n))$ error by constructing a unitary that for any $n$-qubit computational basis state $\ket{z}$, uses classical arithmetic in superposition to perform the following sequence of operations using $T+P+A+1 = \poly(n)$ ancilla qubits:
    \begin{align}
        \ket{0}^{\otimes T}\ket{0}^{\otimes P}\ket{0}\ket{0}^{\otimes A}\ket{z} &\mapsto \ket{h(z)}\ket{0}^{\otimes P}\ket{0}\ket{0}^{\otimes A}\ket{z} \\
        &\mapsto  \ket{h(z)}\ket{\theta_z}\ket{0}\ket{0}^{\otimes A}\ket{z} \\
        &\mapsto  \ket{h(z)}\ket{\theta_z}\ket{0}\ket{0}^{\otimes A}\ket{z} \\
        &\mapsto  \ket{h(z)}\ket{\theta_z}\left(\sin(\theta_z) \ket{0} + \cos(\theta_z)\ket{1}\right)\ket{0}^{\otimes A}\ket{z} \\
        &\mapsto  \ket{h(z)}\ket{0}^{\otimes P}\left(\sin(\theta_z) \ket{0} + \cos(\theta_z)\ket{1}\right)\ket{0}^{\otimes A}\ket{z} \\
        &\mapsto  \ket{0}^{\otimes T}\ket{0}^{\otimes P}\left(\sin(\theta_z) \ket{0} + \cos(\theta_z)\ket{1}\right)\ket{0}^{\otimes A}\ket{z}
    \end{align}
    Noting that $\sin(\theta_z)$ approximates $h(z)$ up to error $\exp(-\Omega(P))$, we can see that this unitary $(1,T+P+A+1)$-block-encodes $\,g(H/|E^*|)$ up to the same overall error in operator norm. This block-encoding can be turned into a controlled-block-encoding with an arbitrary number of controls by controlling every operation in the above process.  
    
    A $(1,\poly(n))$-block-encoding of the first term of $H_b$ can be constructed in a similar fashion by conjugating with Hadamard gates and noting it becomes diagonal in the computational basis. As $H_b$ is a linear combination of the first and second terms with coefficients 1 and $b$, a $(1+b,\poly(n))$-block-encoding of $H_b$ can be given using the linear-combinations-of-block-encodings construction of Ref.~\cite{gilyen2019QSVT}. This can be implemented as a controlled-block-encoding with an arbitrary number of controls by controlling all the operations involved in the circuit. 
\end{proof}

\section{The case where the optimal value is unknown}\label{app:unknown_Estar}

In our specification of the algorithm in Sec.~\ref{sec:algorithm}, we assumed that the optimal value $E^*$ of the cost function $H$ is known ahead of time. In practice, this may often not be the case. In this section, we argue that if the algorithm described in the main text runs in time $O^*(2^{(0.5-c)n})$ when $E^*$ is known, then there is always a slightly more sophisticated algorithm that also runs in time $O^*(2^{(0.5-c)n})$, perhaps with larger polynomial overheads.

This is easy to see in the case of MAX-$k$-CSP, since in this case there are $m$ constraints, where each constraint takes either the value $-1$ or the value $s/(2^k-s)$ for some known quantity $s$ equal to the number of assignments that satisfy the constraints. Thus, the cost of any input string can always be expressed as a rational number $p/q$ with $q = (2^k-1)!$, and $|p|\leq m(2^k-1)!(2^k-1)$. Thus we could always simply enumerate over all possible values of $p$, which, viewing $k$ as a constant, adds overhead of only $\poly(n)$. As we have a theoretical guarantee on the runtime when we choose $E^*$ correctly, we can always terminate the algorithm at that runtime and never risk exceeding the $O^*(2^{(0.5-c)n})$ scaling. 

However, if we allow the constraints to be weighted with arbitrary weights, as is the case for the $k$-spin model, this strategy may not work, as there may be an exponential number of possible values of $E^*$. To handle this more general situation, we come up with something more sophisticated. 

\begin{proposition}\label{prop:unknown_Estar}
    Let $H$ be a cost function for which there are known values $q$, $Q$, and $\varepsilon$ such that $Q/q \leq \poly(n)$, $q \leq |E^*| \leq Q$, and $\varepsilon \geq e^{-\Omega(n)}$ is a lower bound on the separation between the optimal value of $H$ and the second-lowest non-optimal value. Suppose that, if $E^*$ were known, there is a known choice for $b$ and $\eta$ for which Algorithm \ref{algo:main_simple} runs in time $O^*(2^{(0.5-c)n})$. Then there is another quantum algorithm that runs in time $O^*(2^{(0.5-c)n})$ even when $E^*$ is unknown.
\end{proposition}
\begin{proof}
    First examine step 2 of the algorithm as described in the pseudocode in Algorithm \ref{algo:main_simple}. Even when we do not know $E^*$, we may try to run this step with a guess $W$ for $E^*$, in combination with choices $b'$ and $\eta'$ for the parameters $b$ and $\eta$. These choices only impact the second term of $H_b$, as defined in Eq.~\eqref{eq:Hb}. Looking at that term, note that the non-zero part of the function $b'\;g_{\eta'}(x/W)$ can be written as
    \begin{align}
        b'\;g_{\eta'}(x/W) &= b'\frac{\frac{x}{W} + 1 -\eta'}{\eta'} \\ 
        &= \frac{b'}{W\eta'} x + \frac{b'(1-\eta')}{\eta'} \\
        &=  Q^{-1} \phi x + \theta \phi\,,
    \end{align}
    where 
    \begin{align}
        \theta = \frac{W(1-\eta')}{Q} \qquad \qquad \phi = \frac{b'Q}{\eta' W} \label{eq:thetaphi}
    \end{align}
    Any choice of $(W,\eta',b')$ that gives rise to the same $(\theta,\phi)$ via Eq.~\eqref{eq:thetaphi} will result in an equivalent step 2 of the algorithm. In other words, specifying $(W,\eta',b')$ is one more degree of freedom than is necessary. 
    
    For a given $(\theta,\phi)$ choice, we let $b_{\theta,\phi}$ and $\eta_{\theta,\phi}$ denote the $b'$ and $\eta'$ values that would arise from choosing $W = |E^*|$ along with $\theta$ and $\phi$. Turning around Eq.~\eqref{eq:thetaphi}, we find
    \begin{align}
        \eta_{\theta,\phi} &= 1-\frac{Q\theta}{|E^*|} \\
        b_{\theta,\phi} &= \frac{|E^*|}{Q}\phi -\phi \theta
    \end{align}
    
    Note that the region $(\theta,\phi) \in [0,1]\times [0,Q/q]$ leads to a set of $(b_{\theta,\phi},\eta_{\theta,\phi})$ that cover the region $[0,1] \times [0,1]$. Note also that the magnitude of the partial derivatives $\partial b_{\theta,\phi}/\partial \phi$, $\partial b_{\theta,\phi}/\partial \theta$, $\partial \eta_{\theta,\phi}/\partial \phi$, $\partial \eta_{\theta,\phi}/\partial \theta$ are each bounded above by $Q/q$ on this region. Thus, if we cast a net of gridpoints over the region $[0,1] \times [0,Q/q]$ for parameters $(\theta,\phi)$ with spacing $\delta q/2Q$ in each dimension, we can say that for the ``correct'' choices $(b,\eta)$ (known ahead of time), one of the grid points $(\theta,\phi)$ gives rise to 
    \begin{align}\label{eq:gridpoint_close}
        b-\delta \leq b_{\theta,\phi} \leq b \qquad \eta - \delta \leq \eta_{\theta,\phi} \leq \eta
    \end{align}
    Additionally, note that for any constants $b',\eta' <1$, the partial derivatives of the function $c(b',\eta') = b' F(1-\eta')/\eta'$ is $C$-Lipschitz over the input region $(b',\eta')\in [0,b] \times [0,\eta]$, for some constant $C$ which depends on $\eta$ and $b$, but not $n$. By choosing $\delta = 1/(\sqrt{2}nC)$, we can guarantee that the gridpoint $(\theta,\phi)$ that leads Eq.~\eqref{eq:gridpoint_close} to hold will have
    \begin{equation}
        c(b_{\theta,\phi}, \eta_{\theta,\phi}) \geq c(b,\eta)-1/n\,.
    \end{equation}
    Thus, this choice of $\theta$ and $\phi$ leads to an algorithm with runtime $\poly(n) 2^{(0.5-c(b,\eta)+1/n)n} = \poly(n) 2^{(0.5-c(b,\eta))n}$, where the second equality follows by absorbing the constant into the $\poly(n)$. Note that, since $\delta = 1/\poly(n)$ and $Q/q = \poly(n)$, there are at most $\poly(n)$ gridpoints. We have established that one of these gridpoints will successfully prepare $\ket{\psi_b}$ in step 2 of the algorithm (although we cannot be sure which gridpoint it was). 
    
    Now we must implement step 3 without knowing $E^*$. This is tricky because we do not know which outcome to amplify.  However, here we can do a binary search. Let us assume that we have chosen the ``correct'' gridpoint (we come back to this later). Under this assumption, we can assume that step 2 prepares $\ket{\psi_b}$. Now, we perform a computational basis measurement on $\ket{\psi_b}$ yielding outcome $z$, but we perform amplitude amplification on the outcome that $H(z) \leq U$, for some choice of $U$. If $E^* \leq U$, then the output will indeed be a string $z$ for which $H(z) \leq U$, and we will learn that $E^* \leq U$.  If the output is a string $z$ for which $H(z) > U$, this gives us high confidence that $E^* > U$ since, if not, we would have found a string for which $H(z) \leq U$ with probability at least $1-e^{-\Omega(n)}$. By choosing $U$ repeatedly via binary search, we can find a $\varepsilon/2$-close estimate for $E^*$ in $O(\log(\varepsilon^{-1}))$ time steps. (We must choose the failure probability of the algorithm small enough that all of the binary search steps succeed with high probability.) This precision is sufficiently small to exactly resolve the optimal cost of $H$, and thus we can run the algorithm with this estimate for $E^*$ and it will succeed in the stated runtime. 
    
    We do the above procedure on every gridpoint in succession. The correct gridpoint is guaranteed to output an optimal assignment with high probability. Every other gridpoint may do something different but will always output an assignment with greater cost, and thus, after we have run the algorithm for all the gridpoints, we will be able to output an optimal assignment with high probability simply by choosing the lowest-cost among all choices. 
\end{proof}

\bibliographystyle{utphys}
\bibliography{references}
\end{document}